\documentclass[12pt]{article}
\usepackage[sectionbib]{natbib}
\usepackage{fancyheadings,rotating}
\usepackage{sectsty,secdot}
\usepackage{amsmath,amssymb,amsfonts,amsthm,bm,booktabs,array,graphicx,multirow,longtable}
\usepackage{algorithm,algpseudocode}

\usepackage[hidelinks]{hyperref}
\hypersetup{pdftitle={Byzantine-Robust Sparse Regression with Missing Covariates},pdfauthor={Xiaofei Wu, Rongmei Liang and Jian Qing Shi}}
\sectionfont{\fontsize{12}{14pt plus.8pt minus .6pt}\selectfont}
\subsectionfont{\fontsize{12}{14pt plus.8pt minus .6pt}\selectfont}
\numberwithin{equation}{section}
\newtheorem{theorem}{Theorem}
\newtheorem{lemma}{Lemma}
\newtheorem{proposition}{Proposition}
\newtheorem{corollary}{Corollary}
\theoremstyle{remark}

\newcommand{\E}{\mathbb E}
\newcommand{\Pp}{\mathbb P}
\newcommand{\norm}[1]{\left\|#1\right\|}
\newcommand{\supp}{\operatorname{supp}}
\newcommand{\diag}{\operatorname{diag}}
\newcommand{\soft}{\operatorname{soft}}
\newcommand{\thstar}{\theta^\star}
\newcommand{\ghat}{\widehat g}
\newcommand{\cH}{\mathcal H}
\newcommand{\cS}{\mathcal S}
\newcommand{\cD}{\mathcal D}
\newcommand{\kktr}{\mathcal K}
\newcommand{\ind}{\mathbf 1}

\AddToHook{env/figure/begin}{\renewcommand{\baselinestretch}{1}\small\selectfont}

\begin{document}

\newcommand{\maindisplays}{\setlength{\abovedisplayskip}{6pt plus2pt minus1pt}\setlength{\belowdisplayskip}{6pt plus2pt minus1pt}\setlength{\abovedisplayshortskip}{0pt plus2pt}\setlength{\belowdisplayshortskip}{6pt plus2pt minus1pt}}
\AddToHook{cmd/normalsize/after}{\maindisplays}
\maindisplays
\renewcommand{\baselinestretch}{1.15}
\fontsize{12}{14pt plus.8pt minus .6pt}\selectfont
\markright{\hbox{\footnotesize\rm Preprint}\hfill}
\markboth{\hfill{\footnotesize\rm WU, LIANG AND SHI} \hfill}
{\hfill {\footnotesize\rm ROBUST SPARSE REGRESSION WITH MISSING DATA} \hfill}
\renewcommand{\thefootnote}{}
\setcounter{page}{1}
\vspace{0.8pc}
\begin{center}
{\large\bf Byzantine-Robust Sparse Regression\\
with Missing Covariates\par}
\vspace{.4cm}
{Xiaofei Wu$^{1}$, Rongmei Liang$^{2}$ and Jian Qing Shi$^{3*}$\par}
\vspace{.25cm}
{\it $^{1}$Yunnan University\par}
{\it $^{2}$Southern University of Science and Technology\par}
{\it $^{3}$Beijing Normal-Hong Kong Baptist University\par}
\vspace{.3cm}
\end{center}
\fontsize{9}{11.5pt plus.8pt minus .6pt}\selectfont
\begin{quotation}
\noindent {\it Abstract:}
We study distributed sparse linear regression with missing covariates and Byzantine workers. Correcting for missing observations can make the regression loss nonconvex, while Byzantine reports can distort gradient updates and impair estimation. A communication-efficient method is proposed that combines robust gradient aggregation with folded-concave regularization. Under appropriate regularity conditions, after $t$ communication rounds the estimation error is bounded by a transient component determined by the initial error and decaying geometrically with $t$, plus a residual component reflecting finite-sample variability and Byzantine contamination. Simulations and two real-data examples evaluate estimation accuracy, variable selection, and communication costs.

\vspace{9pt}
\par
\end{quotation}
\fontsize{12}{14pt plus.8pt minus .6pt}\selectfont

\def\thefigure{\arabic{figure}}
\def\thetable{\arabic{table}}
\renewcommand{\theequation}{\thesection.\arabic{equation}}

\section{Introduction}

Privacy and data-governance constraints, limited bandwidth, and large data volumes can make centralizing observations from medical institutions, sensor nodes, or mobile devices impractical~\citep{predd2006,mcmahan2017,rieke2020}. Distributed sparse regression addresses this setting by allowing each worker to retain its local observations while a trusted master receives gradient summaries and estimates a sparse coefficient vector~\citep*{jordan2019}. In these systems, covariate values may be missing because of sensor dropout, intermittent communication, skipped survey or clinical measurements, or incomplete administrative records~\citep{loh2012}. Participating nodes may also be faulty or compromised because of sensor malfunction, federated-client compromise, or corrupted institutional data pipelines. Such a Byzantine node may return arbitrary, coordinated, or adaptive messages instead of a valid local gradient~\citep{yin2018,tu2021,tu2023}. Missingness and Byzantine behavior may occur simultaneously. We study this combined setting in which a trusted master coordinates the computation without collecting the workers' raw observations.

Missing covariates and Byzantine workers create two interacting sources of error in distributed sparse regression. Missing covariates make the local products needed for least-squares estimation incomplete. Under missing completely at random (MCAR), dividing each observed product by its observation probability restores the corresponding complete-data moment in expectation~\citep{loh2012}. The correction, however, can give the quadratic loss negative curvature and make it unbounded below~\citep{loh2012}. Byzantine workers add a separate reliability problem: they may send arbitrary gradient reports~\citep{yin2018,tu2023}, so the master cannot regard every update as a valid estimate of the common gradient. Robust aggregation can reduce the influence of such reports~\citep{yin2018,tu2021}, but it does not remove the curvature induced by missingness. These two difficulties must be handled jointly: local moments must be corrected while limiting the influence of arbitrary worker reports on the master's updates.

For distributed learning with nonconvex losses, \citet*{liu2025} establish sparse optimization paths and geometric convergence under honest communication, including in missing-covariate regression. \citet*{tu2023} allow Byzantine workers in sparse M-estimation, but their method uses convex least-squares subproblems. Addressing both difficulties requires control of sparse iterates and estimation error when the loss can have negative curvature and the reported gradients can be corrupted.

We extend the communication-efficient proximal-homotopy framework of \citet*{liu2025} by combining distributed gradient correction~\citep*{jordan2019} with coordinatewise trimming~\citep{yin2018}. The master solves weighted $\ell_1$ subproblems along a decreasing penalty path, using the smoothly clipped absolute deviation (SCAD) penalty of \citet{fan2001} or the minimax concave penalty (MCP) of \citet{zhang2010}. Each round exchanges coefficient and gradient vectors; all inner updates remain at the master (Figure~\ref{fig:framework}). A uniform bound on the surrogate gradient accounts for arbitrary worker reports and differences between local curvature matrices. This bound links robust aggregation to the sparse-path argument for an indefinite loss and yields a geometrically decreasing initial-error term plus a statistical residual under the stated conditions. Section~\ref{sec:missingness} quantifies how observation probabilities and local sample sizes enter these bounds.

The remainder of the paper is organized as follows. Section~\ref{sec:related} presents related work and preliminaries. Sections~\ref{sec:algorithm} and~\ref{sec:missingness} develop the method and theory. Section~\ref{sec:experiments} reports experiments, followed by discussion in Section~\ref{sec:conclusion}. Complete proofs are given in Sections~\ref{supp-sec:mainproofs}--\ref{supp-sec:additional} of the Supplementary Material, which also contains implementation details and additional experimental results.

\begin{figure}[H]
\centering
\includegraphics[width=\textwidth,height=.72\textheight,keepaspectratio]{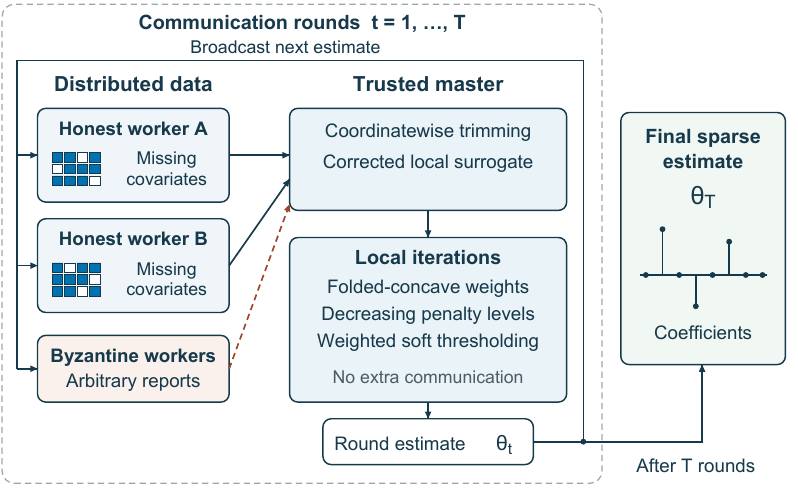}
\caption{Distributed sparse regression with missing covariates and Byzantine workers. Empty grid cells denote missing values. The master trims gradient reports and solves weighted $\ell_1$ subproblems locally. Intermediate estimates start the next round; after $T$ rounds, the procedure returns $\theta_T$. Inner iterations require no worker communication.}
\label{fig:framework}
\end{figure}

\section{Related Work and Preliminaries}\label{sec:related}
\subsection{Related Work}
\citet{loh2012} develop corrected regression for noisy or missing covariates, with statistical and optimization guarantees for constrained estimators. Their work shows that adjusting sample moments for missing observations can produce a nonconvex quadratic loss. For sparse estimation with folded-concave penalties such as SCAD and MCP, local linear approximation replaces the penalty by a weighted $\ell_1$ penalty~\citep{fan2014}, while path-following methods solve successive regularized problems as the penalty level decreases~\citep*{wang2014}. To reduce communication in distributed estimation, \citet*{jordan2019} construct local surrogate losses using gradient information from other machines, allowing local optimization steps between communication rounds. When worker reports may be arbitrary, robust aggregation limits their influence on these updates~\citep{yin2018,tu2021,zhou2021}. \citet{tu2021} also discuss adding Lasso, SCAD, or MCP penalties to a robust surrogate loss.

\citet*{liu2025} analyze weighted $\ell_1$ subproblems with decreasing penalty levels and establish sparse, geometrically convergent inner iterates under local curvature conditions along sparse directions. Their framework includes missing-covariate regression under honest communication. For convex-loss sparse M-estimation, \citet*{tu2023} use convex least-squares subproblems and establish estimation and support-recovery guarantees. \citet*{wang2026} combine local $\ell_1$-regularized robust estimation with robust server aggregation. The present setting requires sparse optimization guarantees when corrected losses may have negative curvature and worker reports may be arbitrary.

For data without attacks, \citet*{yang2025} derive worst-case lower bounds relating Byzantine tolerance to accuracy. This motivates our paired comparison of trimming levels in Section~\ref{sec:experiments}. Supplementary Table~\ref{supp-tab:position} summarizes the connections to prior work.

\subsection{Model and Attacks}\label{sec:model}\label{sec:dataattacks}
The system consists of one trusted master (node 1) and $m-1$ workers. Honest node $j$ holds $n_j$ observations; the sample sizes are fixed positive integers and may differ across nodes. Let $\cH$ and $\mathcal B$ denote the fixed sets of honest and Byzantine nodes, respectively, with $1\in\cH$. Write $m_{\mathrm H}=|\cH|$ and $m_{\mathrm B}=|\mathcal B|$, so that $m=m_{\mathrm H}+m_{\mathrm B}$ and $N_{\mathrm H}=\sum_{j\in\cH}n_j$. Thus $m_{\mathrm H}$ counts honest nodes and $N_{\mathrm H}$ counts their pooled observations. The honest set is fixed independently of the observations.

Honest observations are IID from the linear model $Y=X^\top\thstar+\varepsilon$, with $\E(X\varepsilon)=0$. To focus on slope estimation, we use a centered population model, $\E X=0$ and $\E Y=0$, so no intercept appears. Let $\cS=\supp(\thstar)$ and let $s\ge1$ be an integer upper bound on $|\cS|$. The response is fully observed. For the covariates, define the mask $M_X\in\{0,1\}^p$ by $M_{X,l}=1$ when $X_l$ is observed and $0$ otherwise. Each record contains $(Y,M_X,Z)$, where $Z=M_X\odot X$ and $\odot$ denotes coordinatewise multiplication. Thus missing entries are stored as zeros in $Z$ and identified by $M_X$. Mask coordinates are independent Bernoulli$(q)$, independent of $(X,\varepsilon)$, with known $q\in(0,1]$.

In round $t$, the master broadcasts $u_t=\theta_{t-1}$ and receives one vector from each worker. Including its own gradient, the messages are
\begin{equation*}
 v_j^{(t)}=
 \begin{cases}
  g_j(u_t), & j\in\cH,\\
  a_j^{(t)}\in\mathbb R^p, & j\in\mathcal B,
 \end{cases}
\end{equation*}
where $g_j$ is the corrected gradient defined below. Byzantine workers may coordinate and choose $a_j^{(t)}$ after observing the honest training data, current iterate and preceding messages. Their vectors have no magnitude or distribution restriction. The attack changes worker messages; the trusted master's data and computation remain intact. The master trims $b$ values from each tail of every coordinate, using a budget satisfying $m_{\mathrm B}\le b<m/2$. The main experiments use Gaussian replacement and reports chosen within the central range of honest gradients to evade trimming (Section~\ref{sec:experiments}). The Supplementary Material also examines amplified sign reversal of the honest mean gradient.

\subsection{Missingness Correction}\label{sec:correction}
At node $j$, let $Z_j\in\mathbb R^{n_j\times p}$ contain the observed covariate vectors as rows and let $Y_j\in\mathbb R^{n_j}$ collect the responses. Following the MCAR moment correction of \citet[Example~2]{loh2012}, define
\begin{equation*}
A_j=\frac{Z_j^\top Z_j}{n_j},\qquad \gamma_j=\frac{Z_j^\top Y_j}{n_jq},\qquad
\Gamma_j=\frac{A_j-(1-q)\diag(A_j)}{q^2}.
\end{equation*}
Here $\diag(A_j)$ is the diagonal matrix with the same diagonal as $A_j$. Then $\E\Gamma_j=\Sigma=\E XX^\top$ and $\E\gamma_j=\Sigma\thstar$. The corrected loss and gradient are
\begin{equation*}
 L_j(x)=\tfrac12x^\top\Gamma_jx-\gamma_j^\top x,\qquad g_j(x)=\Gamma_jx-\gamma_j.
\end{equation*}
At $q=1$, $\Gamma_j=A_j$ is positive semidefinite and $L_j$ is convex. For $q<1$, the correction can introduce negative curvature. In particular, when $p>n_j$ and $(A_j)_{ll}>0$ for every $l$, any nonzero $v\in\ker Z_j$ satisfies
\begin{equation*}
v^\top\Gamma_jv=-\frac{1-q}{q^2}\sum_l(A_j)_{ll}v_l^2<0.
\end{equation*}
Hence $L_j$ is nonconvex and $L_j(tv)\to-\infty$ as $|t|\to\infty$. Adding a weighted $\ell_1$ penalty does not bound it below, because the penalty grows only linearly along this direction. Positive curvature may still hold on small coordinate sets; the sparse-path analysis in Section~\ref{sec:algorithm} uses this property to control the optimization iterates.

\section{Algorithm and Theory}\label{sec:algorithm}
\subsection{Robust Surrogate}\label{sec:surrogate}
Fix a communication round and suppress its index, writing $u$ for the broadcast estimate. Coordinatewise trimming removes the largest and smallest $b$ messages and averages the rest, giving $\ghat(u)$ for that round's reports. Following distributed gradient correction~\citep*{jordan2019} and its robust counterpart~\citep{tu2021}, form
\begin{equation*}
\ell_u(x)=L_1(x)-\{g_1(u)-\ghat(u)\}^\top x.
\end{equation*}
This construction gives $\nabla\ell_u(u)=\ghat(u)$ while retaining $\nabla^2\ell_u(x)=\Gamma_1$. Thus workers transmit gradient vectors, and the master reuses its local quadratic curvature for all inner updates. Substituting $g_1(u)=\Gamma_1u-\gamma_1$ cancels the two linear terms involving $\gamma_1$, yielding
\begin{equation*}
 c(u)=\Gamma_1u-\ghat(u),\qquad \ell_u(x)=\tfrac12x^\top\Gamma_1x-c(u)^\top x.
\end{equation*}
Response moments enter through the reports supplied to the aggregator: each honest input is $g_j(u)=\Gamma_ju-\gamma_j$, including the master's input. Because the linear correction leaves $\Gamma_1$ unchanged, the surrogate can retain the negative curvature described in Section~\ref{sec:correction}.

To connect this surrogate to estimation error, write $g_\cH=m_{\mathrm H}^{-1}\sum_{j\in\cH}g_j$, $\Gamma_\cH=m_{\mathrm H}^{-1}\sum_{j\in\cH}\Gamma_j$, and $E(u)=\ghat(u)-g_\cH(u)$. Its gradient at the true coefficient satisfies
\begin{equation}
 \nabla\ell_u(\thstar)=g_\cH(\thstar)+E(u)+(\Gamma_\cH-\Gamma_1)(u-\thstar).
\label{eq:score}
\end{equation}
The first term is the sampling error of the honest-node mean gradient; the second is the discrepancy introduced by trimming and arbitrary reports. The third couples the current estimation error with the difference between honest-average and master curvature. The conditions in Section~\ref{sec:recursion} give \mbox{$\norm{\nabla\ell_u(\thstar)}_\infty\le A+B\norm{u-\thstar}_2$} uniformly over the relevant sparse region. This bound determines the penalty level and leads to the error recursion in Theorem~\ref{thm:outer}. Here $A$ bounds the score at $u=\thstar$, and $B$ controls propagation of the current coefficient error.

Identity~\eqref{eq:score} holds for every realized dataset and message collection, including data-dependent iterates $u=\theta_{t-1}$. The common sparse-region bound in Section~\ref{sec:recursion} therefore covers reused local data and adaptive Byzantine reports at every communication round.

Since $\cH$ is unknown, $g_\cH$ serves only as an analytical reference, with equal node weights matching the trimmed mean. Under the common distribution, $\E g_j(x)=\Sigma(x-\thstar)$ for every fixed $x$, regardless of $n_j$. Equal and sample-size weights therefore share a population target but differ in variance, quantified through the effective sample size in Section~\ref{sec:concentration}.

\subsection{Iterative Updates}\label{sec:boundpolicy}
Each communication round constructs one robust surrogate. The master updates the coefficients along a decreasing penalty path before broadcasting the new estimate. We use $t$ for communication rounds, $r$ for penalty levels within a round, and $k$ for local updates at a fixed level. The bound-based policy sets penalties and a fixed step from supplied error and curvature bounds; Theorem~\ref{thm:outer} analyzes this policy. The main performance comparisons use the data-tuned policy, which starts from a local Lasso estimate and uses validation-selected penalties and backtracking. Its schedule and stopping rules are specified in Supplementary Section~\ref{supp-sec:implementation}; Supplementary Section~\ref{supp-sec:endpoint} gives a conditional bound for its returned estimates.

At round $t$, set $u=\theta_{t-1}$ and sort coordinate $l$ of the received vectors, including the master's gradient, as $v_{(1),l}^{(t)}\le\cdots\le v_{(m),l}^{(t)}$. The aggregate is
\begin{equation*}
 [\ghat(u)]_l=\frac{1}{m-2b}\sum_{i=b+1}^{m-b}v_{(i),l}^{(t)},\qquad l=1,\ldots,p.
\end{equation*}
The resulting surrogate has gradient $\nabla\ell_u(x)=\Gamma_1(x-u)+\ghat(u)$. Thus all local updates reuse the same reports; workers send new gradients only in the next round.

We first specify the bound-based policy. Initialize $\theta_0=0$ with $\norm{\thstar}_2\le R_0$. Given a sparsity bound $s$, score bounds $A>0$, $B\ge0$, and sparse curvature bounds $0<\kappa\le M$, let $R=R_{t-1}$ be the current error bound. The target and starting penalty levels are
\begin{equation}
 \lambda_t=8(A+BR),\quad
 \Lambda_t=\max\{\lambda_t,R/(c_h\sqrt s),\kappa R/(2\sqrt s)\}.
 \label{eq:levels}
\end{equation}
Here $a$ is the penalty's shape parameter: $c_h=(a+1)/2$ for SCAD with $a>2$, and $c_h=a/2$ for MCP with $a>1$. The target dominates the surrogate score bound.

Let $p_\Lambda$ be the chosen SCAD or MCP penalty. Freeze $h_l=p'_{\Lambda_t}(|u_l|)/\Lambda_t\in[0,1]$ for this round. A coordinate is weakly penalized when $h_l<1/2$. The starting level limits such coordinates outside $\cS$ to at most $s$ and ensures $R\le2\sqrt s\Lambda_t/\kappa$, providing a sufficiently accurate start for the penalty path. Using these fixed weights, set
\begin{equation*}
 \lambda_{t,0}=\Lambda_t,\qquad
 \lambda_{t,r+1}=\max\{\lambda_t,\eta\lambda_{t,r}\},\qquad
 w_{t,r}=\lambda_{t,r}h,
\end{equation*}
where $1/2\le\eta<1$ ($12/13$ in code). Let $J_t$ be the first level reaching $\lambda_{t,J_t}=\lambda_t$. At level $r$, write $w=w_{t,r}$ and use the weighted objective $F_{u,w}(x)=\ell_u(x)+\sum_lw_l|x_l|$. These weights linearize $(\lambda_{t,r}/\Lambda_t)p_{\Lambda_t}$ at $u$; their relative values stay fixed while their common scale decreases.

Let $\widehat x_r$ denote the returned endpoint at level $r$. Start the first level at $u$ and each later level at the previous endpoint:
\begin{equation*}
 x_0^{(0)}=u,\qquad x_r^{(0)}=\widehat x_{r-1}\quad(r\ge1).
\end{equation*}
At a fixed level, suppress $t,r$ and apply
\begin{equation}
 \begin{aligned}
 z^{(k)}&=x^{(k)}-\{\Gamma_1(x^{(k)}-u)+\ghat(u)\}/M,\\
 x^{(k+1)}&=T_w(x^{(k)})=\soft_{w/M}(z^{(k)}).
 \end{aligned}\label{eq:update}
\end{equation}
The first line is a surrogate gradient step of size $1/M$. The second applies $[\soft_{w/M}(z)]_l=\operatorname{sgn}(z_l)\max\{|z_l|-w_l/M,0\}$, setting small coordinates to zero and shrinking the others. Both operations are performed at the master without communication.

The Karush--Kuhn--Tucker (KKT) residual measures the largest coordinatewise violation of first-order stationarity:
\begin{equation*}
 \kktr_w(x)=\min_{\xi\in\partial\norm{x}_1}
 \norm{\nabla\ell_u(x)+w\odot\xi}_\infty.
\end{equation*}
A nonzero coordinate requires $[\nabla\ell_u(x)]_l+w_l\operatorname{sgn}(x_l)=0$; a zero coordinate requires $|[\nabla\ell_u(x)]_l|\le w_l$. Thus the check includes inactive coordinates. Define the support bound
\begin{equation}
 K=2s+\left\lceil (32/3)^2(M/\kappa-1)^2s\right\rceil.\label{eq:K}
\end{equation}
Theorem~\ref{thm:outer} guarantees this bound without explicitly truncating the iterates. Stop the current level, with $\lambda=\lambda_{t,r}$, when
\begin{equation}
 \kktr_w(x)\le\epsilon_\lambda,\qquad 0<\epsilon_\lambda\le\min\{\lambda/26,(\lambda/4)\sqrt{s/(2K)}\}.
\label{eq:tol}
\end{equation}
Supplementary Lemma~\ref{supp-lem:path} identifies an invariant neighborhood defined by $\norm{x}_0\le K$ and $\norm{x-\thstar}_2\le4\sqrt s\lambda/\kappa$, on which $T_w$ has a unique fixed point $x_\lambda$. Within this neighborhood, the stopping rule gives $\norm{x-x_\lambda}_2\le\sqrt{2K}\epsilon_\lambda/\kappa$. The tolerance therefore limits numerical error to $\sqrt s\lambda/(4\kappa)$. At the final penalty level, set
\begin{equation*}
 \theta_t=\widehat x_{J_t},\qquad R_t=\frac{2\sqrt s\lambda_t}{\kappa}.
\end{equation*}
The final endpoint is broadcast in the next communication round. Algorithm~\ref{alg:robust} summarizes these steps.

\begin{algorithm}[!htb]
\renewcommand{\baselinestretch}{1.15}\selectfont
\caption{Robust sparse iteration with bound-based tuning}\label{alg:robust}
\begin{algorithmic}[1]
\Require Bounds $A,B,R_0,\kappa,M,s$; rounds $T$; budget $b$; ratio $\eta\in[1/2,1)$
\State Set $\theta_0=0$; compute $K$ from~\eqref{eq:K} \Comment{Master}
\State Set $\rho=16\sqrt sB/\kappa$; stop if $\rho\ge1$ \Comment{Master}
\For{$t=1,\ldots,T$}
  \State Broadcast $u=\theta_{t-1}$ \Comment{Master}
  \State Return $g_j(u)=\Gamma_ju-\gamma_j$ \Comment{Honest worker $j\ge2$}
  \State Include $g_1(u)$, trim reports, and form $c(u)$ \Comment{Master}
  \State Set $\lambda_t,\Lambda_t$ by~\eqref{eq:levels}; construct $\{\lambda_{t,r}\}_{r=0}^{J_t}$ \Comment{Master}
  \State Freeze $h_l=p'_{\Lambda_t}(|u_l|)/\Lambda_t$; set $x=u$ \Comment{Master}
  \For{$r=0,\ldots,J_t$} \Comment{Master; penalty levels}
    \State Set $\lambda=\lambda_{t,r}$, $w=\lambda h$, and start at $x$
    \State Apply~\eqref{eq:update} until~\eqref{eq:tol} holds; store the endpoint as $x$
  \EndFor
  \State Set $\theta_t=x$ and $R_t=2\sqrt{s}\lambda_t/\kappa$ \Comment{Master}
\EndFor
\State \Return $\theta_T$ \Comment{Master}
\end{algorithmic}
\end{algorithm}

The check $\rho<1$ requires small error propagation through $B$ relative to sparse curvature $\kappa$ (Section~\ref{sec:recursion}). Algorithm~\ref{alg:robust} stops if it fails; more rounds or tighter KKT tolerances do not change $\rho$.

The local matrix-vector product can be evaluated as
\begin{equation*}
 \Gamma_jx=\frac{Z_j^\top(Z_jx)}{n_jq^2}
 -\frac{1-q}{q^2}\diag(A_j)x.
\end{equation*}
This costs $O(n_jp)$ operations and $O(n_jp+p)$ storage, avoiding a $p\times p$ Gram matrix. Each round sends a $p$-vector of coefficients to each worker and receives a $p$-vector of gradients. At 8 bytes per entry (64-bit floating point), these vectors total $16(m-1)p$ bytes per round. Inner updates stay at the master. Coordinatewise sorting costs $O(pm\log m)$. 

\subsection{Sparse Convergence}\label{sec:sparsepath}\label{sec:recursion}
The corrected loss may be indefinite, but positive curvature on sparse coordinate sets can keep the updates stable. The next result combines this path control with an error bound across communication rounds. Let $\mathcal E$ be the honest-data event on which fixed nonnegative bounds $a_H,\tau,d$ satisfy
\begin{gather*}
 \norm{g_\cH(\thstar)}_\infty\le a_H,\qquad
 \max_{j\in\cH}\norm{g_j(\thstar)-g_\cH(\thstar)}_\infty\le\tau,\\
 \max_{j\in\cH}\norm{(\Gamma_j-\Gamma_\cH)e}_\infty\le d\norm{e}_2
 \quad\text{whenever }\norm{e}_0\le K+s.
\end{gather*}
These bounds control the average honest score, its variation across nodes, and curvature differences along sparse errors. Set $\omega=3b/(m-2b)$, $A=a_H+\omega\tau$, and $B=(1+\omega)d$. The trimmed-mean bound and~\eqref{eq:score} then give
\begin{equation*}
 \norm{\nabla\ell_u(\thstar)}_\infty\le A+B\norm{u-\thstar}_2
\end{equation*}
for every $K$-sparse $u$ and every permitted collection of reports. The event depends only on honest data, so the same bound applies as the iterates and attack reports change. Section~\ref{sec:concentration} gives choices with probability at least $1-\delta$.

\begin{theorem}[Sparse convergence]\label{thm:outer}
On $\mathcal E$, suppose $m_{\mathrm B}\le b<m/2$, $A>0$, and $\rho=16\sqrt s B/\kappa<1$. For $K$ in~\eqref{eq:K}, assume every principal submatrix of $\Gamma_1$ of order at most $\min(4K,p)$ has eigenvalues in $[\kappa,M]$, with $0<\kappa\le M$. Run the bound-based policy of Section~\ref{sec:boundpolicy} from $\norm{\theta_0}_0\le K$ and $\norm{\theta_0-\thstar}_2\le R_0$, using its frozen weights, penalty schedule, step $1/M$, and stopping rule~\eqref{eq:tol}. Then:
\begin{enumerate}
\setlength{\itemsep}{0pt}\setlength{\parsep}{0pt}
\item[(i)] Every inner iterate has at most $K$ nonzero coordinates. Each communication round has finitely many penalty levels, and every level reaches its prescribed KKT tolerance in finitely many updates.
\begin{samepage}
\item[(ii)] With \mbox{$R_t=16\sqrt s(A+BR_{t-1})/\kappa$}, each returned estimate satisfies $\norm{\theta_t-\thstar}_2\le R_t$ and
\begin{equation}
 \norm{\theta_t-\thstar}_2\le\rho^tR_0+
 \frac{16\sqrt s A}{\kappa(1-\rho)}(1-\rho^t).
\label{eq:recursion}
\end{equation}
\end{samepage}
\end{enumerate}
Both conclusions hold on the same event for all requested rounds and arbitrary adaptive reports from the fixed Byzantine set.
\end{theorem}

The single-stage invariance and contraction result is Supplementary Lemma~\ref{supp-lem:path}; Supplementary Section~\ref{supp-sec:outerproof} connects successive penalty levels and rounds. To interpret~\eqref{eq:recursion}, write
\begin{equation*}
 R_\infty=\frac{16\sqrt s A}{\kappa(1-\rho)},\qquad
 R_t=R_\infty+\rho^t(R_0-R_\infty).
\end{equation*}
When $R_0>R_\infty$, the excess of the error bound above its residual level contracts by $\rho$ each round. Sampling variability and Byzantine contamination determine that level. If $R_\infty\ge R_0$, the bound does not guarantee improvement over $R_0$.

To relate $\rho<1$ to the local data, let $\Delta$ bound $\max_{j\in\cH}\max_{l,k}|(\Gamma_j-\Sigma)_{lk}|$, the largest entrywise covariance error across honest nodes. One may then take $d=2\sqrt{K+s}\Delta$ and $B=(1+\omega)d$. Hence it suffices that
\begin{equation*}
 \Delta<\frac{\kappa}{32\sqrt{s(K+s)}(1+\omega)}.
\end{equation*}
Larger local samples reduce $\Delta$; Section~\ref{sec:concentration} quantifies the effects of sample sizes, observation probability, and trimming budget.

If $R_t>0$ and the minimum signal exceeds $2R_t$, thresholding gives $\{l:|\theta_{t,l}|>R_t\}=\cS$ (Supplementary Corollary~\ref{supp-cor:support}). For unthresholded support recovery, see Supplementary Proposition~\ref{supp-prop:oracle}.

\section{Statistical Accuracy}\label{sec:missingness}
Under the MCAR model, the bounds in Theorem~\ref{thm:outer} can be expressed in terms of observation probabilities and sample sizes. We first identify the variability introduced by missing covariates and then state the resulting estimation bound.

\subsection{Score Variance}
Let $\phi_l$ be coordinate $l$ of the corrected gradient at $\thstar$ for one observation. Averaging over the observation masks gives
\begin{equation*}
 \nu_l^2=\E\phi_l^2=\frac{\E(X_l^2\varepsilon^2)}q
 +\frac{1-q}{q^2}\sum_{k\ne l}(\thstar_k)^2\E(X_l^2X_k^2).
\end{equation*}
Since $\E(X\varepsilon)=0$, this is the score variance, including for correlated covariates. The first term comes from response noise; the second is the additional variability from missing predictors. Supplementary Section~\ref{supp-sec:scoreidentities} gives the calculation.

For Gaussian $X$ and independent noise of variance $\sigma^2$, the fourth moment is $\Sigma_{ll}\Sigma_{kk}+2\Sigma_{lk}^2$. With independent standardized features, the formula reduces to $\nu_l^2=\sigma^2/q+(1-q)\norm{\thstar_{-l}}_2^2/q^2$, and hence $\nu_l\le\sigma/\sqrt q+\sqrt{1-q}\norm{\thstar}_2/q$. The missing-predictor term vanishes at $q=1$ and carries the factor $(1-q)/q^2$, compared with $1/q$ for response noise. Supplementary Section~\ref{supp-sec:theoryvariance} reports a numerical check of this identity.

\subsection{Statistical Error}\label{sec:concentration}
Assume jointly sub-Gaussian $X$ with directional scale $K_X$ and sub-Gaussian noise with scale $K_\varepsilon$. Let $n_{\min}=\min_{j\in\cH}n_j$. For equally weighted honest nodes, the effective sample size is $N_{\mathrm{eff}}=m_{\mathrm H}^2/\sum_{j\in\cH}n_j^{-1}$. It equals $N_{\mathrm H}=m_{\mathrm H}n$ when all $n_j=n$, and is no greater than $N_{\mathrm H}$ otherwise.

Let $C>0$ denote an absolute constant that may vary between occurrences. For failure level $\delta\in(0,1)$, write $\nu_V=\max_l\nu_l$, $\ell_H=\log(Cp/\delta)$ and $\ell_m=\log(Cmp^2/\delta)$. The leading score contribution is
\begin{equation*}
 V_\delta=\nu_V\left\{\sqrt{\frac{\ell_H}{N_{\mathrm{eff}}}}
 +\omega\sqrt{\frac{\ell_m}{n_{\min}}}\right\}.
\end{equation*}
The two terms describe honest-node averaging and the effect of trimming. Let $Q_\delta$ collect the finite-sample tail and truncation-bias terms. Its full expression and simultaneous choices of $a_H,\tau,d,\Delta$ are given in Supplementary Section~\ref{supp-sec:statbounds}. Here $\Delta$ bounds the largest entrywise error of the honest local covariance matrices around $\Sigma$. These choices give $A\le C(V_\delta+Q_\delta)$, which supplies the statistical error term in Theorem~\ref{thm:outer}. In the corollary below, the first condition transfers population eigenvalue bounds from $\Sigma$ to sparse blocks of $\Gamma_1$, and the second ensures $\rho<1$ in the error recursion.

\begin{corollary}[Statistical error]\label{cor:statistical}
Under the IID sub-Gaussian MCAR model, use these bounds with failure level $\delta$ and $A>0$, and suppose $m_{\mathrm B}\le b<m/2$. Choose $0<\kappa<M$ and compute $K$ from~\eqref{eq:K}. Suppose every principal submatrix of $\Sigma$ of order at most $\min(4K,p)$ has eigenvalues in $[\mu,\nu_{\rm pop}]$, where $\kappa<\mu\le\nu_{\rm pop}<M$, and
\begin{equation*}
 4K\Delta\le\min\{\mu-\kappa,M-\nu_{\rm pop}\},\qquad
 \rho=\frac{32\sqrt{s(K+s)}(1+\omega)\Delta}{\kappa}<1.
\end{equation*}
For the bound-based policy initialized as in Theorem~\ref{thm:outer}, with probability at least $1-\delta$, every inner iterate is $K$-sparse and
\begin{equation*}
 \|\theta_t-\theta^\star\|_2\le\rho^tR_0+
 \frac{C\sqrt{s}(V_\delta+Q_\delta)}{\kappa(1-\rho)}(1-\rho^t).
\end{equation*}
\end{corollary}
The proof is in Supplementary Section~\ref{supp-sec:concentration}. The residual retains $Q_\delta$ and separates the benefit of pooled information through $N_{\mathrm{eff}}$ from \mbox{contamination} controlled by $n_{\min}$. The sparse curvature requirement depends on the variances of the corrected single-observation covariance products. With bounded population fourth moments, these variances are at most of order $q^{-2}$. Retaining this variance dependence gives a leading sufficient local-sample-size requirement of order $q^{-2}$ for sparse curvature, rather than the $q^{-4}$ obtained from a coarse tail bound. The full sufficient conditions involve $K^2$, logarithmic and remainder terms, and $\rho<1$.

\section{Experiments}\label{sec:experiments}
The experiments assess estimation and variable selection under joint missingness and Byzantine attacks, examine convergence conditions, and measure prediction accuracy and computational cost.

\subsection{Setup}
The main comparison uses $m=12$, $n_j=512$, $p=256$, Gaussian AR(1) covariates with correlation $0.6$, and standard Gaussian noise. Six active coefficients form two separated blocks of three adjacent predictors, with coefficient $1$ in the first block and $-1$ in the second. For each $q\in\{0.9,0.7,0.5\}$, we use one no-attack setting and four settings crossing one or two Byzantine nodes with Gaussian or in-range attacks. Additional settings vary the support layout, correlation and dimension, or use $q=0.3$ or complete observations. All 23 settings use 30 paired repetitions across methods.

We compare six corrected methods combining mean, median, trimmed or VRMOM aggregation with Lasso, SCAD or MCP penalties. VRMOM adapts \citet{tu2021} to corrected scores. Uncorrected SCAD uses the same trimming and SCAD steps with zero-filled moments. M-SLARD-ZF and T-SLARD-ZF apply the median and trimmed-mean squared-loss updates of \citet*{tu2023} to zero-filled data. All methods use eight rounds, the same solver tolerances, and a seven-point validation grid. A shared independent complete validation sample of size 400 selects among successful candidates; a separate complete test sample has size 2000. The SLARD updates use this common tuning procedure. Supplementary Section~\ref{supp-sec:competitive} gives every setting, penalty and failure rule.

Gaussian attacks replace reports by $N(0,20^2I)$ vectors; in-range attacks choose honest 30th or 70th coordinate percentiles. The budget is $b=m_{\mathrm B}$ in this comparison. Estimation errors are conditional on convergence, with success counts reported separately. Exact support uses threshold $10^{-6}$ and counts failed fits as failures; ten common thresholds assess sensitivity to this definition. Paired coefficient and test-error differences use the same simulated datasets. A sensitivity benchmark, including undertrimming, is reported in Supplementary Section~\ref{supp-sec:earlierbenchmark}.

\subsection{Synthetic Results}
At $q=0.7$ and $m_{\mathrm B}=b=2$, Trim-SCAD has mean coefficient errors $0.190$ under Gaussian attacks and $0.139$ under in-range attacks, compared with $0.253$ and $0.242$ for T-SLARD-ZF (Table~\ref{tab:competitive}). Both Trim-SCAD and Trim-MCP converge in all 30 repetitions under each attack. Their errors are close to VRMOM-SCAD and below Trim-Lasso. The corresponding paired test-MSE differences, Trim-SCAD minus T-SLARD-ZF, are $-0.0079$ and $-0.0114$, with Monte Carlo standard errors $0.0027$ and $0.0019$. Against Trim-Lasso, the reductions are $0.0437$ and $0.0329$. Thus the coefficient gains are accompanied by improved prediction on independent test observations.

\begin{table}[!htb]\centering
\caption{Coefficient error and support recovery with 30\% missingness and two Byzantine nodes among twelve. Errors are means $\pm$ Monte Carlo SE among successful fits. Conv. and Exact count converged fits and exact recoveries, both out of 30.}
\label{tab:competitive}
\par\begingroup\renewcommand{\baselinestretch}{1.2}\small\selectfont
\setlength{\tabcolsep}{3pt}\renewcommand{\arraystretch}{1.05}
\begin{tabular}{llrrlrr}\toprule
 & \multicolumn{3}{c}{Gaussian} & \multicolumn{3}{c}{In-range}\\
\cmidrule(lr){2-4}\cmidrule(lr){5-7}
Method & $\ell_2$ error & Conv. & Exact & $\ell_2$ error & Conv. & Exact\\\midrule
Trim-SCAD & $0.190\pm0.014$ & 30 & 28 & $0.139\pm0.009$ & 30 & 19 \\
Trim-MCP & $0.192\pm0.014$ & 30 & 27 & $0.139\pm0.009$ & 30 & 22 \\
Trim-Lasso & $0.266\pm0.013$ & 30 & 9 & $0.200\pm0.010$ & 30 & 5 \\
VRMOM-SCAD & $0.193\pm0.014$ & 30 & 29 & $0.139\pm0.009$ & 30 & 20 \\
Median-SCAD & $0.195\pm0.015$ & 30 & 28 & $0.156\pm0.011$ & 30 & 23 \\
Mean-SCAD & -- & 0 & 0 & $0.130\pm0.008$ & 30 & 22 \\
Uncorrected SCAD & $0.709\pm0.054$ & 30 & 30 & $0.648\pm0.052$ & 30 & 30 \\
M-SLARD-ZF & $0.258\pm0.010$ & 30 & 0 & $0.267\pm0.011$ & 30 & 0 \\
T-SLARD-ZF & $0.253\pm0.009$ & 30 & 0 & $0.242\pm0.009$ & 30 & 0 \\
\bottomrule\end{tabular}\endgroup\end{table}

The zero-filled SCAD variant recovers the support in all reference repetitions but has mean coefficient errors $0.709$ and $0.648$. The corrected method thus improves coefficient accuracy in this correlated design even when variable selection is already accurate. Supplementary Figure~\ref{supp-fig:competitivepaired} shows paired reference differences and support recovery.

\begin{figure}[htbp]
\centering\includegraphics[width=\textwidth,height=.72\textheight,keepaspectratio]{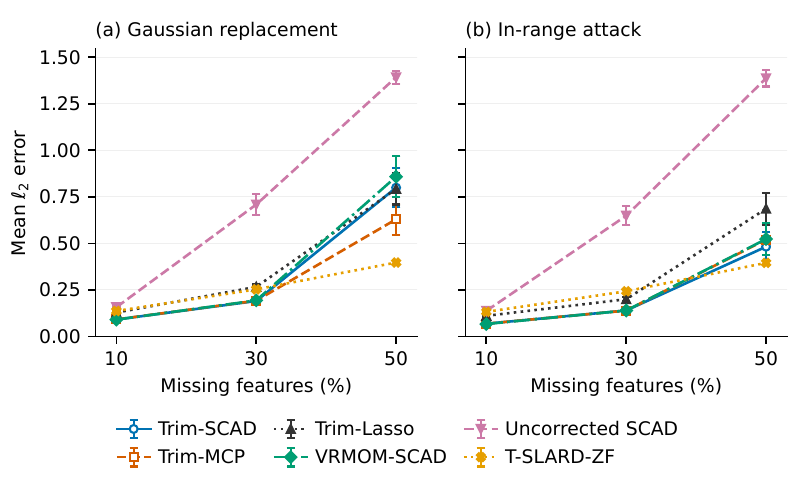}
\caption{Coefficient error with missing covariates and two Byzantine nodes among twelve. Points show means and bars one Monte Carlo standard error over 30 repetitions at each missingness level. Panels use Gaussian and in-range attacks; the complete nine-method comparison and success counts are in the supplement.}
\label{fig:competitive}
\end{figure}

Figure~\ref{fig:competitive} compares observation probabilities at fixed local sample size. At 50\% missingness, T-SLARD-ZF has lower mean coefficient and test errors than Trim-SCAD. The support comparison also depends on the threshold: at $q=0.7$, Trim-SCAD recovers 28 and 19 exact supports under the original threshold, versus none for T-SLARD-ZF, whereas both methods recover all 30 at threshold $0.2$. Supplementary Figure~\ref{supp-fig:thresholds38} reports all ten thresholds. The coefficient and prediction comparisons evaluate the original estimates and do not depend on this support definition.

A follow-up planned after the fixed-sample study uses new seeds and increases $n_j$ to $512$, $1004$ and $2788$ at $q=0.7,0.5,0.3$, respectively. Figure~\ref{fig:samples38} shows all six combinations of missingness and attack. Trim-SCAD and Trim-MCP have lower mean coefficient and prediction errors than T-SLARD-ZF throughout, with accuracy close to VRMOM-SCAD. At 50\% missingness, Trim-SCAD's coefficient errors are $0.214$ and $0.174$ under Gaussian and in-range attacks, compared with $0.310$ and $0.303$ for T-SLARD-ZF. The paired test-MSE reductions are $0.0272$ and $0.0373$, with Monte Carlo standard errors $0.0040$ and $0.0052$. Every method receives the same observations within each setting; the local sample sizes at 50\% and 70\% missingness are about $1.96$ and $5.45$ times the original size. Supplementary Section~\ref{supp-sec:competitive_sample_size} gives the rule and complete results.

\begin{figure}[H]
\centering\includegraphics[width=\textwidth,height=.72\textheight,keepaspectratio]{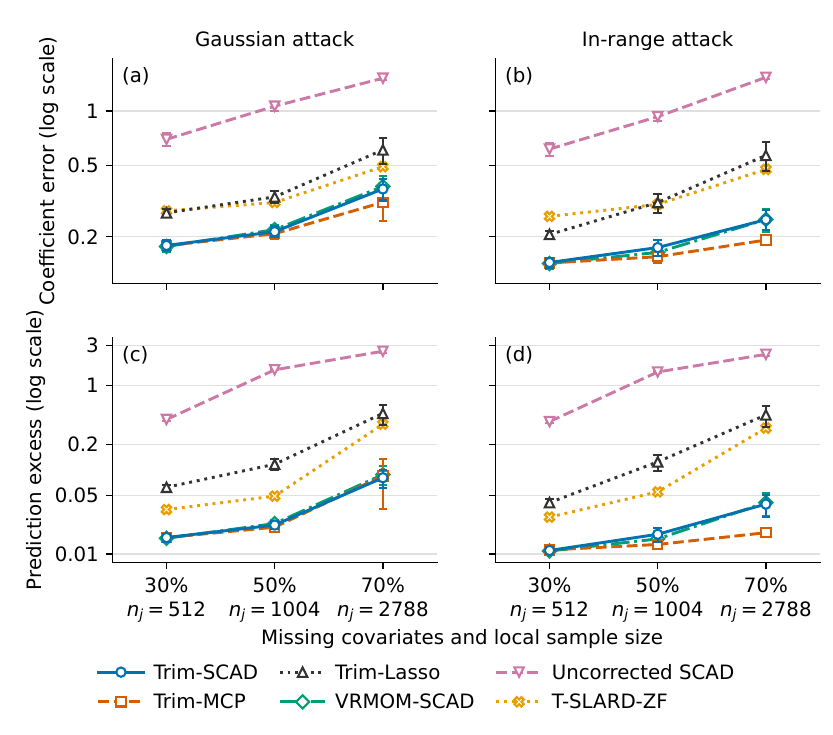}
\caption{Coefficient and prediction errors when local sample size increases with missingness. Each setting has two Byzantine nodes among twelve and 30 paired repetitions. Points show means and bars one Monte Carlo standard error. Vertical axes use log scales. Prediction excess is the squared prediction error averaged over the population covariate distribution, excluding response noise. All six displayed methods succeed in every repetition. The horizontal labels give missingness and observations per node.}
\label{fig:samples38}
\end{figure}

Sensitivity checks vary dimension and the trimming budget. At $q=0.7$ and $m_{\mathrm B}=b=2$, Trim-SCAD remains competitive at $p=768>n_j=512$, with errors $0.186$ and $0.137$ under Gaussian and in-range attacks, respectively. Undertrimming ($m_{\mathrm B}=4,b=2$) makes every trimmed candidate fail in the sensitivity experiment, while median and VRMOM methods remain stable. On 30 paired no-attack datasets, budgets $b=0,2,4$ give mean errors $0.0948$, $0.1007$ and $0.1011$, with all fits converging (Supplementary Table~\ref{supp-tab:cleanbudget}).

A separate allocation experiment distributes 2400 observations across eight nodes and varies the ratio of the largest to smallest assigned sample size. When the two smallest remaining workers are Byzantine, Trim-SCAD's mean errors are $0.083$ and $0.177$ for the largest and smallest master blocks, respectively, at ratio 4, and $0.107$ and $0.988$ at ratio 16. Honest sample totals also change with allocation and attack placement. Supplementary Figure~\ref{supp-fig:mainimbalance} reports all placements.

\subsection{Theory Checks}
We examine Theorem~\ref{thm:outer} in a separate Gaussian MCAR study with $p=24,s=2,m=8$ and noise standard deviation $0.5$. Simulation truth, honest-node identities, and the master's full spectrum provide valid inputs for the bound-based policy. At $q=0.7$ and $m_{\mathrm B}=b=1$, $\rho<1$ holds in $0/20$, $4/20$, and $20/20$ samples for $n_j=20000,80000,320000$. At $n_j=320000$, all inner supports have size at most two ($K=5$), all roundwise error bounds hold, and all 20 fits satisfy the support-recovery condition and recover the true support. Mean error is $0.00108$, while the median final bound is $0.247$ (Figure~\ref{fig:theoryvalidation}).

A separate $p=4000$ quadratic construction has heterogeneous honest gradients and a negative curvature direction, but all principal blocks of order at most $20$ have eigenvalues in $[0.990,1.004]$. With $K=5$, all 140 runs satisfy the sparse-path, inner-contraction, and outer-error checks under clean, Gaussian, and adaptive in-range reports. Supplementary Section~\ref{supp-sec:theoryvalidation} gives the constructions, all prespecified samples, and variations in missingness and attack budget.

\begin{figure}[!htb]
\centering\includegraphics[width=\textwidth,height=.72\textheight,keepaspectratio]{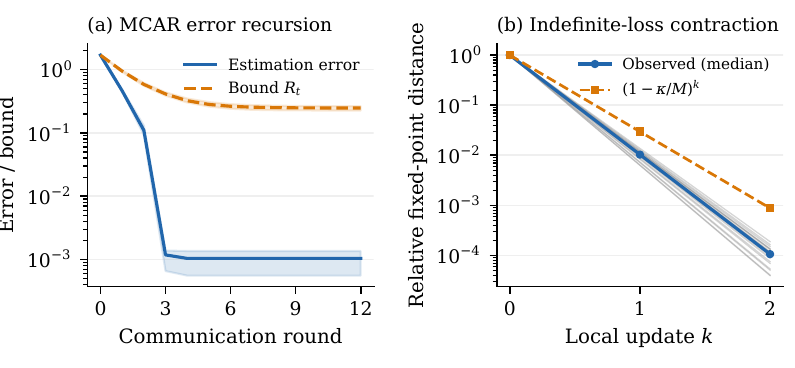}
\caption{Checks of the conditional theory, with 20 repetitions per panel. (a) MCAR estimation errors and bounds at $n_j=320000,q=0.7,b=m_{\mathrm B}=1$: medians and interquartile ranges. (b) Relative distance to the exact fixed point at the final penalty level of the first round in the indefinite construction: all paths (gray), their median, and the contraction envelope. Both panels use adaptive in-range attacks.}
\label{fig:theoryvalidation}
\end{figure}

\subsection{Paths and Ablations}\label{sec:paths}
For the data-tuned comparisons, we evaluate the conditional endpoint bound after fitting, using saved estimates and reconstructed moments (Supplementary Theorem~\ref{supp-thm:endpoint}). For each returned estimate, endpoint curvature is the smallest eigenvalue of $\Gamma_1$ restricted to the union of its support and the true support. Across the 1380 selected Trim-SCAD/MCP paths in the fixed-sample study, 1329 have positive curvature at every endpoint and 97 satisfy the bound's contraction condition throughout. At the $q=0.7,m_{\rm B}=2$ reference, all 120 method--attack paths have positive endpoint curvature, but none satisfies this contraction condition. The shared 30 master samples also fail the uniform sparse-curvature requirement. Thus neither the uniform sparse-path guarantee nor the sufficient endpoint contraction condition covers this reference comparison. Supplementary Section~\ref{supp-sec:competitive_endpoints} reports every setting, residual checks, and both satisfied and unsatisfied conditions.

A separate $p=768$ path study examines sample size and the role of intermediate penalty levels. Endpoint contraction holds throughout eight rounds in all ten fits at $n_j=512$, versus five at $n_j=256$. The direct-target variant skips intermediate penalty levels while retaining the roundwise surrogate and frozen-weight rule. Its selected coefficient error differs from continuation by less than $3\times10^{-5}$ across all 60 paired fits. Supplementary Section~\ref{supp-sec:sampledpaths} gives the designs and component comparisons. A separate bounded-design experiment uses independent trusted observations to calibrate error radii for fitted candidates (Supplementary Section~\ref{supp-sec:trustedcomparison}).

\subsection{Real Data}
The Superconductivity data~\citep{uci} contain 21263 observations and 81 features. Ten splits use 14000 training, 3000 validation and 4263 test rows. Training data determine centering and scaling before MCAR masking ($q=0.7$); validation and test covariates remain complete. Training is split across 20 nodes, with a common four-point tuning grid.

Under 10\% Gaussian attacks, Trim-SCAD, Trim-MCP and VRMOM-SCAD have mean RMSEs $21.166$, $21.032$ and $21.003$, below the training-mean baseline $34.275$. Mean-SCAD fails under Gaussian attacks; on clean data, it and Trim-SCAD attain about $21.125$. Supplementary Section~\ref{supp-sec:earlierreal} reports all methods and success counts.

We also predict Abalone ring count from seven measurements~\citep{abalone} and split superconductivity by element-composition groups. Across ten shared splits, Trim-SCAD, Trim-Lasso and VRMOM-SCAD succeed in 8, 10 and 9 Abalone fits and 10, 9 and 10 grouped superconductivity fits. On jointly successful splits, Trim-SCAD minus VRMOM-SCAD has mean RMSE differences $-0.0006$ and $-0.0065$. The supplement gives all conditional RMSEs and paired variability.

\subsection{Runtime and Communication}
A new process study measures complete validation tuning for Trim-SCAD, Trim-MCP, Trim-Lasso, VRMOM-SCAD and T-SLARD-ZF. It uses the primary design at $q=0.7,0.5$, two Gaussian or in-range attackers, and 30 new paired seeds per setting. Timing includes method-specific moment construction and all seven candidates, including failures, through validation and selection; shared data preparation and worker startup are recorded separately. At $q=0.7$, Trim-SCAD's mean coefficient errors are $0.166$ and $0.143$ under Gaussian and in-range attacks, close to VRMOM-SCAD's $0.168$ and $0.137$ and below T-SLARD-ZF's $0.255$ and $0.247$. Median tuning times are $0.138$ and $0.133$ seconds for Trim-SCAD, versus $0.236$ and $0.235$ for VRMOM-SCAD; both have a median vector payload of $1.802$ MB. Trim-SCAD is faster than VRMOM-SCAD in every paired run across all four settings. Supplementary Section~\ref{supp-sec:fulltuning38} summarizes all 600 runs and gives the timing protocol.

A separate dimension study uses coordinator and worker processes on the same Windows computer, communicating through pipes. With $m=8,n=1024$, $p\in\{1000,4000,10000\}$ and $q\in\{0.5,0.7,0.9\}$, Trim-SCAD, Trim-Lasso and VRMOM-SCAD give 270 runs across ten seeds and 25\% Gaussian attacks. Methods share a penalty multiplier and run in randomized order. Timing includes initialization, communication, aggregation and inner optimization; setup and memory are reported separately in the supplement.

We check coefficient-error targets $0.1$ and $0.2$ at initialization and at each successful penalty-stage endpoint. For each target, time and bytes are measured until the first such point reaching it, and cost medians condition on attainment. Targets not reached within eight rounds are censored; later solver failures are recorded separately. Truth is used only for this evaluation. At $q=0.7,p=10000$, Trim-SCAD and VRMOM-SCAD reach error $0.1$ in 10/10 runs, with median times $0.503$ and $0.754$ seconds and payload $1.12$ MB. At $q=0.5,p=10000$, neither reaches $0.1$, and final solver successes fall to 1/10 and 0/10. The supplement reports all configurations and both targets.

\section{Discussion}\label{sec:conclusion}
Uniform control of Byzantine gradient perturbations supports a sparse-path guarantee for corrected regression losses. The admitted MCAR runs and the indefinite-quadratic experiments satisfy the conditional guarantees, whereas the reference comparison does not meet their sufficient conditions. In correlated designs with missing covariates and Byzantine nodes, Trim-SCAD and Trim-MCP improve on corrected Lasso and zero-filled SLARD at moderate missingness, with accuracy close to VRMOM-SCAD. Larger local samples improve this comparison at greater missingness. Full validation-tuning measurements show lower full-tuning time than VRMOM-SCAD with the same median vector payload in the reference settings. Correction need not improve finite-sample accuracy with independent predictors. Prediction on two public datasets is comparable to robust alternatives.

The method uses a corrected quadratic loss, with complete-data regression recovered at $q=1$. The observation probability is treated as known; estimating it from masks requires additional concentration arguments. A conservative Byzantine budget increases the bound's residual level; the no-attack comparison measures its empirical cost. Coordinate-specific MCAR probabilities and additive measurement error admit corresponding moment corrections, detailed in the supplement. Generalized linear models require separate loss corrections; missing-at-random and missing-not-at-random mechanisms require corresponding observation models~\citep{little2019}. Sharper path conditions, calibration with fewer trusted observations, and naturally missing data remain directions for further work.

\section{Supplementary Material}
The Supplementary Material contains proofs, statistical calculations, implementation details, and additional results. The accompanying reproducibility archive provides analysis code, frozen results, seeds, and data provenance. The public datasets~\citep{abalone,uci} are redistributed with attribution under their licenses.

\section*{Acknowledgments}
This work was supported in part by the National Key R\&D Program of China (Grant No. 2023YFA1011400), the Guangdong Higher Education Upgrading Plan (Grant No. 0400008-26), and the BNBU Research Grant (Grant No. R0700165-27).

OpenAI Codex assisted with text, proof drafts, experimental code, and figure-generation code. All authors have reviewed and approved the final manuscript.

\setlength{\bibsep}{2pt}

\vskip .65cm
\noindent\begin{minipage}{\textwidth}
\noindent Xiaofei Wu. Yunnan Key Laboratory of Statistical Modeling and Data Analysis, School of Mathematics and Statistics, Yunnan University, Kunming, China.
\par\noindent E-mail: \href{mailto:xfwu1016@ynu.edu.cn}{xfwu1016@ynu.edu.cn}
\par\medskip
\noindent Rongmei Liang. Department of Statistics and Data Science, Southern University of Science and Technology, Shenzhen 518055, China.
\par\noindent E-mail: \href{mailto:liang_r_m@163.com}{liang\_r\_m@163.com}
\par\medskip
\noindent Jian Qing Shi (corresponding author). Guangdong Provincial/Zhuhai Key Laboratory of Interdisciplinary Research and Application for Data Science, Beijing Normal-Hong Kong Baptist University, Zhuhai, 519087, China.
\par\noindent E-mail: \href{mailto:jianqingshi@bnbu.edu.cn}{jianqingshi@bnbu.edu.cn}
\end{minipage}

\clearpage
\setcounter{secnumdepth}{2}
\setcounter{section}{0}
\setcounter{subsection}{0}
\setcounter{equation}{0}
\setcounter{figure}{0}
\setcounter{table}{0}
\setcounter{theorem}{0}
\setcounter{lemma}{0}
\setcounter{proposition}{0}
\setcounter{corollary}{0}
\setcounter{remark}{0}
\setcounter{algorithm}{0}
\providecommand{\theHalgorithm}{}
\renewcommand{\theHsection}{supp.\arabic{section}}
\renewcommand{\theHfigure}{supp.\arabic{figure}}
\renewcommand{\theHtable}{supp.\arabic{table}}
\renewcommand{\theHtheorem}{supp.\arabic{theorem}}
\renewcommand{\theHlemma}{supp.\arabic{lemma}}
\renewcommand{\theHproposition}{supp.\arabic{proposition}}
\renewcommand{\theHcorollary}{supp.\arabic{corollary}}
\renewcommand{\theHremark}{supp.\arabic{remark}}
\renewcommand{\theHalgorithm}{supp.\arabic{algorithm}}
\renewcommand{\theHequation}{supp.\arabic{section}.\arabic{equation}}

\makeatletter
\newcommand{\compacttable}{\renewcommand{\baselinestretch}{1}\ifdim\f@size pt>11pt\small\fi\selectfont\renewcommand{\arraystretch}{1.08}}
\makeatother
\makeatletter
\renewcommand\section{\@startsection{section}{1}{\z@}{-2.3ex plus-.5ex minus-.2ex}{1.2ex plus.2ex}{\normalfont\fontsize{12}{14pt}\selectfont\bfseries}}
\renewcommand\subsection{\@startsection{subsection}{2}{\z@}{-1.8ex plus-.4ex minus-.2ex}{.8ex plus.2ex}{\normalfont\fontsize{12}{14pt}\selectfont\bfseries}}
\makeatother
\AddToHook{env/table/before}{\begingroup\compacttable}
\AddToHook{env/table/after}{\endgroup}
\AddToHook{env/figure/before}{\begingroup\renewcommand{\baselinestretch}{1}\small\selectfont}
\AddToHook{env/figure/after}{\endgroup}
\AddToHook{env/tabular/before}{\begingroup\compacttable}
\AddToHook{env/tabular/after}{\endgroup}
\AddToHook{env/longtable/before}{\begingroup\compacttable}
\AddToHook{env/longtable/after}{\endgroup}
\subsubsectionfont{\fontsize{12}{14pt}\selectfont}

\raggedbottom
\setlength{\intextsep}{6pt plus2pt minus2pt}
\setlength{\textfloatsep}{10pt plus2pt minus2pt}
\setlength{\floatsep}{8pt plus2pt minus2pt}
\newcommand{\suppdisplays}{\setlength{\abovedisplayskip}{6pt plus2pt minus2pt}\setlength{\belowdisplayskip}{6pt plus2pt minus2pt}\setlength{\abovedisplayshortskip}{3pt plus1pt}\setlength{\belowdisplayshortskip}{6pt plus2pt minus2pt}}
\AddToHook{cmd/normalsize/after}{\suppdisplays}
\suppdisplays
\renewcommand{\baselinestretch}{1.15}
\fontsize{12}{14pt plus.8pt minus .6pt}\selectfont
\markright{\hbox{\footnotesize\rm Supplementary Material}\hfill}
\markboth{\hfill{\footnotesize\rm WU, LIANG AND SHI} \hfill}
{\hfill {\footnotesize\rm ROBUST SPARSE REGRESSION WITH MISSING DATA} \hfill}
\renewcommand{\thefootnote}{}

\vspace{0.8pc}
\begingroup\renewcommand{\baselinestretch}{1}\selectfont
\begin{center}
{\large\bf Byzantine-Robust Sparse Regression\\
with Missing Covariates\par}
\vspace{.4cm}
{Xiaofei Wu$^{1}$, Rongmei Liang$^{2}$ and Jian Qing Shi$^{3*}$\par}
\vspace{.25cm}
{\it $^{1}$Yunnan University\par}
{\it $^{2}$Southern University of Science and Technology\par}
{\it $^{3}$Beijing Normal-Hong Kong Baptist University\par}
\vspace{.3cm}
{\bf Supplementary Material\par}
\end{center}\endgroup
\noindent This supplement gives complete proofs, implementation details, and experimental designs and results. Ancillary tables retain the exhaustive tuning, paired-comparison and cost listings.
\par
\setcounter{section}{0}
\setcounter{equation}{0}
\renewcommand{\theequation}{S\arabic{section}.\arabic{equation}}
\renewcommand{\thesection}{S\arabic{section}}
\renewcommand{\thetheorem}{S\arabic{theorem}}
\renewcommand{\thelemma}{S\arabic{lemma}}
\renewcommand{\theproposition}{S\arabic{proposition}}
\renewcommand{\thecorollary}{S\arabic{corollary}}
\fontsize{12}{14pt plus.8pt minus .6pt}\selectfont

\section*{Contents}
\begingroup\renewcommand{\baselinestretch}{1}\small\selectfont
\noindent\hyperref[supp-contents-1]{S1.\quad Guide and Notation}\dotfill\pageref{supp-contents-1}\par
\noindent\hyperref[supp-contents-2]{S2.\quad Proofs of Main Results}\dotfill\pageref{supp-contents-2}\par
\noindent\hyperref[supp-contents-3]{S3.\quad Additional Results}\dotfill\pageref{supp-contents-3}\par
\noindent\hyperref[supp-contents-4]{S4.\quad Implementation}\dotfill\pageref{supp-contents-4}\par
\noindent\hyperref[supp-contents-5]{S5.\quad Joint Missingness and Attacks}\dotfill\pageref{supp-contents-5}\par
\noindent\hyperref[supp-contents-6]{S6.\quad Sensitivity Benchmarks}\dotfill\pageref{supp-contents-6}\par
\noindent\hyperref[supp-contents-7]{S7.\quad Theory Validation}\dotfill\pageref{supp-contents-7}\par
\noindent\hyperref[supp-contents-8]{S8.\quad Additional Checks}\dotfill\pageref{supp-contents-8}\par
\noindent\hyperref[supp-contents-9]{S9.\quad Additional Real Data}\dotfill\pageref{supp-contents-9}\par
\noindent\hyperref[supp-contents-10]{S10.\quad Computational Cost}\dotfill\pageref{supp-contents-10}\par
\noindent\hyperref[supp-contents-11]{S11.\quad Path and Component Experiments}\dotfill\pageref{supp-contents-11}\par
\noindent\hyperref[supp-contents-12]{S12.\quad Node-Size Imbalance}\dotfill\pageref{supp-contents-12}\par
\noindent\hyperref[supp-contents-13]{S13.\quad Trusted-Sample Comparison}\dotfill\pageref{supp-contents-13}\par
\noindent\hyperref[supp-contents-14]{S14.\quad Model Extensions}\dotfill\pageref{supp-contents-14}\par
\noindent\hyperref[supp-contents-15]{S15.\quad Reproducibility}\dotfill\pageref{supp-contents-15}\par
\endgroup\clearpage

\section{Guide and Notation}\label{supp-contents-1}\label{supp-sec:guide}
Tables~\ref{supp-tab:position}--\ref{supp-tab:guarantees} locate the contributions and guarantees. Sections~\ref{supp-sec:mainproofs}--\ref{supp-sec:additional} give proofs; implementation and empirical studies follow.

\begin{table}[H]
\caption{Related results and the extension studied here.}
\label{supp-tab:position}\centering\renewcommand{\baselinestretch}{1}\footnotesize\selectfont\setlength{\tabcolsep}{3pt}
\begin{tabular}{>{\raggedright\arraybackslash}p{.28\textwidth}>{\raggedright\arraybackslash}p{.29\textwidth}>{\raggedright\arraybackslash}p{.35\textwidth}}\toprule
Source & Contribution & Present extension\\\midrule
\citet{suppbib-loh2012} & Corrected regression with missing covariates & Byzantine-robust sparse path for corrected losses.\\
\citet{suppbib-liu2025} & Weighted homotopy; sparse iterates and geometric convergence & Uniform score control under adaptive Byzantine reports.\\
Tu et al.~\citeyearpar{suppbib-tu2021} & Lasso/SCAD/MCP surrogate regularization suggested; theory left open & Sparse-path guarantees for corrected quadratic losses.\\
\citet{suppbib-tu2023} & Byzantine sparse M-estimation with convex subproblems & Indefinite MCAR-corrected loss and missingness-dependent error.\\\bottomrule
\end{tabular}
\end{table}

\begin{table}[H]
\caption{Guarantees, required information and corresponding evidence.}
\label{supp-tab:guarantees}\centering\renewcommand{\baselinestretch}{1}\footnotesize\selectfont\setlength{\tabcolsep}{3pt}
\begin{tabular}{>{\raggedright\arraybackslash}p{.16\textwidth}>{\raggedright\arraybackslash}p{.28\textwidth}>{\raggedright\arraybackslash}p{.25\textwidth}>{\raggedright\arraybackslash}p{.17\textwidth}}\toprule
Result & Conditions or supplied information & What is controlled & Evidence\\\midrule
Uniform path & Valid $A,B,R_0,\kappa,M,s$; $4K$-sparse curvature; $\rho<1$ & Every inner support and outer error & Constructed quadratic and MCAR bound-based audits\\
Endpoint recursion & Positive curvature on $\cS\cup\supp(x_t)$; full KKT residual & Returned errors; geometric bound if $\alpha_t\le\bar\alpha<1$ and $\zeta_t\le\bar\zeta$ & All selected competitive paths; sampled $p>n_1$ paths\\
Local check & Observed active block, weights and residual & Continuation of one fixed subproblem & Endpoint checks\\
Masked calibration & Independent trusted sample; known bounds and sparsity & Simultaneous coefficient-error radii & Bounded-design calibration\\\bottomrule
\end{tabular}
\end{table}

\subsection{Notation}\label{supp-sec:notation}
We use the model and attack assumptions of main-text Section~\ref{sec:model}: fixed honest and Byzantine sets $\cH,\mathcal B$, trusted master $1\in\cH$, budget $m_{\mathrm B}\le b<m/2$, IID honest observations, and arbitrary coordinated reports depending on honest data, current iterates and message history. Masks have independent Bernoulli$(q)$ coordinates, independent of $(X,\varepsilon)$, with known $q>0$; $Z=M_X\odot X$. The centered model satisfies $Y=X^\top\thstar+\varepsilon$ and $\E(X\varepsilon)=0$.

Write $\cS=\supp(\thstar)$, $|\cS|\le s$, $\norm{G}_{\max}=\max_{lk}|G_{lk}|$ and $G_{UU}$ for a principal submatrix; vector norms have their usual meanings. For fixed honest-node sizes $n_j>0$, use $N_{\mathrm H}=\sum_{j\in\cH}n_j$, $n_{\min}=\min_{j\in\cH}n_j$, $n_{\max}=\max_{j\in\cH}n_j$ and $N_{\mathrm{eff}}=m_{\mathrm H}^2/\sum_{j\in\cH}n_j^{-1}$. Reserve $M$ for the sparse eigenvalue upper bound. Concentration assumes $\sup_{\norm v_2=1}\norm{v^\top X}_{\psi_2}\le K_X$ and $\norm\varepsilon_{\psi_2}\le K_\varepsilon$. Noise independence is used only for the simplified Gaussian variance formula; a supplied valid $s$ is needed for bound-based fitting and calibration, not data-tuned fitting.

\section{Proofs of Main Results}\label{supp-contents-2}\label{supp-sec:mainproofs}
\subsection{Score Identities}\label{supp-sec:scoreidentities}
The corrected moments follow \citet[Example~2]{suppbib-loh2012}. For one observation, write $G=q^{-2}\{ZZ^\top-(1-q)\diag(ZZ^\top)\}$ and $d=ZY/q$.
For $l\ne k$, $G_{lk}=M_{X,l}M_{X,k}X_lX_k/q^2$; $G_{ll}=M_{X,l}X_l^2/q$. Averaging over masks gives $\E(G\mid X)=XX^\top$ and $\E(d\mid X,Y)=XY$. Expanding the score,
\begin{align*}
(G\thstar-d)_l
 &=\frac{M_{X,l}}qX_l\left\{X_l\thstar_l+
 \sum_{k\ne l}\frac{M_{X,k}}qX_k\thstar_k-Y\right\}\\
 &=\frac{M_{X,l}}qX_l\left\{\sum_{k\ne l}(M_{X,k}/q-1)X_k\thstar_k-\varepsilon\right\}.
\end{align*}
Conditionally on $(X,\varepsilon)$, the off-diagonal mask fluctuations are independent, centered, and have variance $(1-q)/q$. Independence of $M_{X,l}$ gives
\begin{equation*}
 \E(\phi_l^2\mid X,\varepsilon)=\frac{X_l^2\varepsilon^2}{q}
 +\frac{1-q}{q^2}\sum_{k\ne l}(\thstar_k)^2X_l^2X_k^2.
\end{equation*}
The score has mean $-\E X_l\varepsilon=0$, proving the main-text variance formula. For Gaussian $X$, $\E X_l^2X_k^2=\Sigma_{ll}\Sigma_{kk}+2\Sigma_{lk}^2$. Hessian variances follow from $\E G_{lk}^2=q^{-2}\E X_l^2X_k^2$ off-diagonal and $\E G_{ll}^2=q^{-1}\E X_l^4$ on-diagonal.

A numerical check enumerates all $2^4$ masks at five observation rates and verifies this conditional second-moment identity.

\subsection{Trimmed-Mean Bound}
\begin{lemma}\label{supp-lem:trim}
Let $m$ scalar messages contain $m_{\mathrm H}$ honest values and $m_{\mathrm B}$ arbitrary values with honest mean $\bar z_H$, where $m=m_{\mathrm H}+m_{\mathrm B}$, and suppose $m_{\mathrm B}\le b<m/2$. Let $T_b$ delete the smallest and largest $b$ values and average the remaining $m-2b$ values. Then
\begin{equation*}
 |T_b-\bar z_H|\le\frac{3b}{m-2b}\max_{j\in H}|z_j-\bar z_H|.
\end{equation*}
\end{lemma}
\begin{proof}
Since $m_{\mathrm B}\le b$, retained malicious values lie in the honest range. If $r$ remain, $2b-m_{\mathrm B}+r$ honest values were deleted. After centering at $\bar z_H$, retained and deleted honest sums are opposites. The absolute retained total is therefore at most $2b-m_{\mathrm B}+2r\le3b$ times the largest honest deviation. Divide by $m-2b$; equality holds for $b=m_{\mathrm B}=0$.
\end{proof}
The coordinatewise bound holds deterministically for every adaptive message realization.

\subsection{Sparse Path}\label{supp-sec:sparseproof}
\begin{lemma}[Single-stage sparse path]\label{supp-lem:path}
Let $K$ be defined by main-text~\eqref{eq:K}, and assume every principal submatrix of $\Gamma_1$ of order at most $\min(4K,p)$ has eigenvalues in $[\kappa,M]$, with $0<\kappa\le M$. Fix $\lambda>0$ and a quadratic surrogate with score $v=\nabla\ell(\thstar)$ satisfying $\norm{v}_\infty\le\lambda/8$. Suppose $0\le w_l\le\lambda$, and an exceptional set $E\supseteq\cS$ satisfies $|E|\le2s$ and $w_l\ge\lambda/2$ for $l\notin E$. Define
\begin{equation*}
 \cD_\lambda=\{x:\norm{x}_0\le K,\ \norm{x-\thstar}_2\le4\sqrt s\lambda/\kappa\}.
\end{equation*}
Then the following hold.
\begin{enumerate}
\setlength{\itemsep}{0pt}\setlength{\parsep}{0pt}
\item[(i)] \textit{Sparse invariance.} The map $T_w$ in~\eqref{eq:update} sends $\cD_\lambda$ into itself. Every iterate started there has at most $K$ nonzero coordinates.
\item[(ii)] \textit{Contraction and accuracy.} On $\cD_\lambda$, $T_w$ contracts with modulus $a_0=1-\kappa/M$ to a unique fixed point $x_\lambda$, satisfying \mbox{$\norm{x_\lambda-\thstar}_2\le(9\sqrt2/8)\sqrt s\lambda/\kappa$}.
\item[(iii)] \textit{Finite stopping.} Iteration from $\cD_\lambda$ reaches any positive KKT tolerance in finitely many steps. Every point in $\cD_\lambda$ satisfying~\eqref{eq:tol} has error less than $2\sqrt s\lambda/\kappa$.
\end{enumerate}
\end{lemma}

For the proof, write $G=\Gamma_1$, $v=G\thstar-c$, and $C_K=(32/3)^2(M/\kappa-1)^2s$, so that $K=2s+\lceil C_K\rceil$. Denote the update map by $T$.

\subsubsection{Support Invariance}
Fix $x\in\cD_\lambda$ and $e=x-\thstar$. If $\norm{Tx}_0>K$, choose $J\subseteq\supp(Tx)\setminus E$ with $|J|=K-2s+1$. Since $\thstar_J=0$, activation implies
\begin{equation*}
 \frac{\lambda}{2M}\sqrt{|J|}
 <\norm{[(I-G/M)e]_J-v_J/M}_2.
\end{equation*}
Set $U=J\cup\supp(x)\cup\cS$. Its cardinality is at most
$K-2s+1+K+s=2K-s+1\le4K$. The sparse eigenvalue assumption implies
\begin{equation*}
\begin{aligned}
 \norm{[(I-G/M)e]_J}_2
 &\le\norm{I_U-G_{UU}/M}_{\mathrm{op}}\norm{e}_2\\
 &\le(1-\kappa/M)\norm{e}_2.
\end{aligned}
\end{equation*}
Thus $(3\lambda/8)\sqrt{|J|}<(M-\kappa)\norm{e}_2\le4(M-\kappa)\sqrt s\lambda/\kappa$, giving $|J|<C_K$, contrary to $|J|=\lceil C_K\rceil+1$. This proves output sparsity without assuming it. For $K\ge p$, sparsity is trivial, but the following arguments require full positive definiteness.

\subsubsection{Radius Invariance}
For $l\notin E$, $\thstar_l=0$ and $|v_l|\le\lambda/8<w_l$, so $T(\thstar)_l=0$. On $E$, soft-threshold displacement is bounded by input displacement plus threshold. Hence
\begin{equation*}
 \norm{T(\thstar)-\thstar}_2
 \le\sqrt{2s}\frac{\lambda+\lambda/8}{M}
 =c_0\frac{\sqrt s\lambda}{M},\qquad c_0=9\sqrt2/8.
\end{equation*}
The union of the output supports of $T(x),T(\thstar)$ and $\supp(x)\cup\cS$ has size at most $2K+3s\le4K$, since $K\ge2s$. Nonexpansiveness gives
\begin{equation*}
\begin{aligned}
 \norm{T(x)-\thstar}_2
 &\le(1-\kappa/M)\norm{x-\thstar}_2+c_0\sqrt s\lambda/M\\
 &\le4\sqrt s\lambda/\kappa.
\end{aligned}
\end{equation*}
The last step uses $c_0<4$; support invariance now gives $T(\cD_\lambda)\subseteq\cD_\lambda$.

\subsubsection{Contraction}
For $x,y\in\cD_\lambda$, the union $U$ of their input and output supports has $|U|\le4K$. Since $x-y$ is supported in $U$, nonexpansiveness and sparse curvature imply
\begin{equation*}
 \norm{T(x)-T(y)}_2\le(1-\kappa/M)\norm{x-y}_2.
\end{equation*}
The set $\cD_\lambda$, a ball intersected with finitely many coordinate subspaces, is closed. Contraction gives a unique fixed point, convergence throughout the set, and, by the radius recursion, $\norm{x_\lambda-\thstar}_2\le c_0\sqrt s\lambda/\kappa$.

Soft-threshold fixed points are equivalent to $Gx_\lambda-c+w\odot\xi_\lambda=0$ for some $\xi_\lambda\in\partial\norm{x_\lambda}_1$, establishing stationarity in the sparse neighborhood.

\subsubsection{Stopping Rule}
For an $\epsilon$-KKT point $x\in\cD_\lambda$, choose $r,\xi$ with $Gx-c+w\odot\xi=r$ and $\norm{r}_\infty\le\epsilon$. Subtract the fixed-point relation and take the inner product with $d=x-x_\lambda$. Convexity of the weighted absolute-value function gives subgradient monotonicity; $d$ has at most $2K$ nonzeros. Thus
\begin{equation*}
 \kappa\norm{d}_2^2\le d^\top r\le\norm{d}_1\epsilon\le\sqrt{2K}\norm{d}_2\epsilon.
\end{equation*}
Thus $\norm{d}_2\le\sqrt{2K}\epsilon/\kappa$. With the specified tolerance, the total error is at most $(c_0+1/4)\sqrt s\lambda/\kappa<2\sqrt s\lambda/\kappa$.

For $x^+=T(x)$, proximal optimality gives KKT residual $(G-MI)(x^+-x)$. This difference has at most $2K$ nonzeros. Adding any output coordinate to its support allows the sparse spectral bound to control that row, giving $\kktr_w(x^+)\le2M\norm{x^+-x}_2$. For $0<\kappa<M$, step differences decay at $a_0=1-\kappa/M$; with $r_0=\norm{x_0-x_\lambda}_2$, $\kktr_w(x_{k+1})\le2M(1+a_0)a_0^k r_0$.
Solving this inequality for $k$ gives finite stopping at any positive tolerance. If $\kappa=M$, one step reaches the fixed point.

\subsection{Outer Recursion}\label{supp-sec:outerproof}
For SCAD~\cite{suppbib-fan2001},
\begin{equation*}
 p'_{\Lambda}(z)=\begin{cases}\Lambda,&z\le\Lambda,\\
 (a\Lambda-z)/(a-1),&\Lambda<z<a\Lambda,\\0,&z\ge a\Lambda,
 \end{cases}
\end{equation*}
whereas MCP~\cite{suppbib-zhang2010} has $p'_{\Lambda}(z)=(\Lambda-z/a)_+$. In both cases $p'_{\Lambda}(z)\ge\Lambda/2$ for $z\le c_h\Lambda$, with the constants in the main text. At outer round $t$, define $E=\cS\cup\{l\notin\cS:|u_l|>c_h\Lambda_t\}$. Then
\begin{equation*}
 |E\setminus\cS|(c_h\Lambda_t)^2\le\norm{u-\thstar}_2^2\le R_{t-1}^2,
\end{equation*}
so $|E|\le2s$ by the base-level choice. All levels use the same shape $h$ and therefore the same exceptional set.

Initially, $\norm{u-\thstar}_2\le2\sqrt s\Lambda_t/\kappa$, and the target score bound holds at every larger level. Each stage starts in its invariant domain and ends within $2\sqrt s\lambda_{\rm old}/\kappa$ of truth. Since $\lambda_{\rm new}\ge\lambda_{\rm old}/2$, this is at most $4\sqrt s\lambda_{\rm new}/\kappa$. Induction covers all stages and inner iterates through $\lambda_t$.

For the outer induction, expand the local gradient around the truth:
\begin{equation*}
 g_j(u)-g_\cH(u)=g_j(\thstar)-g_\cH(\thstar)+(\Gamma_j-\Gamma_\cH)(u-\thstar).
\end{equation*}
The current error has support at most $K+s$. Lemma~\ref{supp-lem:trim} therefore bounds $\norm{E(u)}_\infty$ by $\omega(\tau+dR_{t-1})$. The surrogate identity in main-text~\eqref{eq:score} yields
\begin{equation*}
 \norm{\nabla\ell_u(\thstar)}_\infty\le a_H+dR_{t-1}+\omega(\tau+dR_{t-1})=A+BR_{t-1}.
\end{equation*}
With $\lambda_t=8(A+BR_{t-1})$, Lemma~\ref{supp-lem:path} bounds the error by $R_t=16\sqrt s(A+BR_{t-1})/\kappa$. Iterating proves the outer bound. Sparse-region uniformity covers data-dependent iterates and attacks without sample splitting.

\begin{corollary}[Support recovery]\label{supp-cor:support}
If the bound $R_t>0$ is valid and $\min_{l\in\cS}|\thstar_l|>2R_t$, then $\{l:|\theta_{t,l}|>R_t\}=\cS$.
\end{corollary}
\begin{proof}
$\|\theta_t-\thstar\|_\infty\le R_t$. Hence $|\theta_{t,l}|\le R_t$ for $l\notin\cS$, whereas $|\theta_{t,l}|\ge|\thstar_l|-R_t>R_t$ for $l\in\cS$. Thresholding at $R_t$ therefore recovers $\cS$.
\end{proof}

\subsection{Concentration}\label{supp-sec:concentration}
\subsubsection{Statistical Inputs}\label{supp-sec:statbounds}
We specify the bounds used in main-text Corollary~\ref{cor:statistical}. The IID MCAR model and sub-Gaussian scales are as in main-text Section~\ref{sec:concentration}; node sample sizes are fixed. Use $n_{\min},n_{\max},N_{\mathrm{eff}}$ from Section~\ref{supp-sec:notation}.
Let $\nu_l^2$ be the score variance calculated in main-text Section~\ref{sec:missingness}. The Hessian-coordinate variances are
\begin{equation*}
 \nu_{G,lk}^2=q^{-2}\E(X_l^2X_k^2)-\Sigma_{lk}^2\ (l\ne k),\qquad
 \nu_{G,ll}^2=q^{-1}\E X_l^4-\Sigma_{ll}^2.
\end{equation*}
Set $\nu_V=\max_l\nu_l$ and $\nu_G=\max_{lk}\nu_{G,lk}$. The centered score and Hessian coordinates have exponential tail scales $B_V$ and $B_G$ satisfying
\begin{equation*}
 B_V\lesssim q^{-2}K_X(K_X\norm{\thstar}_2+K_\varepsilon),\qquad
 B_G\lesssim q^{-2}K_X^2.
\end{equation*}

For a failure level $\delta\in(0,1)$, set
\begin{equation*}
 t_\delta=\log(Cmn_{\max}p^2/\delta),\qquad
 \ell_H=\log(Cp/\delta),\qquad \ell_m=\log(Cmp^2/\delta),
\end{equation*}
and define
\begin{equation*}
 D(N,\ell;\nu,B_0)=C\left\{\nu\sqrt{\ell/N}+B_0t_\delta\ell/N\right\}
 +2B_0(t_\delta+1)e^{-t_\delta}.
\end{equation*}
With adjusted absolute constants, a valid simultaneous choice is
\begin{align*}
 a_H&=C\nu_V\sqrt{\ell_H/N_{\mathrm{eff}}}
       +\frac{CB_Vt_\delta\ell_H}{m_{\mathrm H}n_{\min}}
       +2B_V(t_\delta+1)e^{-t_\delta},\\
 \tau&=2D(n_{\min},\ell_m;\nu_V,B_V),\qquad
 \Delta=D(n_{\min},\ell_m;\nu_G,B_G).
\end{align*}
Set $d=2\sqrt{K+s}\Delta$, with $K$ given by main-text~\eqref{eq:K}. The proof below gives these bounds on the same event of probability at least $1-\delta$. Set $\omega=3b/(m-2b)$, $A=a_H+\omega\tau$, and $B=(1+\omega)d$.

The score contribution to the statistical residual is separated into
\begin{align*}
 V_\delta&=\nu_V\left\{\sqrt{\frac{\ell_H}{N_{\rm eff}}}
             +\omega\sqrt{\frac{\ell_m}{n_{\min}}}\right\},\\
 Q_\delta&=B_V\left\{\frac{t_\delta\ell_H}{m_{\rm H}n_{\min}}
              +\frac{\omega t_\delta\ell_m}{n_{\min}}+(1+\omega)(t_\delta+1)e^{-t_\delta}\right\}.
\end{align*}
The two terms of $V_\delta$ describe leading variance contributions from the honest average and trimming. In $Q_\delta$, the first two terms are finite-sample tail contributions and the last is the truncation-bias contribution. Adjusting absolute constants gives $A\le C(V_\delta+Q_\delta)$; no term of $Q_\delta$ is discarded in the main-text error bound.

\subsubsection{Concentration Proof}

\begin{lemma}[Proof-only truncation]\label{supp-lem:concentration}
Let $V_1,\ldots,V_N$ be independent centered scalars satisfying $\E V_i^2\le\nu^2$ and $\Pp(|V_i|>t)\le2e^{-t/B_0}$. Fix $t_\delta,\ell>0$ and define $T=B_0t_\delta$. Except on an event of probability at most $2Ne^{-t_\delta}+2e^{-\ell}$,
\begin{equation*}
\begin{aligned}
 \left|N^{-1}\sum_{i=1}^N V_i\right|
 &\le\nu\sqrt{2\ell/N}+\frac{4B_0t_\delta\ell}{3N}\\
 &\quad+2B_0(t_\delta+1)e^{-t_\delta}.
\end{aligned}
\end{equation*}
\end{lemma}
\begin{proof}
Define $W_i=V_i\ind\{|V_i|\le T\}$. Then $\operatorname{Var}(W_i)\le\E W_i^2\le\nu^2$ and $|W_i-\E W_i|\le2T$. The lost expectation satisfies
\begin{align*}
 |\E W_i|&=|\E V_i\ind\{|V_i|>T\}|\\
 &\le T\Pp(|V_i|>T)+\int_T^\infty\Pp(|V_i|>t)\,dt
 \le2B_0(t_\delta+1)e^{-t_\delta}.
\end{align*}
Bernstein's inequality bounds the centered $W_i$ average by $\nu\sqrt{2\ell/N}+4T\ell/(3N)$ with failure probability $2e^{-\ell}$. All $V_i=W_i$ except on an event of probability $2Ne^{-t_\delta}$. Adding the bias proves the claim.
\end{proof}

Apply the lemma to at most $p+p^2$ centered coordinates per honest observation. With $t_\delta=\log(Cmn_{\max}p^2/\delta)$, any truncation consumes at most a fixed fraction of $\delta$. Use $\ell_m=\log(Cmp^2/\delta)$ for the union over nodes/coordinates, and $\ell_H=\log(Cp/\delta)$ over honest-mean score coordinates. The common no-truncation event covers both. Enlarging universal $C$ allocates the remaining failure probability. The resulting bounds in Section~\ref{supp-sec:statbounds} hold jointly with probability at least $1-\delta$.
For completeness, write one coordinate of the honest reference score as a weighted sum of independent centered observations:
\begin{equation*}
 (g_\cH(\thstar))_l=\sum_{j\in\cH}\sum_{i=1}^{n_j}
 \frac{\phi_{ji,l}}{m_{\mathrm H}n_j}.
\end{equation*}
Weights sum to one, have squared sum $1/N_{\mathrm{eff}}$ and maximum $1/(m_{\mathrm H}n_{\min})$. Truncating at $B_Vt_\delta$ gives weighted variance at most $\nu_V^2/N_{\mathrm{eff}}$, centered term bound $2B_Vt_\delta/(m_{\mathrm H}n_{\min})$, and bias at most $2B_V(t_\delta+1)e^{-t_\delta}$. Weighted Bernstein proves the score bound. The no-truncation event covers at most $mn_{\max}$ honest observations. Applying the scalar lemma at each $n_j$ and upper-bounding by $D(n_{\min},\ell_m;\nu,B_0)$ proves the remaining bounds. Sizes are fixed independently of observations; equal sizes give $N_{\mathrm{eff}}=m_{\mathrm H}n$.

The node-average score is also bounded by the maximum local bound, giving $\tau=2D(n_{\min},\ell_m;\nu_V,B_V)$.
The entrywise difference $\Gamma_j-\Gamma_\cH$ is at most $2\Delta$. For $r$-sparse $e$, its product with this matrix is bounded in sup norm by $2\Delta\norm{e}_1\le2\Delta\sqrt r\norm{e}_2$. Taking $r=K+s$ yields the specified $d$.

For a support $U$ of order $r\le4K$, $\norm{(\Gamma_1-\Sigma)_{UU}}_{\mathrm{op}}\le r\Delta$. Weyl's eigenvalue bound proves the sparse curvature event under the condition in the main text. To interpret a sample-size requirement, let $g_0=\min(\mu-\kappa,M-\nu_{\rm pop})>0$ and assume the exponentially small bias term is at most $g_0/(12K)$. Sufficient leading conditions have the form
\begin{equation*}
 n_{\min}\gtrsim \frac{K^2\nu_G^2\ell_m}{g_0^2}+\frac{K B_Gt_\delta\ell_m}{g_0}.
\end{equation*}
The logarithmic factor $t_\delta$ depends on $n_{\max}$. The bound uses the smallest honest-node sample size to control all local Hessians, including the master's. Outer contraction further requires $\rho=32\sqrt{s(K+s)}(1+\omega)\Delta/\kappa<1$. For bounded population fourth moments, $\nu_G^2$ is of order at most $q^{-2}$, while $B_G$ remains of order $q^{-2}$.

For the tail bounds, condition on masks: the off-diagonal signal's linear coefficient norm is at most a constant times $\norm{\thstar}_2/q$. Multiplication by $M_{X,l}X_l/q$ gives a conservative $q^{-2}$ score scale, since products of sub-Gaussian variables are sub-exponential even under dependence. Centering changes only constants. The noise term permits the sharper $q^{-1}K_XK_\varepsilon$ bound; the conservative scale suffices.

For main-text Corollary~\ref{cor:statistical}, the common event controls the honest score, node dispersion and master curvature. Thus $A=a_H+\omega\tau\le C(V_\delta+Q_\delta)$ after adjusting $C$. Substitute this and $d=2\sqrt{K+s}\Delta$ into~\eqref{eq:recursion} to obtain $\rho$ and the stated error bound.

\section{Additional Results}\label{supp-contents-3}\label{supp-sec:additional}
\subsection{Endpoint Error}\label{supp-sec:endpoint}\label{supp-sec:endpointproof}
Main-text Theorem~\ref{thm:outer} controls every inner iterate of the bound-based policy. The next result bounds the estimates returned by the data-tuned policy using curvature only on their supports together with the true support. Its coefficients depend on $\thstar$, $\cS$ and honest-node identities, which the simulations use for evaluation.
Let $x_0$ be the initialization and, for $t\ge1$, let $x_t$ be returned from the surrogate formed at $x_{t-1}$, with nonnegative weights $w_t$ and full KKT residual at most $\epsilon_t$. For $t\ge0$, write $U_t=\cS\cup\supp(x_t)$, $r_t=|U_t|$, and $e_t=x_t-\thstar$.
For $t\ge1$ with nonempty $U_t$, write $D_{jt}=(\Gamma_j-\Gamma_\cH)_{U_t,U_{t-1}}$ and define
\begin{align*}
 \kappa_t&=\lambda_{\min}(\Gamma_{1,U_tU_t}),\qquad a_t=\|g_\cH(\thstar)_{U_t}\|_\infty,\\
 \tau_t&=\max_{j\in\cH}\|(g_j(\thstar)-g_\cH(\thstar))_{U_t}\|_\infty,\qquad d_t=\max_{j\in\cH}\|D_{jt}\|_{2\to\infty},
\end{align*}
where $\|D\|_{2\to\infty}$ is the largest row Euclidean norm; a matrix with an empty index set has norm zero. With $\omega=3b/(m-2b)$, put
\begin{equation*}
 \alpha_t=\frac{\|D_{1t}\|_{\rm op}+\omega\sqrt{r_t}d_t}{\kappa_t},\qquad
 \zeta_t=\frac{\sqrt{r_t}(a_t+\omega\tau_t+\epsilon_t)+\|w_{t,U_t}\|_2}{\kappa_t}.
\end{equation*}
Here $\alpha_t$ propagates the previous round's error, while $\zeta_t$ accounts for score variation, Byzantine contamination, penalization and numerical error. If $U_t=\varnothing$, then $e_t=0$ and we set $\alpha_t=\zeta_t=0$.
\begin{theorem}[Endpoint recursion]\label{supp-thm:endpoint}
If $m_{\rm B}\le b<m/2$ and $\kappa_t>0$ at each returned endpoint with $U_t\ne\varnothing$, then
\begin{equation}\label{supp-eq:endpoint}
 \|e_t\|_2\le\alpha_t\|e_{t-1}\|_2+\zeta_t.
\end{equation}
Consequently, if $\alpha_t\le\bar\alpha<1$ and $\zeta_t\le\bar\zeta$ for $1\le t\le T$, then
\begin{equation*}
 \|e_T\|_2\le\bar\alpha^T\|e_0\|_2
       +\frac{\bar\zeta(1-\bar\alpha^T)}{1-\bar\alpha}.
\end{equation*}
The conclusion allows data-dependent supports, weights and adaptive Byzantine messages.
\end{theorem}

\begin{proof}
Condition on realized samples, messages, weights and endpoints. Set $u=x_{t-1}$ and use Theorem~\ref{supp-thm:endpoint}'s notation. The empty-$U_t$ case is immediate. Otherwise, $u-\thstar$ is supported on $U_{t-1}$; the honest-gradient identity and coordinatewise trimming give
\begin{align*}
 \max_{j\in\cH}\|(g_j(u)-g_\cH(u))_{U_t}\|_\infty
 &\le\tau_t+d_t\|e_{t-1}\|_2,\\
 \|E(u)_{U_t}\|_2
 &\le\sqrt{r_t}\omega(\tau_t+d_t\|e_{t-1}\|_2).
\end{align*}
Choose $\xi_t\in\partial\|x_t\|_1$ so that the full residual
$z_t=\Gamma_1x_t-c(u)+w_t\odot\xi_t$ has sup norm at most $\epsilon_t$. Combining this relation with the surrogate identity yields
\begin{align*}
 \Gamma_{1,U_tU_t}e_{t,U_t}
 &=z_{t,U_t}-(w_t\odot\xi_t)_{U_t}
   -g_\cH(\thstar)_{U_t}-E(u)_{U_t}\\
 &\quad+D_{1t}e_{t-1,U_{t-1}}.
\end{align*}
Here $D_{1t}=(\Gamma_1-\Gamma_\cH)_{U_t,U_{t-1}}$ fixes the last term's sign; its operator norm controls curvature mismatch directly. Positive definiteness on $U_t$ yields
\begin{align*}
 \kappa_t\|e_t\|_2
 &\le\{\|D_{1t}\|_{\rm op}+\omega\sqrt{r_t}d_t\}\|e_{t-1}\|_2\\
 &\quad+\sqrt{r_t}(a_t+\omega\tau_t+\epsilon_t)+\|w_{t,U_t}\|_2.
\end{align*}
Division proves \eqref{supp-eq:endpoint}, without independence between $U_t$ and the data. Induction from $R_0\ge\|e_0\|_2$ with $R_t=\alpha_tR_{t-1}+\zeta_t$ gives $\|e_t\|_2\le R_t$, even for $\alpha_t\ge1$. If all $\alpha_t\le\bar\alpha<1$ and $\zeta_t\le\bar\zeta$, sum the geometric series. Empty true and estimated supports give zero error without inversion.

\end{proof}

This curvature condition concerns the realized supports; uniform path theory additionally controls support growth and termination. Simultaneous entrywise concentration can bound the random coefficients. Simulations instead audit exact values using truth after fitting. Theorem~\ref{supp-thm:local} controls continuation of a fixed subproblem from an admitted endpoint; Theorem~\ref{supp-thm:calibration} gives observable radii from independent samples.

\subsection{Oracle Conditions}\label{supp-sec:oracle}
Fix symmetric $(G,c)$ and frozen weights $w_l=(\lambda/\Lambda)p'_\Lambda(|u_l|)\ge0$. Define $x^o_{\cS}=G_{\cS\cS}^{-1}c_{\cS}$ and $x^o_{\cS^c}=0$ when $\cS\ne\varnothing$ and $G_{\cS\cS}$ is positive definite; set $x^o=0$ when $\cS=\varnothing$. Use the KKT residual from main-text Section~\ref{sec:boundpolicy} with $(G,c)$.

\begin{proposition}[Fixed-surrogate oracle]\label{supp-prop:oracle}
Suppose $w_{\cS}=0$ and
\begin{equation*}
 \gamma=\min_{l\notin\cS}\{w_l-|(Gx^o-c)_l|\}>0,
\end{equation*}
with $\gamma=\infty$ when $\cS^c=\varnothing$. Let $x$ satisfy $\kktr_w(x)\le\epsilon$ and $\|x\|_0\le K$. On $V=\cS\cup\supp(x)\ne\varnothing$, assume $\kappa I\preceq G_{VV}\preceq MI$ with $\kappa>0$. Then $x^o$ satisfies the full weighted KKT system and
\begin{equation}\label{supp-eq:oracle-distance}
 \|x-x^o\|_2\le e_{\rm opt}:=\sqrt{K+s}\epsilon/\kappa.
\end{equation}
If also $\min_{l\in\cS}|x_l^o|>e_{\rm opt}$ and $Me_{\rm opt}+\epsilon<\gamma$, then $\supp(x)=\cS$. The signal condition is vacuous when $\cS=\varnothing$; the inactive condition is vacuous when $\cS^c=\varnothing$. If $V=\varnothing$, both conclusions hold directly.
\end{proposition}

\begin{proof}
The active gradient vanishes, and the inactive margin yields $Gx^o-c+w\odot\xi^o=0$ for some $\xi^o\in\partial\|x^o\|_1$. Choose $\xi\in\partial\|x\|_1$ so that $z=Gx-c+w\odot\xi$ has sup norm at most $\epsilon$. For $d=x-x^o$, subtraction and weighted subgradient monotonicity give
\begin{equation*}
 \kappa\|d\|_2^2\le d^\top z\le\sqrt{K+s}\|d\|_2\epsilon.
\end{equation*}
Division gives~\eqref{supp-eq:oracle-distance} when $d\ne0$; the zero case is immediate. The signal condition prevents active zeros. A false positive $l\in\supp(x)\setminus\cS$ would require $|(Gx-c)_l|\ge w_l-\epsilon$, whereas $l\in V$, $\supp(d)\subseteq V$ and the upper spectral bound give
\begin{equation*}
 |(Gx-c)_l|\le w_l-\gamma+M\|d\|_2<w_l-\epsilon,
\end{equation*}
a contradiction.
\end{proof}

For SCAD, $\|u-x^o\|_\infty\le\Lambda$ and $\min_{l\in\cS}|x_l^o|>(a+1)\Lambda$ ensure zero active weights and inactive weights $\lambda$. MCP uses its actual inactive weights in the margin.

For $G=\Gamma_1$, $c=\Gamma_1u-\ghat(u)$ and positive definite $\Gamma_{\cH,\cS\cS}$, let $x_H^o$ minimize the equally weighted honest-node loss on $\cS$; both oracles are zero if $\cS=\varnothing$. Then
\begin{equation*}
 G_{\cS\cS}(x^o-x_H^o)_{\cS}
 =\{(\Gamma_1-\Gamma_\cH)(u-x_H^o)-E(u)\}_{\cS}.
\end{equation*}
Honest-oracle equivalence additionally requires the curvature-mismatch and aggregation terms, after inversion by $G_{\cS\cS}$, to vanish at the honest oracle's standard-error scale.

\subsection{Local Certificate}\label{supp-sec:local}
Fix surrogate $(G,c)$, nonnegative $w$ and endpoint $\bar x$. For its computed support $U=\supp(\bar x)$, set $g=G\bar x-c$. If $U\ne\varnothing$, define
\begin{align*}
 \mu_U&=\lambda_{\min}(G_{UU}),\qquad L_U=\lambda_{\max}(G_{UU}),\\
 \eta_U&=\|g_U+w_U\odot\operatorname{sign}(\bar x_U)\|_2.
\end{align*}
For $L>0$, write $T_L(x)=\soft_{w/L}\{x-(Gx-c)/L\}$.

\begin{theorem}[Local continuation]\label{supp-thm:local}
Suppose $U\ne\varnothing$ and $\mu_U>0$. Put $r_U=\eta_U/\mu_U$ and assume
\begin{equation}\label{supp-eq:local-margin}
 |g_l|+\|G_{lU}\|_2r_U\le w_l\qquad(l\notin U).
\end{equation}
Then $T_{L_U}$ preserves the set
\begin{equation*}
 \mathcal B_U=\{x:\supp(x)\subseteq U,\ \|x-\bar x\|_2\le r_U\}
\end{equation*}
and contracts on it with modulus $1-\mu_U/L_U$. Its unique fixed point $x^\dagger$ in this set satisfies the full weighted KKT conditions and $\|x^\dagger-\bar x\|_2\le r_U$. Starting anywhere in $\mathcal B_U$, continuation remains supported in $U$ and reaches every positive KKT tolerance in finitely many steps. For $U=\varnothing$, $\bar x=0$ is a fixed point for every $L>0$ whenever $|c_l|\le w_l$ for all $l$.
\end{theorem}

\begin{proof}
For $U\ne\varnothing$, abbreviate $L=L_U$, $\mu=\mu_U$, $a=1-\mu/L$, $r=r_U$ and $T=T_L$. For $x\in\mathcal B_U$ and $l\notin U$, $x_l=0$ and~\eqref{supp-eq:local-margin} gives
\begin{equation*}
 |(Gx-c)_l|\le |g_l|+\|G_{lU}\|_2\|x-\bar x\|_2\le w_l.
\end{equation*}
Thus $(Tx)_l=0$, including on the boundary. Nonexpansiveness and $\mu I\preceq G_{UU}\preceq LI$ give, for $x,y\in\mathcal B_U$,
\begin{equation*}
 \|Tx-Ty\|_2\le\|(I_U-G_{UU}/L)(x-y)_U\|_2\le a\|x-y\|_2.
\end{equation*}
With $z_U=\operatorname{sign}(\bar x_U)$, nonexpansiveness and $\soft_{w_U/L}(\bar x_U+w_U\odot z_U/L)=\bar x_U$ give $\|T\bar x-\bar x\|_2\le\eta_U/L$, so
\begin{equation*}
 \|Tx-\bar x\|_2\le a\|x-\bar x\|_2+\eta_U/L\le ar+\eta_U/L=r.
\end{equation*}
The ball is complete, so Banach's theorem gives a unique fixed point and convergence throughout it; soft-threshold optimality gives full KKT conditions. Coordinates in $U$ may change sign or vanish.

For finite stopping, set $x^+=Tx$ and $d=x^+-x$. Proximal optimality gives $\xi^+\in\partial\|x^+\|_1$ with $Gx^+-c+w\odot\xi^+=(G-LI)d$. Since $\supp(d)\subseteq U$, $D_U=\max_l\|G_{lU}-LI_{lU}\|_2<\infty$ gives $\kktr_w(x^+)\le D_U\|d\|_2$. For $0<a<1$, step lengths at most $a^k\|x_1-x_0\|_2$ ensure finite stopping. If $a=0$, one step reaches the fixed point; if $r=0$, the center is fixed. For $U=\varnothing$, $|c_l|\le w_l$ gives full KKT and singleton invariance.
\end{proof}

The proof permits $r\ge\eta_U/\mu_U$ satisfying the inactive inequalities. For nonempty supports, computation uses $(\eta_U+10^{-10})/\mu_U$, $\mu_U>10^{-10}$ and inactive slack above $10^{-10}$, recomputed in double precision after surrogate or weight changes. This certifies fixed-subproblem continuation; the recorded prefix gives earlier observed supports.

\subsection{Masked Calibration}\label{supp-sec:calibration}
Fix candidates $x_1,\ldots,x_J$, $J\ge1$, from training and attacks. Reserve a prespecified sample of $N_{\rm cal}$ IID trusted observations from the same population, independent of training and unseen by fitting or attackers, with MCAR masks as in main-text Section~\ref{sec:model}. Assume known population bounds $|X_l|\le H$, $|Y|\le Y_{\max}$, valid integer $s\ge1$ and known $q\in(0,1]$. Compute corrected moments $\widetilde G,\widetilde\gamma$.

Fix $0<\delta<1$ and set $\delta_c=\delta_g=\delta/2$. Define
\begin{equation*}
 \Delta_N=\frac{H^2}{q^2}\sqrt{\frac{2\log(2p^2/\delta_c)}{N_{\rm cal}}}.
\end{equation*}
For $1\le r\le p$, let $b_r(\widetilde G)$ be the minimum, over rows, of the diagonal minus the sum of the $r-1$ largest absolute off-diagonal entries. For each candidate, put
\begin{align*}
 r_j&=\min\{\|x_j\|_0+s,p\},\qquad
 \underline\kappa_j=b_{r_j}(\widetilde G)-r_j\Delta_N,\\
 C_j&=H^2\|x_j\|_1/q^2+HY_{\max}/q,\qquad
 \epsilon_j=C_j\sqrt{\frac{2\log(2Jp/\delta_g)}{N_{\rm cal}}}.
\end{align*}

\begin{samepage}
\begin{theorem}[Calibration bound]\label{supp-thm:calibration}
Under the calibration assumptions above, with probability at least $1-\delta$, the following bound holds for every candidate with $\underline\kappa_j>0$:
\begin{equation}\label{supp-eq:calibration}
 \|x_j-\thstar\|_2\le\widehat R_j
 :=\frac{\sqrt{r_j}\{\|\widetilde Gx_j-\widetilde\gamma\|_\infty+\epsilon_j\}}{\underline\kappa_j}.
\end{equation}
The same event covers any admitted index selected using the calibration statistics, including the minimizer of $\widehat R_j$. If no candidate has positive calibrated curvature, the rule abstains. Equivalently, the probability of returning an admitted candidate whose bound fails is at most $\delta$.
\end{theorem}
\end{samepage}

\begin{proof}
Condition on training, attacks and candidates. Calibration remains IID, with corrected moments $G_i,d_i$ for observation $i$. The bound $|(G_i)_{lk}|\le H^2/q^2$, Hoeffding and a union over $p^2$ entries give $\|\widetilde G-\Sigma\|_{\max}\le\Delta_N$ with probability at least $1-\delta_c$. Gershgorin gives $\lambda_{\min}(\widetilde G_{VV})\ge b_r(\widetilde G)$ for nonempty $|V|\le r$, hence
\begin{equation*}
 \lambda_{\min}(\Sigma_{VV})\ge b_r(\widetilde G)-|V|\Delta_N\ge b_r(\widetilde G)-r\Delta_N.
\end{equation*}
This holds simultaneously over all supports, including those selected by training.

The corrected score for candidate $x_j$ is
\begin{equation*}
 (G_ix_j-d_i)_l=\frac{M_{i,l}X_{i,l}}q
 \left\{X_{i,l}x_{j,l}+\sum_{k\ne l}\frac{M_{i,k}}qX_{i,k}x_{j,k}-Y_i\right\}.
\end{equation*}
Its magnitude is at most $C_j$ since $q\le1$, with mean $\Sigma(x_j-\thstar)$ by MCAR and $\E(X\varepsilon)=0$. Hoeffding on $[-C_j,C_j]$ and a union over $Jp$ pairs give, with probability at least $1-\delta_g$,
\begin{equation*}
 \|\widetilde Gx_j-\widetilde\gamma-\Sigma(x_j-\thstar)\|_\infty\le\epsilon_j
 \qquad(1\le j\le J).
\end{equation*}
On their intersection, $e_j=x_j-\thstar$ is $r_j$-sparse. If $\underline\kappa_j>0$ and $e_j\ne0$, then
\begin{align*}
 \underline\kappa_j\|e_j\|_2^2
 &\le e_j^\top\Sigma e_j\le\|e_j\|_1\|\Sigma e_j\|_\infty\\
 &\le\sqrt{r_j}\|e_j\|_2\{\|\widetilde Gx_j-\widetilde\gamma\|_\infty+\epsilon_j\}.
\end{align*}
Division proves~\eqref{supp-eq:calibration}; zero error is immediate. Simultaneity covers any admitted selection, and integration over training randomness gives the unconditional guarantee.
\end{proof}

Each prespecified sample size has a separate guarantee; data-dependent stopping requires a simultaneous-in-size bound. Sample maxima alone cannot establish the population bounds $H,Y_{\max}$. Gaussian experiments use local optimization checks; calibration uses the bounded designs below. These Hoeffding bounds can require a $q^{-4}$ sample-size factor; main-text Section~\ref{sec:concentration}'s variance-based curvature analysis has a leading $q^{-2}$ factor under its stated conditions.

\subsection{Constrained Tuning}\label{supp-sec:selection}
Let $\mathcal C\subseteq[1,\infty)$ be a finite set of inflation factors. All candidates use the same valid inputs $A,B,R_0,\kappa,M,s$. For $c\in\mathcal C$, set $R_0^{(c)}=R_0$ and define
\begin{equation*}
 \lambda_t^{(c)}=8c(A+BR_{t-1}^{(c)}),\qquad
 R_t^{(c)}=2\sqrt{s}\lambda_t^{(c)}/\kappa.
\end{equation*}
The candidate uses its own $R_{t-1}^{(c)}$ and $\lambda_t^{(c)}$ in the base-level rule of main-text Section~\ref{sec:boundpolicy}, with the same frozen-weight, stage-ratio and stopping rules. Put
\begin{equation*}
 \rho_c=16c\sqrt{s}B/\kappa=c\rho,\qquad
 \mathcal C_{\rm adm}=\{c\in\mathcal C:\rho_c<1\}.
\end{equation*}

\begin{corollary}[Bound-based selection]\label{supp-cor:selection}
Suppose the assumptions of main-text Theorem~\ref{thm:outer} hold on a common event of probability at least $1-\delta$. For every $c\in\mathcal C_{\rm adm}$, all inner supports have size at most $K$, and on that same event
\begin{equation}\label{supp-eq:selection-bound}
 \|\theta_T^{(c)}-\thstar\|_2
 \le\rho_c^TR_0+\frac{16c\sqrt{s}A}{\kappa(1-\rho_c)}(1-\rho_c^T).
\end{equation}
Any validation-selected admitted candidate retains its own bound. No factor depending on $|\mathcal C|$ is needed in the failure probability. If $\mathcal C_{\rm adm}=\varnothing$, the selection rule abstains.
\end{corollary}

\begin{proof}
The common score envelope is at most $\lambda_t^{(c)}/8$, since $c\ge1$. The exceptional-set bound uses the candidate's own base and radius and remains valid. The sparse-path argument therefore yields $R_t^{(c)}=16c\sqrt{s}A/\kappa+\rho_cR_{t-1}^{(c)}$. Solving this recursion proves~\eqref{supp-eq:selection-bound}. The score and curvature event holds uniformly over all allowed sparse inputs, so it controls every candidate's adaptive path. Selection preserves the bound for the chosen index. If different estimated input bounds are used across candidates, their joint validity requires a separate argument.
\end{proof}

The routine \texttt{certified\_select} filters $c<1$ and $\rho_c\ge1$, solves each retained policy, and returns the endpoint with lowest validation loss together with all traces. It uses supplied score bounds. For $J$ separately solved candidates of $T$ rounds, the communication budget is $16(m-1)pJT$ bytes, excluding control messages.

\subsection{Master Sample Size}\label{supp-sec:rank}
\begin{proposition}[Master sample size]\label{supp-prop:rank}
Suppose $p>n_1$, $q<1$, and every empirical diagonal of $A_1$ is positive. The uniform sparse-curvature condition of main-text Theorem~\ref{thm:outer} requires $4K\le n_1$, and hence $n_1\ge8s$.
\end{proposition}
\begin{proof}
If $\min(4K,p)>n_1$, select $|V|=n_1+1$ columns. A nonzero $v\in\ker Z_{1,V}$ exists and, by the moment correction,
\begin{equation*}
 v^\top\Gamma_{1,VV}v=-\frac{1-q}{q^2}\sum_{l\in V}(A_1)_{ll}v_l^2<0.
\end{equation*}
This contradicts the positive lower sparse eigenvalue. Since $p>n_1$, it follows that $4K\le n_1$. The definition $K\ge2s$ gives the last assertion.
\end{proof}

\section{Implementation}\label{supp-contents-4}
\subsection{Solver Policies}\label{supp-sec:implementation}
Both policies use \path{code/robust_sparse.py}. The bound-based policy takes valid bounds, uses step $1/M$, and checks the scalar contraction condition. The data-tuned policy uses local-Lasso initialization and quadratic backtracking. It accepts $d$ when $d^\top\Gamma_1d\le L\norm d_2^2+10^{-12}$, doubles $L$ upon rejection, and divides the next trial value by $1.05$. The round target is $\max(\lambda_{\rm floor},\lambda_{\rm local}0.65^{t+1})$ and its base is $\max(\lambda_{\rm target},\lambda_{\rm local}0.65^t)$.

Data-tuned tolerances are $\lambda/26$ at intermediate stages and $\min(10^{-5},\lambda_{\rm target}/1000)$ at the final stage. Runs fail upon nonfinite iterates, norm above $10^4$, failed backtracking, or 1200 iterations without meeting tolerance. The bound-based policy has its stated tolerance and no iteration or norm cutoff. Neither policy adds ridge regularization or a PSD projection. SCAD uses $a=3.7$; MCP uses $a=3$.

\subsection{Node Sizes and Penalties}
Each corrected moment uses its node's $n_j$, while trimming gives each node one message. Thus the deterministic proofs remain unchanged; unequal-size concentration supplies the bounds used by the theoretical policy. Data-tuned fitting uses assigned sizes for all $m$ nodes, without honest-node labels:
\begin{equation*}
 \widetilde n_{\min}=\min_{1\le j\le m}n_j,\qquad
 \widetilde N_{\mathrm{eff}}=m^2/\sum_{j=1}^m n_j^{-1}.
\end{equation*}
Its heuristic penalty floor is proportional to
\begin{equation*}
 \frac{\widehat\sigma_Y}{q}\left\{
 \sqrt{\frac{\log p}{\widetilde N_{\mathrm{eff}}}}
 +\frac bm\sqrt{\frac{\log p}{\widetilde n_{\min}}}\right\}.
\end{equation*}
Here node sizes are configuration metadata, and initialization uses the master's $n_1$. Equal sizes give $\widetilde N_{\mathrm{eff}}=mn$ and $\widetilde n_{\min}=n$, recovering the original rule. Calibration uses its separate trusted sample.

\subsection{Aggregators}
For coordinate messages $z_1,\ldots,z_m$, median uses their median. VRMOM requires equal node sizes $n_j=n$ for its common standard-error correction. It uses $J=10$, $\tau_k=k/(J+1)$, $a_k=\Phi^{-1}(\tau_k)$, and the trusted master's per-observation corrected-score standard deviation $\widehat\sigma$:
\begin{equation*}
\begin{aligned}
 \widehat\mu_{\rm VR}
 &=\widehat\mu_{\rm med}-\frac{\widehat\sigma}{m\sqrt n\sum_{k=1}^{J}\varphi(a_k)}\\
 &\quad\times\sum_{j=1}^{m}\sum_{k=1}^{J}
 \left\{\ind\left(z_j\le\widehat\mu_{\rm med}+\frac{\widehat\sigma a_k}{\sqrt n}\right)-\tau_k\right\}.
\end{aligned}
\end{equation*}
This adapts \citet{suppbib-tu2021} to corrected scores, with standard error floored at $10^{-12}$. All aggregators share the same optimization and tuning budget.

\section{Joint Missingness and Attacks}\label{supp-contents-5}\label{supp-sec:competitive}
\subsection{Design}
Table~\ref{supp-tab:primarydesign} gives the primary benchmark settings. Two same-sign active blocks expose missingness bias on the true support; independent predictors, separated active coordinates and alternating signs within blocks provide controls.
\begin{table}[!htbp]
\caption{Primary benchmark design and tuning.}\label{supp-tab:primarydesign}
\centering\setlength{\tabcolsep}{4pt}
\begin{tabular}{>{\raggedright\arraybackslash}p{.20\textwidth}>{\raggedright\arraybackslash}p{.73\textwidth}}\toprule
Item & Specification\\\midrule
Training & $m=12$, $n_j=512$, $p=256$; Gaussian AR(1) covariates with correlation $0.6$ and unit-variance response noise.\\[3pt]
Signal & Six active coefficients in two separated blocks of three adjacent predictors; $1$ in the first block and $-1$ in the second.\\[3pt]
Main grid & $q\in\{0.9,0.7,0.5\}$ crossed with no attack or one/two Gaussian or in-range attackers; correct trimming budget.\\[3pt]
Extensions & $q=0.3$, complete observations, or $p=768>n_j$, together with the control designs: 23 configurations in total, each with 30 independent seeds.\\[3pt]
Fitting & Nine methods from main-text Section~\ref{sec:experiments}; eight rounds and factors $\{0.075,0.15,0.3,0.6,1.2,2.4,4.8\}$; stopping rules in Section~\ref{supp-sec:implementation}.\\[3pt]
Tuning scale & Master's response standard deviation, divided by $q$ for corrected moments.\\[3pt]
Evaluation data & Independent complete samples of 400 validation and 2000 test observations.\\\bottomrule
\end{tabular}
\end{table}

Under the centered linear model and independent coordinate masks with observation probability $q$, the population score of the zero-filled loss at the true coefficient is
\[
 E\{Z(Z^\top\theta^\star-Y)\}=q(1-q)\{\operatorname{diag}(\Sigma)-\Sigma\}\theta^\star.
\]
It vanishes for independent predictors. Within a correlated active block, coefficients of the same sign reinforce the off-diagonal contributions. The support layouts therefore test when moment correction improves estimation, with the independent design providing a zero-score control.

Methods share training data, masks, validation and test samples. Within each dimension, correlation and support layout, shared latent data and mask uniforms also pair observation probabilities and attacks. The plan and seeds were saved before the full run; an independent implementation and timing pilot is retained separately.

\subsection{Comparators and Tuning}
Validation minimizes prediction error among successful candidates, using neither truth nor honest-node identities. SLARD shares the solver and tuning rule in Table~\ref{supp-tab:primarydesign}, whereas \citet{suppbib-tu2023} use BIC and FISTA.

Gaussian reports are independent $N(0,20^2I)$ vectors. In-range reports use
honest 30th/70th coordinate percentiles, with directions fixed by active
coefficients and nearby inactive coordinates. Methods share attack rules/seeds;
in-range values respond to current honest gradients. The main comparison uses
the correct budget; Section~\ref{supp-sec:earlierbenchmark} studies undertrimming.

\subsection{Evaluation}
A successful fit completes eight rounds with every stopping rule satisfied. Coefficient and prediction errors condition on success, with denominators reported. Exact support recovery uses threshold $10^{-6}$ and counts failures as incorrect. Paired differences use jointly successful repetitions and report their count and Monte Carlo SE. Support and target-attainment proportions use all 30 repetitions. Targets $0.1$, $0.2$ and $0.5$ are assessed retrospectively from truth and require final success; an early low-error iterate followed by failure does not count. Coefficient error is never used for selection.

The ancillary cost records include all candidates and failures, with the
timing and payload definitions in \path{ancillary_tables/primary_cost_protocol.tex}.
Figure~\ref{supp-fig:competitivepaired} gives paired reference differences and exact
support counts; all displayed fits succeed in all 30 repetitions.
\begin{figure}[!htbp]
\centering\includegraphics[width=\textwidth,height=.72\textheight,keepaspectratio]{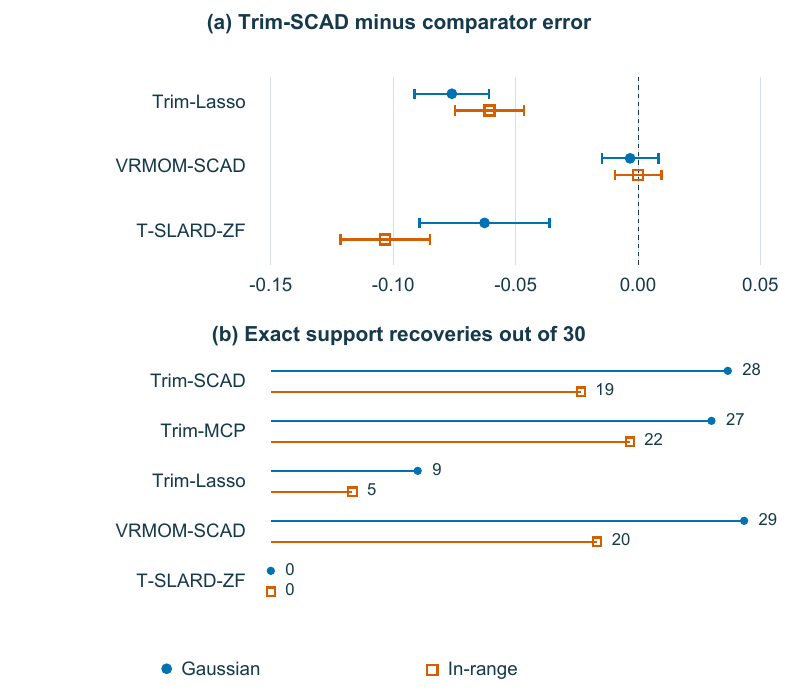}
\caption{Paired comparisons with $q=0.7$, $m=12$, $m_{\rm B}=b=2$, $n_j=512$ and $p=256$. (a) Mean coefficient-error differences, Trim-SCAD minus each comparator, using 30 shared datasets; bars show two paired Monte Carlo standard errors. Negative values favor Trim-SCAD. (b) Exact support recoveries among all 30 repetitions, using the fixed $10^{-6}$ threshold. Circles and squares denote Gaussian and in-range attacks, respectively.}
\label{supp-fig:competitivepaired}
\end{figure}

\begin{table}[!htbp]\centering\begingroup\renewcommand{\baselinestretch}{1.2}\small\selectfont\setlength{\tabcolsep}{3pt}
\renewcommand{\arraystretch}{1.05}
\caption{Configuration IDs for the joint missingness and Byzantine benchmark. Unless changed below, correlation is 0.6, $p=256$, and signals are clustered.}\label{supp-tab:competitive-settings}
\begin{tabular}{lccccc}
\toprule $q$ & \shortstack{No attack\\$m_{\rm B}=0$} & \shortstack{Gaussian\\$m_{\rm B}=1$} & \shortstack{In-range\\$m_{\rm B}=1$} & \shortstack{Gaussian\\$m_{\rm B}=2$} & \shortstack{In-range\\$m_{\rm B}=2$}\\\midrule
0.9 & C01 & C02 & C03 & C04 & C05 \\
0.7 & C06 & C07 & C08 & C09 & C10 \\
0.5 & C11 & C12 & C13 & C14 & C15 \\
\midrule
\multicolumn{6}{l}{Additional configurations: changes from C09}\\
C16 & \multicolumn{5}{l}{Correlation 0.0}\\
C17 & \multicolumn{5}{l}{$q=0.3$}\\
C18 & \multicolumn{5}{l}{$q=0.3$; in-range attack}\\
C19 & \multicolumn{5}{l}{$q=1.0$}\\
C20 & \multicolumn{5}{l}{$p=768$}\\
C21 & \multicolumn{5}{l}{$p=768$; in-range attack}\\
C22 & \multicolumn{5}{l}{Spread signals}\\
C23 & \multicolumn{5}{l}{Alternating signals}\\
\bottomrule
\end{tabular}\endgroup\end{table}

The matrices retain every method and setting. Coefficient errors are conditional
means $\pm$ Monte Carlo SE; a dash denotes no successful fit. Method codes are TS (Trim-SCAD),
TM (Trim-MCP), TL (Trim-Lasso), VR (VRMOM-SCAD), MD (Median-SCAD),
MN (Mean-SCAD), UC (uncorrected SCAD), MS (M-SLARD-ZF), and TT (T-SLARD-ZF).
The count matrix gives successful fits/exact
support recoveries, both out of 30; failures count as incorrect support.
F1 and test MSE, with their original denominators, are retained with these
entries in \path{ancillary_tables/competitive_accuracy.csv}.

\par\begingroup\compacttable\setlength{\tabcolsep}{3pt}
\begin{longtable}{lccccc}
\caption{Coefficient error in all planned configurations: trimmed and median-based corrected methods.}\label{supp-tab:competitive-accuracy}\\
\toprule ID & TS & TM & TL & VR & MD\\\midrule\endfirsthead
\toprule ID & TS & TM & TL & VR & MD\\\midrule\endhead
\bottomrule\endlastfoot
C01 & $0.063\pm0.005$ & $0.061\pm0.004$ & $0.086\pm0.004$ & $0.065\pm0.004$ & $0.081\pm0.005$ \\
C02 & $0.082\pm0.006$ & $0.085\pm0.008$ & $0.110\pm0.006$ & $0.078\pm0.004$ & $0.084\pm0.005$ \\
C03 & $0.064\pm0.005$ & $0.074\pm0.009$ & $0.097\pm0.005$ & $0.067\pm0.004$ & $0.077\pm0.005$ \\
C04 & $0.091\pm0.006$ & $0.091\pm0.006$ & $0.128\pm0.007$ & $0.090\pm0.006$ & $0.093\pm0.007$ \\
C05 & $0.068\pm0.004$ & $0.068\pm0.005$ & $0.110\pm0.004$ & $0.066\pm0.005$ & $0.079\pm0.006$ \\
C06 & $0.125\pm0.009$ & $0.126\pm0.009$ & $0.180\pm0.010$ & $0.131\pm0.010$ & $0.148\pm0.009$ \\
C07 & $0.161\pm0.012$ & $0.165\pm0.011$ & $0.230\pm0.012$ & $0.162\pm0.011$ & $0.174\pm0.012$ \\
C08 & $0.132\pm0.010$ & $0.133\pm0.010$ & $0.184\pm0.009$ & $0.132\pm0.009$ & $0.158\pm0.011$ \\
C09 & $0.190\pm0.014$ & $0.192\pm0.014$ & $0.266\pm0.013$ & $0.193\pm0.014$ & $0.195\pm0.015$ \\
C10 & $0.139\pm0.009$ & $0.139\pm0.009$ & $0.200\pm0.010$ & $0.139\pm0.009$ & $0.156\pm0.011$ \\
C11 & $0.417\pm0.074$ & $0.331\pm0.063$ & $0.710\pm0.112$ & $0.446\pm0.074$ & $0.490\pm0.056$ \\
C12 & $0.714\pm0.097$ & $0.439\pm0.058$ & $0.761\pm0.090$ & $0.706\pm0.095$ & $0.672\pm0.087$ \\
C13 & $0.501\pm0.084$ & $0.308\pm0.035$ & $0.701\pm0.093$ & $0.549\pm0.089$ & $0.693\pm0.099$ \\
C14 & $0.800\pm0.106$ & $0.629\pm0.082$ & $0.792\pm0.083$ & $0.858\pm0.110$ & $0.805\pm0.101$ \\
C15 & $0.482\pm0.079$ & $0.520\pm0.082$ & $0.686\pm0.085$ & $0.523\pm0.087$ & $0.720\pm0.108$ \\
C16 & $0.095\pm0.005$ & $0.099\pm0.006$ & $0.267\pm0.008$ & $0.096\pm0.005$ & $0.096\pm0.005$ \\
C17 & $1.891\pm0.050$ & $2.089\pm0.055$ & $1.850\pm0.048$ & $1.898\pm0.047$ & $1.860\pm0.064$ \\
C18 & $1.835\pm0.062$ & $2.018\pm0.061$ & $1.766\pm0.066$ & $1.820\pm0.066$ & $1.861\pm0.051$ \\
C19 & $0.072\pm0.005$ & $0.074\pm0.005$ & $0.101\pm0.005$ & $0.073\pm0.005$ & $0.072\pm0.007$ \\
C20 & $0.186\pm0.013$ & $0.191\pm0.013$ & $0.274\pm0.010$ & $0.189\pm0.013$ & $0.190\pm0.014$ \\
C21 & $0.137\pm0.011$ & $0.139\pm0.011$ & $0.224\pm0.010$ & $0.132\pm0.010$ & $0.166\pm0.013$ \\
C22 & $0.086\pm0.005$ & $0.086\pm0.005$ & $0.294\pm0.008$ & $0.087\pm0.005$ & $0.087\pm0.004$ \\
C23 & $0.159\pm0.010$ & $0.159\pm0.010$ & $0.669\pm0.038$ & $0.156\pm0.011$ & $0.160\pm0.012$ \\
\end{longtable}\endgroup

\par\begingroup\compacttable\setlength{\tabcolsep}{3pt}
\begin{longtable}{lcccc}
\caption{Coefficient error: mean aggregation and zero-filled comparators.}\\
\toprule ID & MN & UC & MS & TT\\\midrule\endfirsthead
\toprule ID & MN & UC & MS & TT\\\midrule\endhead
\bottomrule\endlastfoot
C01 & $0.063\pm0.005$ & $0.144\pm0.003$ & $0.134\pm0.007$ & $0.120\pm0.005$ \\
C02 & -- & $0.174\pm0.016$ & $0.141\pm0.007$ & $0.135\pm0.006$ \\
C03 & $0.061\pm0.004$ & $0.165\pm0.017$ & $0.147\pm0.007$ & $0.126\pm0.006$ \\
C04 & -- & $0.157\pm0.004$ & $0.141\pm0.009$ & $0.136\pm0.007$ \\
C05 & $0.063\pm0.004$ & $0.137\pm0.004$ & $0.156\pm0.006$ & $0.132\pm0.005$ \\
C06 & $0.125\pm0.009$ & $0.452\pm0.004$ & $0.257\pm0.010$ & $0.237\pm0.008$ \\
C07 & -- & $0.587\pm0.016$ & $0.291\pm0.009$ & $0.283\pm0.008$ \\
C08 & $0.127\pm0.009$ & $0.581\pm0.017$ & $0.293\pm0.010$ & $0.269\pm0.008$ \\
C09 & -- & $0.709\pm0.054$ & $0.258\pm0.010$ & $0.253\pm0.009$ \\
C10 & $0.130\pm0.008$ & $0.648\pm0.052$ & $0.267\pm0.011$ & $0.242\pm0.009$ \\
C11 & $0.417\pm0.074$ & $1.435\pm0.078$ & $0.311\pm0.013$ & $0.290\pm0.009$ \\
C12 & -- & $0.851\pm0.011$ & $0.420\pm0.013$ & $0.400\pm0.009$ \\
C13 & $0.460\pm0.085$ & $0.836\pm0.011$ & $0.430\pm0.012$ & $0.402\pm0.010$ \\
C14 & -- & $1.393\pm0.035$ & $0.409\pm0.010$ & $0.397\pm0.009$ \\
C15 & $0.555\pm0.096$ & $1.387\pm0.044$ & $0.422\pm0.011$ & $0.395\pm0.010$ \\
C16 & -- & $0.090\pm0.006$ & $0.215\pm0.006$ & $0.200\pm0.006$ \\
C17 & -- & $2.214\pm0.044$ & $0.436\pm0.016$ & $0.430\pm0.015$ \\
C18 & $1.874\pm0.060$ & $2.280\pm0.034$ & $0.452\pm0.017$ & $0.418\pm0.014$ \\
C19 & -- & $0.072\pm0.005$ & $0.101\pm0.006$ & $0.101\pm0.005$ \\
C20 & -- & $0.589\pm0.044$ & $0.250\pm0.011$ & $0.241\pm0.010$ \\
C21 & $0.134\pm0.011$ & $0.517\pm0.037$ & $0.270\pm0.011$ & $0.241\pm0.009$ \\
C22 & -- & $0.094\pm0.006$ & $0.702\pm0.011$ & $0.692\pm0.009$ \\
C23 & -- & $0.971\pm0.011$ & $1.192\pm0.013$ & $1.184\pm0.010$ \\
\end{longtable}\endgroup

\par\begingroup\compacttable\setlength{\tabcolsep}{3pt}
\begin{longtable}{lrrrrrrrrr}
\caption{Successful fits/exact recoveries out of 30 for each method and configuration.}\\
\toprule ID & TS & TM & TL & VR & MD & MN & UC & MS & TT\\\midrule\endfirsthead
\toprule ID & TS & TM & TL & VR & MD & MN & UC & MS & TT\\\midrule\endhead
\bottomrule\endlastfoot
C01 & 30/19 & 30/19 & 30/4 & 30/16 & 30/20 & 30/19 & 30/20 & 30/4 & 30/5 \\
C02 & 30/16 & 30/18 & 30/1 & 30/19 & 30/20 & 0/0 & 30/23 & 30/2 & 30/3 \\
C03 & 30/16 & 30/18 & 30/1 & 30/17 & 30/20 & 30/19 & 30/25 & 30/1 & 30/2 \\
C04 & 30/21 & 30/21 & 30/0 & 30/22 & 30/21 & 0/0 & 30/26 & 30/0 & 30/1 \\
C05 & 30/23 & 30/24 & 30/6 & 30/25 & 30/13 & 30/21 & 30/24 & 30/0 & 30/0 \\
C06 & 30/26 & 30/26 & 30/9 & 30/23 & 30/19 & 30/26 & 30/29 & 30/0 & 30/0 \\
C07 & 30/25 & 30/25 & 30/13 & 30/25 & 30/18 & 0/0 & 30/30 & 30/3 & 30/3 \\
C08 & 30/25 & 30/24 & 30/7 & 30/26 & 30/22 & 30/25 & 30/30 & 30/2 & 30/2 \\
C09 & 30/28 & 30/27 & 30/9 & 30/29 & 30/28 & 0/0 & 30/30 & 30/0 & 30/0 \\
C10 & 30/19 & 30/22 & 30/5 & 30/20 & 30/23 & 30/22 & 30/30 & 30/0 & 30/0 \\
C11 & 30/23 & 30/26 & 30/13 & 30/21 & 30/21 & 30/23 & 30/30 & 30/0 & 30/0 \\
C12 & 30/13 & 30/21 & 30/13 & 30/17 & 30/13 & 0/0 & 30/30 & 30/7 & 30/8 \\
C13 & 30/25 & 30/25 & 30/14 & 30/25 & 30/17 & 30/23 & 30/30 & 30/5 & 30/8 \\
C14 & 30/16 & 30/17 & 30/18 & 30/15 & 30/15 & 0/0 & 30/30 & 30/0 & 30/0 \\
C15 & 30/26 & 30/26 & 30/20 & 30/29 & 30/23 & 30/29 & 30/29 & 30/0 & 30/0 \\
C16 & 30/22 & 30/20 & 30/13 & 30/20 & 30/19 & 0/0 & 30/23 & 30/1 & 30/1 \\
C17 & 30/3 & 30/0 & 30/3 & 30/3 & 30/4 & 0/0 & 30/22 & 30/0 & 30/0 \\
C18 & 30/3 & 30/0 & 30/4 & 30/3 & 30/3 & 30/2 & 30/21 & 30/0 & 30/0 \\
C19 & 30/19 & 30/19 & 30/6 & 30/21 & 30/21 & 0/0 & 30/19 & 30/3 & 30/6 \\
C20 & 30/27 & 30/28 & 30/20 & 30/30 & 30/29 & 0/0 & 30/30 & 30/0 & 30/0 \\
C21 & 30/23 & 30/27 & 30/11 & 30/23 & 30/26 & 30/22 & 30/30 & 30/0 & 30/0 \\
C22 & 30/24 & 30/23 & 30/2 & 30/23 & 30/19 & 0/0 & 30/5 & 30/0 & 30/0 \\
C23 & 30/25 & 30/25 & 30/3 & 30/22 & 30/24 & 0/0 & 30/8 & 30/0 & 30/0 \\
\end{longtable}\endgroup

Complete selection diagnostics, paired differences and computation records are in
\path{ancillary_tables/competitive_diagnostics.csv},
\path{ancillary_tables/competitive_paired.csv}, and
\path{ancillary_tables/competitive_costs.csv}, with column definitions and reporting
conventions in their notes. Figure~\ref{supp-fig:competitivepaired} shows the paired reference comparison.

\subsection{Prediction and Support Sensitivity}\label{supp-sec:multimetric38}
Table~\ref{supp-tab:prediction38} supplements coefficient errors with paired test-MSE differences. The saved validation-selected estimates are unchanged. All 23 primary settings and all six sample-size follow-up settings are re-evaluated at ten common absolute thresholds: $10^{-6},10^{-4},10^{-3},0.01,0.025,0.05,0.1,0.2,0.4,0.6$. Thresholds change only the support evaluation; no threshold is selected using the true coefficients, and neither refitting nor prediction evaluation follows thresholding. Failed fits count as zero in unconditional support summaries; conditional summaries report their denominators separately.

\begin{table}[H]
\caption{Paired test-MSE differences, Trim-SCAD minus each comparator. Entries are means with one Monte Carlo standard error in parentheses; each uses 30 jointly successful repetitions. Negative differences favor Trim-SCAD. The follow-up uses new seeds.}\label{supp-tab:prediction38}
\centering\compacttable\setlength{\tabcolsep}{3pt}
\begin{tabular}{rllrrr}\toprule
$q$ & $n_j$ & Attack & Trim-Lasso & VRMOM-SCAD & T-SLARD-ZF\\\midrule
\multicolumn{6}{l}{\textit{Fixed local sample size}}\\
0.7 & 512 & Gaussian & $-0.0437$ (0.0054) & $-0.0000$ (0.0013) & $-0.0079$ (0.0027)\\
0.7 & 512 & In-range & $-0.0329$ (0.0049) & $-0.0004$ (0.0009) & $-0.0114$ (0.0019)\\
0.5 & 512 & Gaussian & $-0.5694$ (0.1615) & $-0.0465$ (0.0402) & $0.2941$ (0.1197)\\
0.5 & 512 & In-range & $-0.6109$ (0.1586) & $-0.0307$ (0.0601) & $0.0897$ (0.1048)\\
\multicolumn{6}{l}{\textit{Sample-size follow-up}}\\
0.7 & 512 & Gaussian & $-0.0450$ (0.0043) & $-0.0002$ (0.0009) & $-0.0188$ (0.0024)\\
0.7 & 512 & In-range & $-0.0285$ (0.0038) & $-0.0002$ (0.0005) & $-0.0173$ (0.0021)\\
0.5 & 1004 & Gaussian & $-0.0956$ (0.0175) & $-0.0013$ (0.0015) & $-0.0272$ (0.0040)\\
0.5 & 1004 & In-range & $-0.1084$ (0.0260) & $0.0025$ (0.0031) & $-0.0373$ (0.0052)\\
0.3 & 2788 & Gaussian & $-0.3873$ (0.1112) & $-0.0074$ (0.0100) & $-0.2686$ (0.0201)\\
0.3 & 2788 & In-range & $-0.4072$ (0.1155) & $-0.0010$ (0.0025) & $-0.2743$ (0.0141)\\
\bottomrule\end{tabular}
\end{table}

Figure~\ref{supp-fig:thresholds38} shows the primary two-attacker settings at $q=0.7,0.5$. T-SLARD-ZF returns additional nonzero coefficients under the original support definition: its exact-recovery count changes from zero to 30 at threshold $0.2$ for $q=0.7$ and at $0.4$ for $q=0.5$, under both attacks. Exact-support rankings therefore depend on the threshold used to define selected variables. The ancillary tables retain all methods, settings and thresholds, together with paired error differences, Monte Carlo standard errors and favorable-pair counts.

\begin{figure}[H]
\centering\includegraphics[width=\textwidth,height=.72\textheight,keepaspectratio]{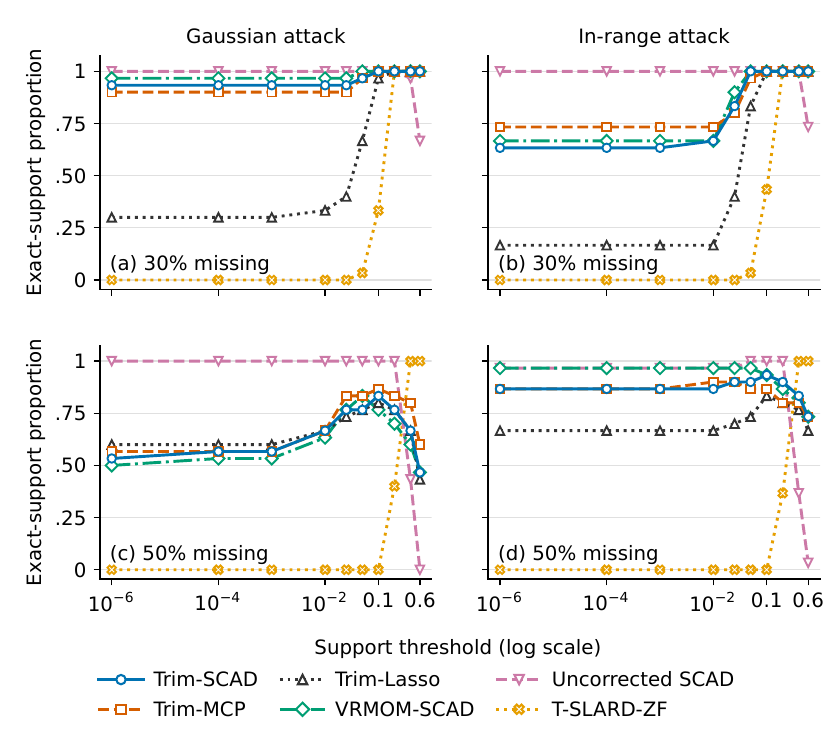}
\caption{Sensitivity of exact support recovery to a common absolute coefficient threshold. Curves use all 30 repetitions, counting failures as incorrect. Columns show Gaussian and in-range attacks; rows show 30\% and 50\% missingness at $n_j=512$. All ten thresholds are evaluation settings; none is used to select or refit an estimator.}
\label{supp-fig:thresholds38}
\end{figure}

\subsection{Sample Size and Missingness}\label{supp-sec:competitive_sample_size}
To examine the error--support tradeoff at greater missingness, a follow-up planned after the fixed-sample results uses new seeds and $n_j=\lceil512(0.7/q)^2\rceil$: sizes $512$, $1004$ and $2788$ for $q=0.7,0.5,0.3$. Each of six observation-probability--attack combinations has 30 repetitions, $m=12$, $p=256$ and $m_{\rm B}=b=2$. The covariance, signals, noise, nine methods, grid and stopping rules are unchanged.

Methods and attacks share data within $(q,n_j)$, but the three cohorts are
independent rather than nested. Increasing samples according to $q^{-2}$ does
not imply equal errors, since variance also depends on signal and correlation.
The tables retain all nine methods from both studies, including failures.

Errors are conditional means $\pm$ Monte Carlo SE; $S/E$ counts successful
fits/exact recoveries, each out of 30. The independent new cohorts use
$n_j=512,1004,2788$ at $q=0.7,0.5,0.3$, respectively. F01/F02, F03/F04 and
F05/F06 denote Gaussian/in-range attacks at those three rates.
Their fixed-$n_j=512$ counterparts are C09/C10, C14/C15 and C17/C18 in
Table~\ref{supp-tab:competitive-accuracy}. Full side-by-side values are in
\path{ancillary_tables/competitive_sample_size.csv}; cohorts are independent,
so no paired SE is computed between sample sizes.
\begin{center}\compacttable\setlength{\tabcolsep}{3pt}
\begin{tabular}{lccc}\toprule
Method & F01 & F02 & F03\\\midrule
Trim-SCAD & \shortstack{$0.179\pm0.012$\\30/28} & \shortstack{$0.144\pm0.009$\\30/19} & \shortstack{$0.214\pm0.015$\\30/20} \\
Trim-MCP & \shortstack{$0.179\pm0.012$\\30/29} & \shortstack{$0.143\pm0.010$\\30/21} & \shortstack{$0.208\pm0.013$\\30/26} \\
Trim-Lasso & \shortstack{$0.272\pm0.012$\\30/9} & \shortstack{$0.206\pm0.009$\\30/4} & \shortstack{$0.333\pm0.026$\\30/17} \\
VRMOM-SCAD & \shortstack{$0.177\pm0.011$\\30/26} & \shortstack{$0.142\pm0.010$\\30/18} & \shortstack{$0.219\pm0.014$\\30/19} \\
Median-SCAD & \shortstack{$0.182\pm0.012$\\30/23} & \shortstack{$0.173\pm0.011$\\30/23} & \shortstack{$0.198\pm0.013$\\30/12} \\
Mean-SCAD & \shortstack{--\\0/0} & \shortstack{$0.131\pm0.008$\\30/15} & \shortstack{--\\0/0} \\
Uncorrected SCAD & \shortstack{$0.696\pm0.055$\\30/30} & \shortstack{$0.615\pm0.050$\\30/28} & \shortstack{$1.064\pm0.058$\\30/30} \\
M-SLARD-ZF & \shortstack{$0.285\pm0.013$\\30/0} & \shortstack{$0.285\pm0.011$\\30/0} & \shortstack{$0.313\pm0.009$\\30/0} \\
T-SLARD-ZF & \shortstack{$0.279\pm0.011$\\30/0} & \shortstack{$0.260\pm0.009$\\30/0} & \shortstack{$0.310\pm0.008$\\30/0} \\
\bottomrule\end{tabular}\end{center}
\begin{center}\compacttable\setlength{\tabcolsep}{3pt}
\begin{tabular}{lccc}\toprule
Method & F04 & F05 & F06\\\midrule
Trim-SCAD & \shortstack{$0.174\pm0.017$\\30/26} & \shortstack{$0.369\pm0.051$\\30/14} & \shortstack{$0.249\pm0.032$\\30/22} \\
Trim-MCP & \shortstack{$0.155\pm0.012$\\30/30} & \shortstack{$0.311\pm0.065$\\30/22} & \shortstack{$0.192\pm0.012$\\30/24} \\
Trim-Lasso & \shortstack{$0.310\pm0.038$\\30/17} & \shortstack{$0.608\pm0.103$\\30/7} & \shortstack{$0.567\pm0.102$\\30/14} \\
VRMOM-SCAD & \shortstack{$0.164\pm0.014$\\30/28} & \shortstack{$0.381\pm0.054$\\30/12} & \shortstack{$0.250\pm0.036$\\30/22} \\
Median-SCAD & \shortstack{$0.195\pm0.020$\\30/23} & \shortstack{$0.414\pm0.054$\\30/10} & \shortstack{$0.317\pm0.041$\\30/16} \\
Mean-SCAD & \shortstack{$0.159\pm0.013$\\30/28} & \shortstack{--\\0/0} & \shortstack{$0.235\pm0.034$\\30/24} \\
Uncorrected SCAD & \shortstack{$0.930\pm0.047$\\30/30} & \shortstack{$1.521\pm0.034$\\30/30} & \shortstack{$1.539\pm0.032$\\30/30} \\
M-SLARD-ZF & \shortstack{$0.331\pm0.008$\\30/0} & \shortstack{$0.485\pm0.006$\\30/0} & \shortstack{$0.473\pm0.005$\\30/0} \\
T-SLARD-ZF & \shortstack{$0.303\pm0.008$\\30/0} & \shortstack{$0.491\pm0.006$\\30/0} & \shortstack{$0.474\pm0.005$\\30/0} \\
\bottomrule\end{tabular}\end{center}

\subsection{Selected-Path Diagnostics}\label{supp-sec:competitive_endpoints}
The audit evaluates Theorem~\ref{supp-thm:endpoint} on all 1740 selected Trim-SCAD/MCP paths and 13920 endpoints (1380 fixed-sample and 360 follow-up paths), sharing 360 datasets. Seeds and data hashes verify reconstruction of unchanged fits, tuning and paths. Truth and honest-node identities enter only this post-fit evaluation.

Each round uses $U_t=\cS\cup\supp(x_t)$ with all nonzero coefficients, rather than the $10^{-6}$ reporting threshold. We compute $\kappa_t,\alpha_t,\zeta_t$ and full KKT residuals using frozen weights and the surrogate at $x_{t-1}$. Residuals match saved values within $1.6\times10^{-15}$ and meet tolerance (at most $10^{-5}$). All positive-curvature endpoints satisfy the one-step inequality. Geometric recursion also requires $\alpha_t$ uniformly below one and bounded $\zeta_t$; positive curvature alone controls neither contraction nor inner iterates.

Positive curvature holds in every round for 1329/1380 fixed-sample and 360/360 follow-up paths; all-round $\alpha_t<1$ holds for 97 fixed-sample paths, confined to clean/structural controls, and no follow-up path. Table~\ref{supp-tab:competitive-endpoints} retains all configurations, including negative curvature. Full extrema, inner supports and roundwise bounds are in \path{ancillary_tables/competitive_endpoints.csv} and \path{results/competitive_endpoint_audit.jsonl}.

At the main reference ($q=0.7$, $m_{\rm B}=b=2$), all 120 method--attack paths have positive endpoint curvature. Across their endpoints, $\kappa_t\ge0.168$ and $2.25\le\alpha_t\le14.09$, so this sufficient condition does not certify geometric contraction. The 30 shared master samples also fail the uniform sparse-curvature condition: if $r_D=\max_l\Gamma_{1,ll}/\min_l\Gamma_{1,ll}$, any admissible curvature bounds must satisfy $M/\kappa\ge r_D$. Hence
\[
 K\ge 2s+\left\lceil(32/3)^2(r_D-1)^2s\right\rceil.
\]
The lower bound is 178--422, so $4K\ge p=256$ would require full positive definiteness, whereas master minimum eigenvalues are $-0.356$ to $-0.329$. Thus the reference lies outside the certified geometric regime; separate conditional-theory experiments assess contraction under valid conditions.

\par\begingroup\compacttable\setlength{\tabcolsep}{3pt}
\begin{longtable}{lrrrrll}
\caption{Endpoint conditions in every configuration; S/M denotes SCAD/MCP. Pos. and Contr. count paths meeting each condition at all eight endpoints, out of 30. Curvature minima use all endpoints; $\alpha_t$ ranges use positive-curvature endpoints.}\label{supp-tab:competitive-endpoints}\\
\toprule ID & $\min\kappa$ S & $\min\kappa$ M & Pos. S/M & Contr. S/M & $\alpha$ S & $\alpha$ M\\\midrule\endfirsthead
\toprule ID & $\min\kappa$ S & $\min\kappa$ M & Pos. S/M & Contr. S/M & $\alpha$ S & $\alpha$ M\\\midrule\endhead
\bottomrule\endlastfoot
C01 & 0.148 & 0.223 & 30/30 & 16/16 & 0.45--6.47 & 0.48--1.88 \\
C02 & 0.161 & 0.220 & 30/30 & 0/0 & 1.04--13.58 & 1.04--3.52 \\
C03 & 0.189 & 0.192 & 30/30 & 0/0 & 0.90--9.62 & 0.90--9.13 \\
C04 & 0.236 & 0.236 & 30/30 & 0/0 & 1.87--4.81 & 1.87--4.81 \\
C05 & 0.211 & 0.208 & 30/30 & 0/0 & 1.54--12.93 & 1.54--13.26 \\
C06 & 0.170 & 0.172 & 30/30 & 15/11 & 0.66--3.96 & 0.66--4.57 \\
C07 & 0.112 & 0.112 & 30/30 & 0/0 & 1.37--14.23 & 1.37--13.91 \\
C08 & 0.183 & 0.183 & 30/30 & 0/0 & 1.37--4.48 & 1.37--4.21 \\
C09 & 0.173 & 0.168 & 30/30 & 0/0 & 2.25--13.45 & 2.25--14.09 \\
C10 & 0.183 & 0.183 & 30/30 & 0/0 & 2.25--8.01 & 2.25--10.71 \\
C11 & -0.012 & 0.073 & 29/30 & 0/0 & 1.13--12.94 & 1.13--6.12 \\
C12 & -0.012 & 0.073 & 29/30 & 0/0 & 2.18--13.60 & 2.18--10.35 \\
C13 & -0.012 & 0.073 & 29/30 & 0/0 & 2.18--10.49 & 2.18--10.35 \\
C14 & 0.064 & 0.073 & 30/30 & 0/0 & 3.60--25.90 & 3.60--20.63 \\
C15 & 0.073 & 0.073 & 30/30 & 0/0 & 3.60--16.92 & 3.60--16.92 \\
C16 & 0.620 & 0.561 & 30/30 & 12/11 & 0.77--1.56 & 0.80--1.90 \\
C17 & -0.370 & -0.196 & 18/17 & 0/0 & 7.06--240.97 & 7.06--301.61 \\
C18 & -0.442 & -0.442 & 20/17 & 0/0 & 7.16--240.97 & 7.16--240.97 \\
C19 & 0.215 & 0.225 & 30/30 & 0/0 & 1.43--11.80 & 1.43--11.52 \\
C20 & 0.096 & 0.093 & 30/30 & 0/0 & 2.65--31.47 & 2.65--27.08 \\
C21 & 0.124 & 0.137 & 30/30 & 0/0 & 2.71--24.19 & 2.65--21.24 \\
C22 & 0.266 & 0.253 & 30/30 & 8/8 & 0.74--5.38 & 0.74--5.92 \\
C23 & 0.166 & 0.163 & 30/30 & 0/0 & 2.47--11.93 & 2.47--12.19 \\
F01 & 0.133 & 0.138 & 30/30 & 0/0 & 2.60--19.10 & 2.60--20.39 \\
F02 & 0.169 & 0.144 & 30/30 & 0/0 & 2.60--12.51 & 2.60--13.76 \\
F03 & 0.125 & 0.125 & 30/30 & 0/0 & 2.11--9.67 & 2.11--9.67 \\
F04 & 0.126 & 0.142 & 30/30 & 0/0 & 2.11--12.46 & 2.11--6.85 \\
F05 & 0.076 & 0.076 & 30/30 & 0/0 & 1.82--16.21 & 1.59--11.36 \\
F06 & 0.099 & 0.099 & 30/30 & 0/0 & 1.59--9.82 & 1.59--9.16 \\
\end{longtable}\endgroup

\section{Sensitivity Benchmarks}\label{supp-contents-6}
\subsection{Experimental Protocol}
Sensitivity studies use equal sizes $n_j=n$ and seeds $20268008+10000c+r$ ($r=0,\ldots,29$) for synthetic condition $c$, or $20269008+r$ ($r=0,\ldots,9$) for real data. Methods share data and attack randomness. Complete validation selects the first grid minimizer; failed candidates receive infinite loss, and a repetition fails if all candidates fail. Error means condition on success; failure and support rates use all repetitions.

\subsection{Synthetic Stress Tests}\label{supp-sec:earlierbenchmark}
This benchmark uses $m=20$, $n_j=120$, $p=240$, four equally spaced coefficients of alternating sign and magnitude $1.2$, Gaussian AR(1) predictors and standard Gaussian noise. Its reference has $q=0.7$, correlation $0.3$ and $m_{\rm B}=b=2$. Other settings use $q=0.3$, correlation $0.9$, an in-range attack with $m_{\rm B}=b=4$, undertrimming with $m_{\rm B}=4,b=2$, or no attack with $m_{\rm B}=b=0$. Each setting has 30 repetitions and shared complete validation and test samples of sizes 400 and 2000. Seven methods use the four-point grid $\{0.3,0.6,1.2,2.4\}$. The errors below are conditional on at least one candidate reaching the specified KKT tolerance. Support recovery is over all 30 repetitions. The reported uncertainty is the standard error of the conditional mean.
\par\begingroup
\renewcommand{\baselinestretch}{1}\fontsize{10}{11.5}\selectfont
\renewcommand{\arraystretch}{1.08}\setlength{\tabcolsep}{2pt}
\begin{longtable}{lclcclc}
\toprule
 & \multicolumn{3}{c}{Reference} & \multicolumn{3}{c}{High missingness} \\
\cmidrule(lr){2-4}\cmidrule(lr){5-7}
Method & Success & $\ell_2$ error $\pm$ SE & Exact & Success & $\ell_2$ error $\pm$ SE & Exact \\
\midrule\endhead
\bottomrule\endlastfoot
Mean-SCAD & 0/30 & -- & 0\% & 0/30 & -- & 0\% \\
Trim-Lasso & 30/30 & 0.418 $\pm$ 0.035 & 57\% & 30/30 & 2.197 $\pm$ 0.066 & 20\% \\
Trim-SCAD & 30/30 & 0.141 $\pm$ 0.028 & 87\% & 30/30 & 2.197 $\pm$ 0.067 & 20\% \\
Trim-MCP & 30/30 & 0.124 $\pm$ 0.012 & 90\% & 30/30 & 2.018 $\pm$ 0.107 & 30\% \\
Median-SCAD & 30/30 & 0.150 $\pm$ 0.026 & 60\% & 30/30 & 2.232 $\pm$ 0.060 & 17\% \\
VRMOM-Lasso & 30/30 & 0.395 $\pm$ 0.033 & 50\% & 30/30 & 2.208 $\pm$ 0.063 & 17\% \\
VRMOM-SCAD & 30/30 & 0.139 $\pm$ 0.027 & 77\% & 30/30 & 2.172 $\pm$ 0.069 & 23\% \\
\end{longtable}
\begin{longtable}{lclcclc}
\toprule
 & \multicolumn{3}{c}{High correlation} & \multicolumn{3}{c}{In-range attack} \\
\cmidrule(lr){2-4}\cmidrule(lr){5-7}
Method & Success & $\ell_2$ error $\pm$ SE & Exact & Success & $\ell_2$ error $\pm$ SE & Exact \\
\midrule\endhead
\bottomrule\endlastfoot
Mean-SCAD & 0/30 & -- & 0\% & 30/30 & 0.163 $\pm$ 0.047 & 100\% \\
Trim-Lasso & 30/30 & 1.522 $\pm$ 0.098 & 17\% & 30/30 & 0.506 $\pm$ 0.030 & 97\% \\
Trim-SCAD & 29/30 & 0.643 $\pm$ 0.126 & 63\% & 30/30 & 0.184 $\pm$ 0.048 & 90\% \\
Trim-MCP & 29/30 & 0.528 $\pm$ 0.117 & 70\% & 30/30 & 0.165 $\pm$ 0.038 & 93\% \\
Median-SCAD & 29/30 & 0.825 $\pm$ 0.124 & 43\% & 30/30 & 0.181 $\pm$ 0.027 & 83\% \\
VRMOM-Lasso & 30/30 & 1.582 $\pm$ 0.095 & 13\% & 30/30 & 0.491 $\pm$ 0.030 & 100\% \\
VRMOM-SCAD & 29/30 & 0.813 $\pm$ 0.128 & 57\% & 30/30 & 0.172 $\pm$ 0.048 & 93\% \\
\end{longtable}
\begin{longtable}{lclcclc}
\toprule
 & \multicolumn{3}{c}{Undertrimming} & \multicolumn{3}{c}{No attack} \\
\cmidrule(lr){2-4}\cmidrule(lr){5-7}
Method & Success & $\ell_2$ error $\pm$ SE & Exact & Success & $\ell_2$ error $\pm$ SE & Exact \\
\midrule\endhead
\bottomrule\endlastfoot
Mean-SCAD & 0/30 & -- & 0\% & 30/30 & 0.095 $\pm$ 0.014 & 80\% \\
Trim-Lasso & 0/30 & -- & 0\% & 30/30 & 0.388 $\pm$ 0.036 & 57\% \\
Trim-SCAD & 0/30 & -- & 0\% & 30/30 & 0.095 $\pm$ 0.014 & 80\% \\
Trim-MCP & 0/30 & -- & 0\% & 30/30 & 0.087 $\pm$ 0.009 & 87\% \\
Median-SCAD & 30/30 & 0.130 $\pm$ 0.011 & 90\% & 30/30 & 0.135 $\pm$ 0.020 & 87\% \\
VRMOM-Lasso & 30/30 & 0.485 $\pm$ 0.030 & 67\% & 30/30 & 0.404 $\pm$ 0.038 & 53\% \\
VRMOM-SCAD & 30/30 & 0.124 $\pm$ 0.008 & 83\% & 30/30 & 0.096 $\pm$ 0.012 & 80\% \\
\end{longtable}
\endgroup

\subsection{No-Attack Trimming Budget}
The clean-data comparison reuses the 30 synthetic datasets with $c=5$ and $m_{\rm B}=0$. Trim-SCAD at $b=0,2,4$ shares initialization, the four-point grid and the budget-zero penalty schedule, varying only aggregation. All 90 selections succeed; \texttt{results/no\_attack\_budget.json} retains 360 candidate outcomes. Budget-zero selections reproduce the clean benchmark.

Table~\ref{supp-tab:cleanbudget} reports errors, paired standard errors and support
counts. Relative to $b=0$, error increases in 20/30 datasets at $b=2$ and
18/30 at $b=4$. Additional score, prediction and fixed-penalty comparisons
are retained in \path{ancillary_tables/no_attack_extra_notes.tex}.

\begin{table}[H]
\caption{No-attack trimming comparison on 30 paired datasets. Errors are means $\pm$ SE; $\Delta$ is the paired increase in $\ell_2$ error relative to $b=0$.}
\label{supp-tab:cleanbudget}\centering\small
\setlength{\tabcolsep}{12pt}
\begin{tabular}{cccc}\toprule
$b$ & $\ell_2$ error & Paired $\Delta$ & Exact support\\\midrule
0 & $0.0948\pm0.0144$ & --- & 24/30\\
2 & $0.1007\pm0.0153$ & $0.0059\pm0.0026$ & 21/30\\
4 & $0.1011\pm0.0138$ & $0.0063\pm0.0043$ & 24/30\\\bottomrule
\end{tabular}
\end{table}

\subsection{Real Data}\label{supp-sec:earlierreal}
The superconductivity response is critical temperature. Complete training data determine scaling before masking. Real-data SEs describe variability across overlapping random splits. All four candidate outcomes are retained.
All conditions have $q=0.7$. Gaussian attack uses $m_{\mathrm B}=b=2$, stealth uses $m_{\mathrm B}=b=4$, and clean data use $m_{\mathrm B}=b=0$. Without a true coefficient vector for the real dataset, the stealth direction follows the sign of the honest median gradient. The training-mean baseline gives a mean RMSE of $34.275$ on the same splits.
\begin{center}\compacttable\setlength{\tabcolsep}{3pt}
\begin{tabular}{lccc}\toprule
Method & Clean & Gaussian & In-range\\\midrule
Mean-SCAD & \shortstack{21.125 $\pm$ 0.251\\10/10} & \shortstack{--\\0/10} & \shortstack{21.438 $\pm$ 0.049\\10/10} \\
Trim-Lasso & \shortstack{21.246 $\pm$ 0.046\\9/10} & \shortstack{21.542 $\pm$ 0.154\\9/10} & \shortstack{21.563 $\pm$ 0.196\\10/10} \\
Trim-SCAD & \shortstack{21.125 $\pm$ 0.251\\10/10} & \shortstack{21.166 $\pm$ 0.220\\10/10} & \shortstack{21.369 $\pm$ 0.062\\10/10} \\
Trim-MCP & \shortstack{21.214 $\pm$ 0.129\\10/10} & \shortstack{21.032 $\pm$ 0.277\\10/10} & \shortstack{21.413 $\pm$ 0.265\\10/10} \\
Median-SCAD & \shortstack{21.274 $\pm$ 0.224\\10/10} & \shortstack{21.420 $\pm$ 0.269\\10/10} & \shortstack{21.580 $\pm$ 0.218\\10/10} \\
VRMOM-Lasso & \shortstack{21.257 $\pm$ 0.047\\9/10} & \shortstack{21.684 $\pm$ 0.130\\10/10} & \shortstack{21.448 $\pm$ 0.168\\10/10} \\
VRMOM-SCAD & \shortstack{21.278 $\pm$ 0.217\\10/10} & \shortstack{21.003 $\pm$ 0.196\\10/10} & \shortstack{21.368 $\pm$ 0.071\\10/10} \\
\bottomrule\end{tabular}\end{center}
Entries give RMSE (mean $\pm$ SE) and successful fits out of ten.

\section{Theory Validation}\label{supp-contents-7}\label{supp-sec:theoryvalidation}
\subsection{MCAR Samples}\label{supp-sec:theorymcar}
The MCAR audit uses $m=8,p=24$, independent standard Gaussian covariates,
Gaussian noise SD $0.5$, and signals $\theta^\star_1=1.2$,
$\theta^\star_{13}=-1.2$. The reference is $n_j=320000,q=0.7,b=m_{\rm B}=1$;
Table~\ref{supp-tab:theorymcar} lists the variants. The last $m_{\rm B}$ nodes are Byzantine. Adaptive reports select
the honest 30th/70th percentile maximizing the coordinatewise absolute
difference between trimmed and honest-mean gradients. Gaussian reports are
$N(0,20^2I)$; clean and two-attack settings have $b=m_{\rm B}=0,2$.

Set $\kappa=\lambda_{\min}(\Gamma_1)$, $M=\lambda_{\max}(\Gamma_1)$ and
compute $K$. All sampled masters are positive definite, so these bounds cover
every sparse block. For $D_j=\Gamma_j-\Gamma_{\cH}$, sort each row by decreasing
absolute magnitude and take
\begin{equation*}
 d=\max_{j\in\cH,l}\left\{\sum_{k=1}^{r}D_{j,l,(k)}^2\right\}^{1/2},\qquad r=\min(K+s,p).
\end{equation*}
Cauchy--Schwarz bounds $\|D_je\|_\infty\le d\|e\|_2$ for every $r$-sparse
$e$. Together with exact $a_H=\|g_{\cH}(\theta^\star)\|_\infty$ and
$\tau=\max_{j\in\cH}\|g_j(\theta^\star)-g_{\cH}(\theta^\star)\|_\infty$,
these determine $A,B,\rho$ for a truth-assisted check of conditional theory,
without estimating a high-probability event.

Settings share 20 prespecified seeds, with nested samples across sizes and shared latent data, noise and mask uniforms across $q$. Two separate pilot seeds are excluded. SCAD starts at zero with $R_0=\sqrt{2}(1.2)$, $\eta=12/13$, and 12 rounds, without tuning. Samples with $\rho\ge1$ remain inadmissible. Solver-equivalence checks are retained with the audit records.

\begin{table}[H]
\caption{MCAR conditional-theorem checks, with 20 repetitions per setting and departures from the reference in the first column. Valid counts $\rho<1$; $\rho$ is the median over all 20 samples. Among valid fits, Error is mean $\pm$ SE, $R_{12}$ is the median bound, $I$ counts $R_\infty<R_0$ and $S$ counts $1.2>2R_{12}$.}
\label{supp-tab:theorymcar}
\centering\small\setlength{\tabcolsep}{4pt}
\begin{tabular}{lccccc}\toprule
Setting & Valid & $\rho$ & $\ell_2$ error & $R_{12}$ & $I,S$\\\midrule
$n_j=20000$ & 0/20 & 2.358 & --- & --- & $0,0$\\
$n_j=80000$ & 4/20 & 1.063 & $1.5841\pm0.1044$ & 3.644 & $0,0$\\
Reference & 20/20 & 0.483 & $0.0011\pm0.0001$ & 0.247 & $20,20$\\
$q=0.5$ & 20/20 & 0.692 & $0.0218\pm0.0095$ & 0.741 & $20,2$\\
$q=0.9$ & 20/20 & 0.387 & $0.0007\pm0.0001$ & 0.115 & $20,20$\\
Clean & 20/20 & 0.333 & $0.0010\pm0.0001$ & 0.069 & $20,20$\\
Two attacks & 20/20 & 0.792 & $0.4408\pm0.0510$ & 1.304 & $18,0$\\
Gaussian & 20/20 & 0.483 & $0.0009\pm0.0001$ & 0.247 & $20,20$\\\bottomrule
\end{tabular}
\end{table}

Of 160 cases, 124 satisfy $\rho<1$ and complete all rounds; 36 are inadmissible. All 3257 updates meet support/neighborhood bounds with at most two nonzeros, and all endpoints meet KKT tolerances and outer bounds. Bounds improve on $R_0$ in 118 valid cases; all 82 satisfying Corollary~\ref{supp-cor:support} recover support after thresholding. At $n_j=80000$, all four valid fits have $R_\infty>R_0$: contraction alone need not give an informative radius. Larger samples and observation probabilities reduce radii, which remain above observed errors.

\path{results/theory_mcar_plan.json}, \path{results/theory_mcar.json} and \path{results/theory_mcar_moments_*.npz} retain plans, complete traces, corrected moments and inadmissible cases, permitting independent recomputation of every supplied constant.

\subsection{Indefinite Quadratics}\label{supp-sec:theoryindefinite}
For the globally indefinite construction, set $p=4000$, $m=10$, $s=2$, and $\theta^\star=(1.2,-1.2,0,\ldots,0)^\top$. Let
\begin{equation*}
 \Gamma_j=(1+\delta_j)I-1.2vv^\top,\qquad v=p^{-1/2}\mathbf1,
\end{equation*}
where $\delta_j$ is uniform on $[-0.004,0.004]$. For honest nodes, set $r_j=\mu+\xi_j-\overline\xi_{\cH}$ and $\gamma_j=\Gamma_j\theta^\star-r_j$, with independent coordinatewise uniform $\mu_l\in[-0.0005,0.0005]$ and $\xi_{jl}\in[-0.001,0.001]$. The honest average score is $\mu$, its maximum node deviation is at most $0.002$, and $\Gamma_j-\Gamma_{\cH}=(\delta_j-\overline\delta_{\cH})I$. Thus fixed inputs $a_H=0.0005$, $\tau=0.002$, and $d=0.008$ are valid for every generated instance.

Choose $\kappa=0.98$ and $M=1.01$, giving $K=5$. Every principal block with $|U|\le4K=20$ satisfies
\begin{equation*}
 \begin{aligned}
 \lambda_{\min}(\Gamma_{j,UU})&\ge1-0.004-1.2(20/4000)=0.990,\\
 \lambda_{\max}(\Gamma_{j,UU})&\le1.004.
 \end{aligned}
\end{equation*}
Since $v^\top\Gamma_jv=1+\delta_j-1.2\le-0.196$, these losses are globally indefinite while their required sparse blocks have positive curvature, isolating the conditions of main-text Theorem~\ref{thm:outer}.

Seven settings share 20 prespecified seeds, zero initialization, $R_0=\sqrt{2}(1.2)$, frozen within-round SCAD shapes, $\eta=12/13$ and 12 rounds. Clean settings share data and recompute penalties as trimming budgets change. Gaussian attacks replace the last $m_{\rm B}$ reports by $N(0,20^2I)$. Each round, adaptive attacks choose just inside either honest coordinate-range endpoint to maximize the absolute surrogate score at truth. Attacks may use truth; fitting uses the fixed analytical bounds above. Separate pilots check solver equivalence.

\begin{table}[H]
\caption{Controlled indefinite quadratics, with 20 repetitions per setting. Error is mean $\pm$ SE in units of $10^{-4}$; theoretical constants, including $\rho$ and $R_{12}$, are fixed within each row.}
\label{supp-tab:theoryindefinite}
\centering\small\setlength{\tabcolsep}{4pt}
\begin{tabular}{lcccc}\toprule
Attack & $(b,m_{\rm B})$ & $\rho$ & $R_{12}$ & Error $(10^{-4})$\\\midrule
None & $(0,0)$ & 0.185 & 0.01416 & $3.800\pm0.309$\\
None & $(1,0)$ & 0.254 & 0.03869 & $3.849\pm0.328$\\
None & $(2,0)$ & 0.369 & 0.09155 & $3.928\pm0.334$\\
Gaussian & $(1,1)$ & 0.254 & 0.03869 & $3.887\pm0.308$\\
Gaussian & $(2,2)$ & 0.369 & 0.09155 & $3.928\pm0.336$\\
Adaptive & $(1,1)$ & 0.254 & 0.03869 & $5.155\pm0.309$\\
Adaptive & $(2,2)$ & 0.369 & 0.09155 & $6.893\pm0.327$\\\bottomrule
\end{tabular}
\end{table}

All 3265 updates stay in their prescribed neighborhoods with at most two nonzeros ($K=5$). At each level, active-block equations and full KKT checks give an exact fixed point; every recorded distance contracts by at most $1-\kappa/M$. All 140 runs satisfy outer bounds, final signal conditions and thresholded support recovery. Increasing trimming enlarges the guaranteed radius, without proportional changes in observed error. Construction, seeds, paths and solver comparisons are retained in \path{results/theory_indefinite_plan.json}, \path{results/theory_indefinite.json} and \path{results/theory_indefinite_core_equivalence.json}.

\subsection{Score Variance}\label{supp-sec:theoryvariance}
We check the variance identity using $p=3$, $\theta^\star=(0.7,1/\sqrt2,-1/\sqrt2)^\top$, independent standard Gaussian noise, and Gaussian predictors with covariance $\Sigma_{lk}=c^{|l-k|}$, where $c\in\{0,0.5\}$. For the first score coordinate, the Gaussian fourth-moment identity gives
\begin{equation*}
 \nu_1^2=\frac1q+\frac{1-q}{q^2}(1+c^2+c^4).
\end{equation*}
Use $q\in\{0.2,0.3,0.5,0.7,0.9,1\}$ and 30 independent repetitions of 50000 observations, sharing Gaussian draws, noise and mask uniforms across settings within repetitions. All 360 estimates are retained. Direct corrected scores and their separate noise/missing-predictor components agree pointwise within $10^{-10}$.

Mean empirical variances differ from the formula by at most $0.58\%$ across 12 settings. At $q=0.2$, empirical/theoretical values are $24.855/25$ for independent predictors and $31.209/31.25$ for $c=0.5$; at $q=1$, the missing-predictor component is zero. \path{results/theory_variance.json} retains repetitions and components; \path{results/theory_variance_plan.json} gives the fixed design and seeds.

\section{Additional Checks}\label{supp-contents-8}\label{supp-sec:additionalchecks}
The earlier checks below use distinct designs. Full records remain in the reproducibility archive; the last column identifies their protocols and summaries.
\begin{table}[!htbp]
\caption{Earlier checks and their retained outcomes.}
\centering\setlength{\tabcolsep}{4pt}
\begin{tabular}{>{\raggedright\arraybackslash}p{.18\textwidth}>{\raggedright\arraybackslash}p{.61\textwidth}>{\raggedright\arraybackslash}p{.12\textwidth}}\toprule
Study & Outcomes and interpretation & Record\\\midrule
Fixed tuning & Zero estimates at 70\% missingness; success in 23/50 highly correlated samples. Different tuning, seeds and caps prevent direct comparison with the validation-tuned benchmark. All 11700 records and earlier communication-kernel measurements are retained. & [A]\\[3pt]
Bound-based selection & All 30 indefinite quadratic selections pass error and support checks. For MCAR data ($p=12,m=10,s=1$), all ten repetitions across two observation rates admit candidates at $n_j=200000$; smaller samples include empty admissions and uninformative radii. Inputs use truth. & [B]\\[3pt]
Reference curvature & The earlier $p=240,n_1=120$ reference fails Proposition~\ref{supp-prop:rank}'s condition in all 30 samples. The primary reference is audited in Section~\ref{supp-sec:competitive_endpoints}. & [C]\\[3pt]
Local check and calibration & All 30 selected Trim-SCAD fits reproduce to $10^{-10}$. All final endpoints and 976/997 stage endpoints pass the local check; maximum support is 4--12. Exact recovery is 26/30 at thresholds $10^{-8},10^{-6},10^{-4},10^{-3}$. Independent masked calibration uses three prespecified sizes, including empty admissions. & [D]\\\bottomrule
\end{tabular}
\par\smallskip
\noindent Records: [A]~\path{ancillary_tables/exploratory_and_kernel.tex}; [B]~\path{ancillary_tables/earlier_bound_audits.tex} and \path{results/bridge_table.tex}; [C]~\path{ancillary_tables/earlier_reference_curvature.tex}; [D]~\path{ancillary_tables/additional_certificate_checks.tex} (also the earlier node-count study) and \path{results/calibration_table.tex}. The current theory and trusted-only comparisons are in Sections~\ref{supp-sec:theoryvalidation} and~\ref{supp-sec:trustedcomparison}.
\end{table}
\section{Additional Real Data}\label{supp-contents-9}
Abalone has 4177 records~\cite{suppbib-abalone}. Seven continuous measurements predict ring count, excluding sex. Each split uses 2400 training, 800 validation and 977 test records, with $m=10,n=240,m_{\rm B}=b=1$.

For superconductivity~\cite{suppbib-uci}, matching response rows align \texttt{train.csv} and \texttt{unique\_m.csv}. Element counts are normalized by row sums and rounded to ten decimals to define composition groups. Random group-level splits allocate 70\%, 15\% and the remainder to training, validation and testing. Training is shuffled and truncated by fewer than 20 rows for equal blocks, with $m=20,m_{\rm B}=b=2$. Records retain partition sizes and group counts.

Both tasks use $q=0.7$, Gaussian replacement attacks and the same four-candidate grid. Centering and scaling use complete training data before masking; validation and test covariates remain complete. This models artificial missingness after preprocessing. Ten shared splits give 60 selected fits and 240 candidate fits. The table uses conditional means and sample standard deviations (SD); failure counts retain all splits.
\begin{center}\small
\begin{tabular}{llrrrr}\toprule
Task & Method & Success & RMSE & SD & Mean baseline\\\midrule
Abalone & Trim-SCAD & 8/10 & 2.444 & 0.109 & 3.221\\
Abalone & Trim-Lasso & 10/10 & 2.591 & 0.153 & 3.221\\
Abalone & VRMOM-SCAD & 9/10 & 2.433 & 0.136 & 3.221\\
Grouped SC & Trim-SCAD & 10/10 & 21.407 & 0.926 & 34.307\\
Grouped SC & Trim-Lasso & 9/10 & 21.830 & 0.521 & 34.307\\
Grouped SC & VRMOM-SCAD & 10/10 & 21.414 & 0.911 & 34.307\\
\bottomrule

\end{tabular}
\end{center}
Paired differences use only splits where both compared methods succeed. Negative differences favor Trim-SCAD. Since random splits overlap, the SD describes split variability; it is not converted into a confidence interval assuming independent datasets.
\begin{center}\small
\begin{tabular}{llrrr}\toprule
Task & Contrast & Pairs & Mean difference & SD\\\midrule
Abalone & Trim-SCAD minus Trim-Lasso & 8 & -0.153 & 0.154\\
Abalone & Trim-SCAD minus VRMOM-SCAD & 8 & -0.001 & 0.092\\
Grouped SC & Trim-SCAD minus Trim-Lasso & 9 & -0.526 & 0.599\\
Grouped SC & Trim-SCAD minus VRMOM-SCAD & 10 & -0.006 & 0.064\\
\bottomrule

\end{tabular}
\end{center}

\section{Computational Cost}\label{supp-contents-10}
\subsection{Full Validation Tuning}\label{supp-sec:fulltuning38}
The primary design with $m=12,n_j=512,p=256$ is evaluated at $q=0.7,0.5$ under two Gaussian or in-range attackers. Five methods share 30 new seeds per setting, the seven-factor grid, validation and test samples, and eight-round limit. Method order is randomized within datasets. An independent 20-run pilot exactly reproduces the serial benchmark's selected coefficients, validation values and stopping records. The main study retains all 600 runs and 4200 candidates.

Windows processes exchange coefficient and gradient vectors through pipes, using one BLAS thread per process. Timing includes rebuilding each method's moments and every candidate's initialization, gradient exchange, aggregation, optimization, validation and selection. Workers compute gradients before attacks replace received reports at the coordinator. Shared dataset preparation and worker startup are recorded separately; testing occurs after timing. The pool is reused within a dataset, with moments rebuilt for each method.

Table~\ref{supp-tab:processfullcost} reports all five methods and four settings. Time and payload medians use all 30 repetitions; coefficient errors condition on success. Vector payload sums all candidates, including failures, at $16(m-1)p$ bytes per completed round, with MB denoting $10^6$ bytes. Control messages, operating-system framing and startup traffic are excluded from this payload; control messages are counted separately. Final error targets $0.1$, $0.2$ and $0.5$ evaluate the selected estimate after tuning and do not guide stopping or selection. Ancillary CSVs provide prediction, support, targets, failures and paired differences; complete records and setup costs accompany the code.

\begin{table}[H]
\caption{Full validation-tuning cost and final estimation error in the process study. Each setting uses 30 paired repetitions. Success counts successful selected fits out of 30; L2 gives their mean $\pm$ Monte Carlo SE. Time and MB are medians over all repetitions, including every candidate failure. MB denotes $10^6$ actual coefficient/gradient bytes. Shared data preparation and process startup are separate.}
\label{supp-tab:processfullcost}
\centering\small\setlength{\tabcolsep}{3pt}
\begin{tabular}{llcrrr}\toprule
$q$/Attack & Method & Success & L2 & Time (s) & MB\\\midrule
0.7/G & Trim-SCAD & 30/30 & $0.166\pm 0.011$ & 0.138 & 1.802 \\
0.7/G & Trim-MCP & 30/30 & $0.162\pm 0.010$ & 0.132 & 1.802 \\
0.7/G & Trim-Lasso & 30/30 & $0.257\pm 0.021$ & 0.140 & 1.802 \\
0.7/G & VRMOM-SCAD & 30/30 & $0.168\pm 0.012$ & 0.236 & 1.802 \\
0.7/G & T-SLARD-ZF & 30/30 & $0.255\pm 0.009$ & 0.204 & 2.478 \\
\midrule
0.7/I & Trim-SCAD & 30/30 & $0.143\pm 0.010$ & 0.133 & 1.802 \\
0.7/I & Trim-MCP & 30/30 & $0.143\pm 0.009$ & 0.128 & 1.802 \\
0.7/I & Trim-Lasso & 30/30 & $0.204\pm 0.011$ & 0.128 & 1.802 \\
0.7/I & VRMOM-SCAD & 30/30 & $0.137\pm 0.009$ & 0.235 & 1.802 \\
0.7/I & T-SLARD-ZF & 30/30 & $0.247\pm 0.009$ & 0.190 & 2.478 \\
\midrule
0.5/G & Trim-SCAD & 30/30 & $0.897\pm 0.144$ & 0.103 & 1.442 \\
0.5/G & Trim-MCP & 30/30 & $0.638\pm 0.110$ & 0.099 & 1.442 \\
0.5/G & Trim-Lasso & 30/30 & $0.846\pm 0.103$ & 0.104 & 1.194 \\
0.5/G & VRMOM-SCAD & 30/30 & $0.942\pm 0.142$ & 0.181 & 1.442 \\
0.5/G & T-SLARD-ZF & 30/30 & $0.406\pm 0.010$ & 0.166 & 2.433 \\
\midrule
0.5/I & Trim-SCAD & 30/30 & $0.408\pm 0.088$ & 0.107 & 1.442 \\
0.5/I & Trim-MCP & 30/30 & $0.413\pm 0.098$ & 0.102 & 1.442 \\
0.5/I & Trim-Lasso & 30/30 & $0.729\pm 0.100$ & 0.105 & 1.442 \\
0.5/I & VRMOM-SCAD & 30/30 & $0.465\pm 0.103$ & 0.189 & 1.442 \\
0.5/I & T-SLARD-ZF & 30/30 & $0.391\pm 0.010$ & 0.163 & 2.433 \\
\bottomrule\end{tabular}
\par\smallskip\raggedright\footnotesize G: Gaussian reports; I: in-range reports. Every method searches the same seven factors, allowing up to eight rounds per candidate. Time includes method-specific moment construction, all candidate initializations, gradient transport, local solves, validation and selection; truth and test evaluation occur afterwards.
\end{table}

All 600 runs select a successful candidate; 1482 of the 4200 candidates fail and remain included in cost. At $q=0.7$, Trim-SCAD and VRMOM-SCAD have the same median vector payload. T-SLARD-ZF's larger payload partly reflects more candidates completing all eight rounds: it has no initialization failures, whereas each corrected method has 60 among 210 candidates in each reference setting. Payload differences across these methods therefore also reflect candidate stopping patterns.

\subsection{Fixed-Penalty Accuracy Targets}

The cost benchmark uses $m=8,n=1024$, $p\in\{1000,4000,10000\}$, $q\in\{0.5,0.7,0.9\}$, independent standard Gaussian features and noise, and the first four true coefficients alternating $\pm1.2$. Two workers send Gaussian replacements; trimming removes two values per tail. Ten seeds, nine configurations and three methods give 270 runs. All use multiplier $0.6$, eight rounds and one BLAS thread per process; method order is randomized within seeds. Spawned Windows processes exchange double-precision vectors through pipes.

Cost is measured to the first successful stage endpoint reaching the coefficient-error target. Initialization is eligible with zero payload. Truth is used retrospectively. Nonattainment includes solver failure and exhausted budgets; an early hit can precede later failure, so final success is reported separately. Cost medians condition on attainment and must be read with its counts.

\par\begingroup
\renewcommand{\baselinestretch}{1}\fontsize{10}{11.5}\selectfont
\renewcommand{\arraystretch}{1.08}\setlength{\tabcolsep}{3pt}
\begin{longtable}{rrlrrrr@{\hspace{8pt}}rrrr}
\toprule
 & & & \multicolumn{4}{c}{Target $0.1$} & \multicolumn{4}{c}{Target $0.2$}\\
\cmidrule(lr){4-7}\cmidrule(lr){8-11}
$q$ & $p$ & Method & Hits & Time (s) & MB & Round & Hits & Time (s) & MB & Round\\
\midrule\endhead
\bottomrule\endlastfoot
0.5 & 1000 & TS & 2/10 & 0.065 & 0.448 & 4.0 & 10/10 & 0.035 & 0.112 & 1.0 \\
0.5 & 1000 & TL & 0/10 & -- & -- & -- & 1/10 & 0.085 & 0.672 & 6.0 \\
0.5 & 1000 & VR & 0/10 & -- & -- & -- & 10/10 & 0.074 & 0.168 & 1.5 \\
0.5 & 4000 & TS & 1/10 & 0.444 & 2.240 & 5.0 & 6/10 & 0.292 & 0.896 & 2.0 \\
0.5 & 4000 & TL & 0/10 & -- & -- & -- & 0/10 & -- & -- & -- \\
0.5 & 4000 & VR & 0/10 & -- & -- & -- & 5/10 & 0.452 & 0.896 & 2.0 \\
0.5 & 10000 & TS & 0/10 & -- & -- & -- & 1/10 & 0.562 & 1.120 & 1.0 \\
0.5 & 10000 & TL & 0/10 & -- & -- & -- & 1/10 & 0.909 & 2.240 & 2.0 \\
0.5 & 10000 & VR & 0/10 & -- & -- & -- & 1/10 & 0.794 & 1.120 & 1.0 \\
0.7 & 1000 & TS & 10/10 & 0.030 & 0.112 & 1.0 & 10/10 & 0.029 & 0.112 & 1.0 \\
0.7 & 1000 & TL & 0/10 & -- & -- & -- & 7/10 & 0.041 & 0.224 & 2.0 \\
0.7 & 1000 & VR & 10/10 & 0.052 & 0.112 & 1.0 & 10/10 & 0.051 & 0.112 & 1.0 \\
0.7 & 4000 & TS & 10/10 & 0.198 & 0.448 & 1.0 & 10/10 & 0.189 & 0.448 & 1.0 \\
0.7 & 4000 & TL & 0/10 & -- & -- & -- & 2/10 & 0.383 & 2.016 & 4.5 \\
0.7 & 4000 & VR & 10/10 & 0.277 & 0.448 & 1.0 & 10/10 & 0.277 & 0.448 & 1.0 \\
0.7 & 10000 & TS & 10/10 & 0.503 & 1.120 & 1.0 & 10/10 & 0.503 & 1.120 & 1.0 \\
0.7 & 10000 & TL & 0/10 & -- & -- & -- & 5/10 & 0.688 & 3.360 & 3.0 \\
0.7 & 10000 & VR & 10/10 & 0.754 & 1.120 & 1.0 & 10/10 & 0.728 & 1.120 & 1.0 \\
0.9 & 1000 & TS & 10/10 & 0.027 & 0.112 & 1.0 & 10/10 & 0.027 & 0.112 & 1.0 \\
0.9 & 1000 & TL & 0/10 & -- & -- & -- & 10/10 & 0.028 & 0.112 & 1.0 \\
0.9 & 1000 & VR & 10/10 & 0.047 & 0.112 & 1.0 & 10/10 & 0.047 & 0.112 & 1.0 \\
0.9 & 4000 & TS & 10/10 & 0.157 & 0.448 & 1.0 & 10/10 & 0.157 & 0.448 & 1.0 \\
0.9 & 4000 & TL & 0/10 & -- & -- & -- & 9/10 & 0.211 & 0.896 & 2.0 \\
0.9 & 4000 & VR & 10/10 & 0.245 & 0.448 & 1.0 & 10/10 & 0.245 & 0.448 & 1.0 \\
0.9 & 10000 & TS & 10/10 & 0.467 & 1.120 & 1.0 & 10/10 & 0.459 & 1.120 & 1.0 \\
0.9 & 10000 & TL & 0/10 & -- & -- & -- & 9/10 & 0.562 & 2.240 & 2.0 \\
0.9 & 10000 & VR & 10/10 & 0.668 & 1.120 & 1.0 & 10/10 & 0.666 & 1.120 & 1.0 \\
\end{longtable}
\noindent TS = Trim-SCAD; TL = Trim-Lasso; VR = VRMOM-SCAD.\par
\endgroup

Timing includes initialization, gradients, pipe transport, aggregation and inner solves; setup and validation search are excluded. Attacked workers still compute gradients because the master replaces received messages. Payload is \mbox{$16(m-1)p$} bytes per round, excluding framing/control traffic; MB is $10^6$ bytes.

Every 20 ms, memory sampling sums coordinator and worker Windows working sets, excluding calibration. Its overhead is timed. Ancillary records give median/maximum peaks in MiB ($2^{20}$ bytes) and median solve time over all runs; sampling can miss brief peaks or count shared pages repeatedly.
\begin{center}\compacttable\setlength{\tabcolsep}{3pt}
\begin{tabular}{rccc}\toprule
$p$ & Trim-SCAD & Trim-Lasso & VRMOM-SCAD\\\midrule
1000 & 10/10 & 10/10 & 10/10 \\
4000 & 6/10 & 10/10 & 4/10 \\
10000 & 1/10 & 10/10 & 0/10 \\
\bottomrule\end{tabular}\end{center}
Entries are final successes at $q=0.5$ out of ten. At $q=0.7,0.9$, all methods succeed in every repetition. Complete timings and memory values for all 27 cells are in \path{ancillary_tables/process_memory.csv}.

At $q\ge0.7$, all 180 solves succeed. Final success measures optimization stability separately from reaching the error target. All measurements use one host.

\section{Path and Component Experiments}\label{supp-contents-11}\label{supp-sec:sampledpaths}
Six configurations use $m=8,p=768,s=4$, $b=m_{\rm B}=1$, ten independent seeds and eight rounds. Reference: $n_j=256,q=0.7$, Gaussian AR(1) correlation $0.3$ and Gaussian replacement. Variants use $n_j=512$, $q=0.5$, correlation $0.7$, or reports $-20$ times the honest mean gradient. These five designs use active positions $0,192,384,576$. A subsequently added clustered design uses positions $0,1,2,3$ and correlation $0.7$, otherwise retaining the reference; all original outcomes remain. Coefficients alternate $\pm1.2$; noise SD is $0.5$.

Variants share data, attack seeds, a complete validation sample of 400, and factors $0.3,0.6,1.2,2.4$. Full denotes practical Trim-SCAD. No correction uses zero-filled moments with $q=1$ and their natural penalty scale. Direct target retains each round's base and frozen weight shape but skips intermediate levels. Stage refresh recomputes the SCAD shape at each level, retaining the base and targets. Lasso uses constant weights; VRMOM changes only aggregation. Corrected variants share initial local Lasso.

Each stage permits 1200 updates with the original tolerances. Candidate and round budgets match, but continuation can use more stages; update counts include all four candidates and failures. Traced and core solvers agree to $10^{-10}$ for all four pilot candidates.

\small
\begin{center}\compacttable\setlength{\tabcolsep}{5pt}
\begin{tabular}{lccc}\toprule
Variant & Reference & Larger $n$ & More missing\\\midrule
Trim-SCAD & $0.085\pm0.009$ & $0.074\pm0.009$ & $0.670\pm0.121$ \\
No correction & $0.067\pm0.005$ & $0.061\pm0.006$ & $0.107\pm0.009$ \\
Direct target & $0.085\pm0.009$ & $0.074\pm0.009$ & $0.670\pm0.121$ \\
Stage refresh & $0.192\pm0.105$ & $0.074\pm0.009$ & $0.669\pm0.122$ \\
Trim-Lasso & $0.313\pm0.012$ & $0.224\pm0.014$ & $0.805\pm0.114$ \\
VRMOM-SCAD & $0.089\pm0.010$ & $0.076\pm0.008$ & $0.648\pm0.100$ \\
\bottomrule\end{tabular}\end{center}
\begin{center}\compacttable\setlength{\tabcolsep}{5pt}
\begin{tabular}{lccc}\toprule
Variant & Correlated & Sign flip & Clustered\\\midrule
Trim-SCAD & $0.089\pm0.009$ & $0.089\pm0.008$ & $1.098\pm0.285$ \\
No correction & $0.079\pm0.007$ & $0.084\pm0.007$ & $1.283\pm0.030$ \\
Direct target & $0.089\pm0.009$ & $0.089\pm0.008$ & $1.098\pm0.285$ \\
Stage refresh & $0.089\pm0.009$ & $0.086\pm0.009$ & $1.028\pm0.296$ \\
Trim-Lasso & $0.348\pm0.036$ & $0.345\pm0.010$ & $1.776\pm0.062$ \\
VRMOM-SCAD & $0.089\pm0.009$ & $0.085\pm0.008$ & $1.309\pm0.286$ \\
\bottomrule\end{tabular}\end{center}

Entries are coefficient error (mean $\pm$ SE); all cells have 10/10 successful
selected fits. \path{ancillary_tables/path_components.csv} retains candidate
failures, exact support, update counts and selections at the lower (Low) and
upper (High) grid endpoints.
\normalsize
The accompanying records retain all candidate outcomes and iteration paths; conditional accuracy excludes failed candidates.

The post-fit audit uses true support and honest moments to check each one-step inequality and propagated bound, reporting positive curvature and contraction separately. Each round's final weighted subproblem also undergoes the local optimization check.
\begin{center}\small\setlength{\tabcolsep}{6pt}\begin{tabular}{lrrrrr}\toprule
Setting & Positive & Contract & Local & Largest $\alpha_t$ & Final bound\\\midrule
Reference & 10/10 & 5/10 & 10/10 & 2.06 & 3.07 \\
Larger n & 10/10 & 10/10 & 10/10 & 0.96 & 0.75 \\
More missing & 10/10 & 6/10 & 10/10 & 2.38 & 12.31 \\
Correlated & 10/10 & 6/10 & 10/10 & 4.20 & 4.26 \\
Sign flip & 10/10 & 10/10 & 10/10 & 0.93 & 2.00 \\
Clustered & 10/10 & 0/10 & 10/10 & 23.68 & 395122.72 \\
\bottomrule\end{tabular}\end{center}
Positive, Contract and Local count runs meeting each condition in every round. Largest $\alpha_t$ maximizes over rounds and seeds; Final bound is the median propagated radius, including uninformative values. All 60 selections satisfy every one-step and propagated inequality across eight rounds, conditional on the simulated data and true support.

For all 60 master samples, a null vector of an $n_1\times(n_1+1)$ observed-design block has residual below $10^{-8}$ and a negative corrected quadratic form. The diagonal ratio lower-bounds $M/\kappa$, implying $4K_{\min}>n_1$ and excluding uniform curvature by Proposition~\ref{supp-prop:rank}, despite the positive endpoint curvature reported above.

The zero-filled population score is derived in Section~\ref{supp-sec:competitive}. This score is nearly zero for widely separated AR(1) active predictors, but not adjacent ones. Moment correction removes this population bias at increased variance. The clustered comparison remains uncertain: its ten-seed paired SE exceeds the mean improvement.

Paired differences are Full minus comparator on shared data and validation samples (mean $\pm$ Monte Carlo SE across ten independent seeds). Negative values favor Full.
\small
\par\begingroup
\renewcommand{\baselinestretch}{1}\fontsize{10}{11.5}\selectfont
\renewcommand{\arraystretch}{1.08}\setlength{\tabcolsep}{2pt}
\setlength{\tabcolsep}{1pt}
\begin{longtable}{lllll}
\toprule
Setting & No correction & Direct target & Stage refresh & VRMOM-SCAD\\
\midrule\endhead
\bottomrule\endlastfoot
Reference & $0.0183\pm0.0063$ & $0.0000\pm0.0000$ & $-0.1068\pm0.1068$ & $-0.0039\pm0.0040$ \\
Larger n & $0.0136\pm0.0058$ & $0.0000\pm0.0000$ & $-0.0000\pm0.0000$ & $-0.0014\pm0.0021$ \\
More missing & $0.5628\pm0.1235$ & $0.0000\pm0.0000$ & $0.0002\pm0.0001$ & $0.0214\pm0.0248$ \\
Correlated & $0.0101\pm0.0090$ & $0.0000\pm0.0000$ & $-0.0000\pm0.0000$ & $0.0001\pm0.0031$ \\
Sign flip & $0.0057\pm0.0121$ & $0.0000\pm0.0000$ & $0.0033\pm0.0078$ & $0.0045\pm0.0041$ \\
Clustered & $-0.1847\pm0.2817$ & $0.0000\pm0.0000$ & $0.0699\pm0.2836$ & $-0.2109\pm0.1818$ \\
\end{longtable}
\endgroup
\normalsize
Gaussian samples and masks are IID within and across honest nodes. All methods use dense empirical Gram matrices in this statistical audit; the separate matrix-free benchmark reports timing on one host.

\section{Node-Size Imbalance}\label{supp-contents-12}\label{supp-sec:imbalance}
\begingroup\renewcommand{\arraystretch}{0.93}
Set $m=8,p=240,q=0.7$ and 2400 assigned observations, with Gaussian AR(1) correlation $0.3$, four equally spaced alternating $\pm1.2$ coefficients and noise SD $0.5$. Geometric allocations target maximum/minimum ratios 1, 4 and 16; the table gives rounded sizes. The master receives the smallest or largest allocation, and the two smallest or largest remaining workers send Gaussian replacements ($m_{\rm B}=b=2$). Honest observations remain IID; honest totals vary with attack placement under the fixed assigned budget.

Ten seeds per allocation, three methods and four tuning factors give 270 selections and 1080 candidates. Complete validation ($n=400$) selects the converged fit with lowest prediction error, without truth or honest labels. Fits use eight rounds and a final local check. Methods share seeds within configurations; changed sizes alter the data stream.

\begin{center}\small
\begin{tabular}{rllrrrr}\toprule
Ratio & Master & Attacked & $n_1$ & $N_H$ & $N_{\rm eff}$ & Size range\\\midrule
1 & equal & equal & 300 & 1800 & 1800.0 & 300--300 \\
4 & small & small & 135 & 2034 & 1662.3 & 135--546 \\
4 & small & large & 135 & 1410 & 1258.3 & 135--546 \\
4 & large & small & 546 & 2100 & 1872.6 & 135--546 \\
4 & large & large & 546 & 1591 & 1299.6 & 135--546 \\
16 & small & small & 51 & 2211 & 1019.7 & 51--821 \\
16 & small & large & 51 & 1028 & 663.2 & 51--821 \\
16 & large & small & 821 & 2273 & 1466.6 & 51--821 \\
16 & large & large & 821 & 1478 & 681.8 & 51--821 \\
\bottomrule\end{tabular}\end{center}

\par\begingroup
\renewcommand{\baselinestretch}{1}\fontsize{10}{11.5}\selectfont
\renewcommand{\arraystretch}{1.08}\setlength{\tabcolsep}{2pt}
\begin{longtable}{rlllll}
\toprule
Ratio & Master & Attacked & Trim-SCAD & Trim-Lasso & Median-SCAD\\
\midrule\endhead
\bottomrule\endlastfoot
1 & equal & equal & $0.110\pm0.011$ & $0.314\pm0.021$ & $0.122\pm0.013$ \\
4 & small & small & $0.177\pm0.046$ & $0.644\pm0.156$ & $0.269\pm0.132$ \\
4 & small & large & $0.205\pm0.037$ & $0.606\pm0.141$ & $0.251\pm0.043$ \\
4 & large & small & $0.083\pm0.011$ & $0.274\pm0.023$ & $0.081\pm0.011$ \\
4 & large & large & $0.116\pm0.015$ & $0.361\pm0.025$ & $0.124\pm0.017$ \\
16 & small & small & $0.988\pm0.207$ & $1.230\pm0.141$ & $0.964\pm0.192$ \\
16 & small & large & $1.395\pm0.187$ & $1.239\pm0.149$ & $1.367\pm0.177$ \\
16 & large & small & $0.107\pm0.007$ & $0.307\pm0.015$ & $0.111\pm0.009$ \\
16 & large & large & $0.173\pm0.023$ & $0.450\pm0.020$ & $0.163\pm0.028$ \\
\end{longtable}
\noindent Entries are $\ell_2$ error means $\pm$ Monte Carlo SE.
For every method and allocation, successful fits and passing local checks both number $10/10$.
\par\endgroup

Each selection transmits $16(8-1)240\times8=215040$ bytes, excluding the four-candidate search and control traffic. Despite successful local checks, error increases sharply for the smallest master at ratio 16; local convergence need not imply accurate estimation.

\endgroup

\begin{figure}[!htbp]
\centering\includegraphics[width=\textwidth,height=.72\textheight,keepaspectratio]{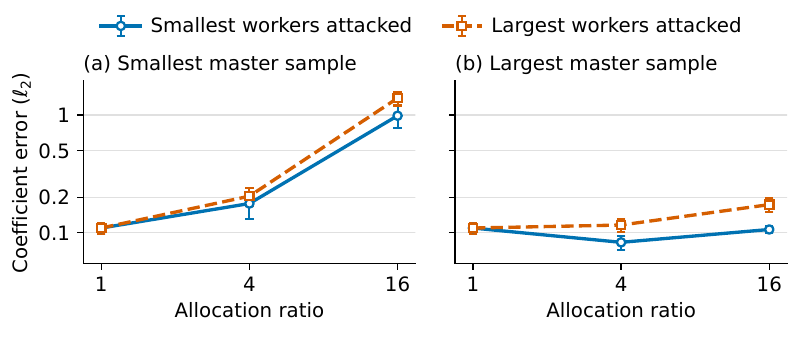}
\caption{Trim-SCAD coefficient error under unequal node sizes, with a total allocation of 2400 observations. The horizontal axis is the target largest-to-smallest allocation ratio; the error axis is logarithmic. Points show means and bars one standard error over ten repetitions. The balanced reference is shared, and all selected fits succeed; honest sample totals vary with attack placement.}
\label{supp-fig:mainimbalance}
\end{figure}

\section{Trusted-Sample Comparison}\label{supp-contents-13}\label{supp-sec:trustedcomparison}
Trusted-resource comparisons use $p=240,q=0.7$, 20 nodes with 120 observations each, two Gaussian attackers and $b=2$. Designs are independent Rademacher features with $s=1,4$, or signs of Gaussian AR(1) features with latent correlation $0.3$ and $s=4$. The first $s$ coefficients alternate $\pm1.2$; uniform $[-0.2,0.2]$ noise gives $H=1,Y_{\max}=1.2s+0.2$.

For ten seeds per design, four distributed scaled-weight SCAD candidates are fixed before drawing trusted samples. Prespecified samples of sizes $N_{\rm cal}=10000,50000$ are nested within seeds and independent of fitting and attacks, with $\delta=0.05$ and no data-dependent stopping. Selection minimizes the radius among successful fits with passing final local checks and positive calibrated curvature.

Trusted-only corrected Lasso and scaled-weight SCAD use the same trusted budget: $N_{\rm cal}-1000$ fitting rows and 1000 masked validation rows. Four shared multipliers minimize $x^\top\widetilde G_{\rm val}x-2\widetilde\gamma_{\rm val}^\top x$, whose population objective equals prediction error up to a constant. SCAD uses eight single-node practical rounds; truth evaluates errors only after selection. Distributed fitting additionally uses 2400 assigned observations.

\begin{center}\small
\begin{tabular}{rrrrrrr}\toprule
$s$ & Latent AR & $N_{\rm cal}$ & Admitted & Covered & Mean radius & Distributed L2\\\midrule
1 & 0.0 & 10000 & 10/10 & 10/10 & 0.479 & 0.0042 \\
1 & 0.0 & 50000 & 10/10 & 10/10 & 0.175 & 0.0042 \\
4 & 0.0 & 10000 & 0/10 & -- & -- & -- \\
4 & 0.0 & 50000 & 10/10 & 10/10 & 2.603 & 0.0698 \\
4 & 0.3 & 10000 & 0/10 & -- & -- & -- \\
4 & 0.3 & 50000 & 0/10 & -- & -- & -- \\
\bottomrule\end{tabular}\end{center}
\begin{center}\small
\begin{tabular}{rrrll}\toprule
$s$ & Latent AR & $N_{\rm cal}$ & Trusted Lasso L2 $\pm$ SE & Trusted SCAD L2 $\pm$ SE\\\midrule
1 & 0.0 & 10000 & $0.0317\pm0.0027$ & $0.0106\pm0.0063$ \\
1 & 0.0 & 50000 & $0.0133\pm0.0012$ & $0.0031\pm0.0020$ \\
4 & 0.0 & 10000 & $0.1020\pm0.0042$ & $0.0591\pm0.0101$ \\
4 & 0.0 & 50000 & $0.0447\pm0.0022$ & $0.0241\pm0.0039$ \\
4 & 0.3 & 10000 & $0.1059\pm0.0051$ & $0.0658\pm0.0146$ \\
4 & 0.3 & 50000 & $0.0503\pm0.0020$ & $0.0182\pm0.0030$ \\
\bottomrule\end{tabular}\end{center}
All trusted-only selections succeed (10/10 per cell). Distributed radii/errors condition on admission; empty cells indicate no certified selection, and every admitted candidate is covered. Four independent signals require 50000 calibration observations; the correlated design admits none at either size.

\begin{center}\small
\begin{tabular}{rrrrrr}\toprule
$s$ & Latent AR & $N_{\rm cal}$ & Moments (s) & Certificate (s) & Two baseline grids (s)\\\midrule
1 & 0.0 & 10000 & 0.018 & 0.001 & 0.029 \\
1 & 0.0 & 50000 & 0.092 & 0.001 & 0.034 \\
4 & 0.0 & 10000 & 0.018 & 0.001 & 0.026 \\
4 & 0.0 & 50000 & 0.094 & 0.001 & 0.030 \\
4 & 0.3 & 10000 & 0.019 & 0.001 & 0.031 \\
4 & 0.3 & 50000 & 0.092 & 0.001 & 0.035 \\
\bottomrule\end{tabular}\end{center}
Times are ten-seed medians on one host with one BLAS thread. Streamed moments use batches of at most 5000, excluding data generation. Certificate timing covers moment correction and four bounds, excluding accumulation; baseline timing covers both fitting/selection grids, excluding shared moments. Calibration uses a dense $240\times240$ double-precision Gram matrix: 460800 bytes (0.439 MiB), plus other arrays and process overhead. Distributed fitting is matrix-free.

\section{Model Extensions}\label{supp-contents-14}
For coordinate-dependent, independent MCAR masks with known $q_l$, the unbiased corrections of \citet[Example~3]{suppbib-loh2012} are $G_{ll}=Z_l^2/q_l$ and $G_{lk}=Z_lZ_k/(q_lq_k)$ for $l\ne k$, and $d_l=Z_lY/q_l$. The exact score variance becomes
\begin{equation*}
 \E\phi_l^2=\frac{\E X_l^2\varepsilon^2}{q_l}
 +\sum_{k\ne l}\frac{1-q_k}{q_lq_k}(\thstar_k)^2\E X_l^2X_k^2.
\end{equation*}
The deterministic path theorems transfer to this quadratic loss once appropriate concentration and curvature events are supplied. 

For mean-zero measurement error $Z=X+W$ independent of $X$, subtracting its known covariance $\Sigma_W$ from the empirical Gram gives an unbiased estimate of $\Sigma$~\cite[Example~1]{suppbib-loh2012}. With $\E W\varepsilon=0$, the empirical response moment is also unbiased. The remaining proofs require a new score-variance and tail analysis; the present MCAR variance formula is not reused. For complete-data GLMs, a nonlinear Hessian can vary along the segment between two sparse iterates. A uniform restricted Hessian condition can potentially replace the fixed quadratic condition, but controlling the surrogate score and whole homotopy domain requires additional analysis. For missing GLMs, inverse-probability substitution inside a nonlinear mean function need not be unbiased. For example, in a scalar Poisson model with mean $\exp(x\theta)$ and $Z=M_Xx$, conditioning on $X=x$ gives $\E\{\exp(Z\theta/q)\mid X=x\}=(1-q)+q\exp(x\theta/q)$.
For $0<q<1$ and $x\theta\ne0$, strict convexity makes this larger than $\exp(x\theta)$. Thus the linear-model correction cannot simply be inserted into this mean function.

\section{Reproducibility}\label{supp-contents-15}
The accompanying reproducibility archive contains editable manuscripts, analysis code, frozen result records, seeds, design records, environment metadata, and a SHA-256 manifest. The public data are redistributed with CC BY 4.0 attribution and a download script. The supplied scripts regenerate the result tables and figures. The ancillary CSV tables accompany this supplement; additional plotted views of the tabulated sensitivity, path and cost results remain in the reproducibility archive.

OpenAI Codex assisted with text, proof drafts, code, and figures.

\setlength{\bibsep}{2pt}

\vskip .65cm
\noindent Xiaofei Wu. Yunnan Key Laboratory of Statistical Modeling and Data Analysis, School of Mathematics and Statistics, Yunnan University, Kunming, China.
\par\noindent E-mail: \href{mailto:xfwu1016@ynu.edu.cn}{xfwu1016@ynu.edu.cn}
\par\medskip
\noindent Rongmei Liang. Department of Statistics and Data Science, Southern University of Science and Technology, Shenzhen 518055, China.
\par\noindent E-mail: \href{mailto:liang_r_m@163.com}{liang\_r\_m@163.com}
\par\medskip
\noindent Jian Qing Shi (corresponding author). Guangdong Provincial/Zhuhai Key Laboratory of Interdisciplinary Research and Application for Data Science, Beijing Normal-Hong Kong Baptist University, Zhuhai, 519087, China.
\par\noindent E-mail: \href{mailto:jianqingshi@bnbu.edu.cn}{jianqingshi@bnbu.edu.cn}

\end{document}